\documentclass[runningheads]{llncs}

\usepackage[american]{babel}
\usepackage[T1]{fontenc}
\usepackage[utf8]{inputenc}

\usepackage{lmodern}
\usepackage[final]{microtype}

\RequirePackage{xifthen}
\newboolean{isArXiv}
\setboolean{isArXiv}{true}

\RequirePackage[nolist]{acronym}
\newcommand{\acitalic}[1]{\textit{#1}}

\usepackage{wrapfig}

\usepackage{subcaption}

\usepackage{listings}
\usepackage[vlined,linesnumbered]{algorithm2e}

\usepackage{amsmath, amssymb, amsfonts, mathtools}

\usepackage{graphicx}
\graphicspath{{gfx/},{experiments/plots_runtime/},}

\usepackage[table,dvipsnames]{xcolor}

\usepackage{booktabs,longtable}

\usepackage[inline]{enumitem}

\usepackage[nospace]{cite}

\usepackage{xfrac,xspace}
\usepackage{relsize}
\usepackage[normalem]{ulem}

\usepackage{keyval}

\usepackage{pifont}

\usepackage{tikz,pgfplots,tikz-3dplot}
\pgfplotsset{compat=1.18}
\usetikzlibrary{arrows.meta,bending,positioning}
\usepackage{pifont,fontawesome}

\usepackage{setspace}

\newcommand{\myparagraph}[1]{\smallskip\noindent \emph{#1}}

\usepackage{comment}

\usepackage{soul}
\usepackage[textwidth=1.2\marginparwidth]{todonotes}

\usepackage{lipsum}

\usepackage[outline]{contour}
\contourlength{1.2pt}

\RequirePackage[%
  bookmarks,
  unicode,
  breaklinks=true,
  ]{hyperref}

\makeatletter
\renewcommand*\AC@verridelabel[1]{}
\makeatother

\usepackage[nameinlink]{cleveref}
\makeatletter
\AddToHook{cmd/appendix/before}{\def\cref@section@alias{appendix}\def\cref@subsection@alias{appendix}}
\AddToHook{cmd/appendix/before}{\crefalias{section}{appendix}}
\makeatother

\Crefname{equation}{eq.}{eqs.}
\crefname{equation}{equation}{equations}
\Crefname{figure}{Fig.}{Figs.}
\crefname{figure}{figure}{figures}
\Crefname{tabular}{Table}{Tables}
\crefname{tabular}{table}{tables}
\Crefname{definition}{Def.}{Defs.}
\crefname{definition}{definition}{definitions}
\Crefname{proposition}{Prop.}{Props.}
\crefname{proposition}{proposition}{propositions}
\Crefname{theorem}{Thm.}{Thms.}
\crefname{theorem}{theorem}{theorems}
\Crefname{section}{Sec.}{Sections}
\crefname{section}{section}{sections}
\Crefname{subsection}{Sec.}{Sections}
\crefname{subsection}{subsection}{subsections}
\crefname{algorithm}{algorithm}{algorithms}
\crefname{algorithm}{Alg.}{Algorithms}
\crefname{listing}{code}{code\ blocks}
\Crefname{equation}{Eq.}{Equations}
\Crefname{appendix}{App.}{Appendices}
\crefformat{footnote}{#2\footnotemark[#1]#3}

\newcommand{\splitatcommas}[1]{%
	\begingroup
	\begingroup\lccode`~=`, \lowercase{\endgroup
		\edef~{\mathchar\the\mathcode`, \penalty0 \noexpand\hspace{0pt plus 1em}}%
	}\mathcode`,="8000 #1%
	\endgroup
}

\makeatletter
\renewcommand{\paragraph}{\@startsection{paragraph}{5}{0em}%
  {.7ex plus .2ex minus .1ex}%
  {-.5em}%
  {\bfseries}}
\makeatother

\makeatletter
\def\orcidID#1{\smash{\mbox{\href{http://orcid.org/#1}{\protect\raisebox{-1.25pt}{\protect\includegraphics{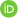}}}}}}
\makeatother

\allowdisplaybreaks[4]

\newcommand{\N}{\mathbb{N}}

\newcommand{\R}{\mathbb{R}}

\newcommand{\prop}{\ensuremath{\varphi}\xspace}
\newcommand{\est}[1][\prop]{\ensuremath{\hat{#1}}\xspace}

\newcommand{\eminus}[1]{\ensuremath{\scalebox{.85}{\text{\scshape e-}#1}}\xspace}

\newcommand{\xtr}[1]{\xrightarrow{\protect{\raisebox{-1pt}[0pt][0pt]{\ensuremath{\scriptstyle{#1}}}}}}

\DeclareMathAlphabet{\mathsc}{OT1}{cmr}{m}{sc}

\newcommand{\dist}[1]{\mathsc{\MakeLowercase{#1}}}
\mathchardef\mhyph="2D

\let\emptyset\varnothing

\newcommand{\hide}[1]{}

\colorlet{colcodekey}{PineGreen!60!green!70!black}
\colorlet{colcodeidentifier}{Sepia}
\colorlet{colcodenumeric}{blue!25!black}
\colorlet{colcodesyncaction}{violet!90}
\colorlet{colcodecomment}{Black!60}
\colorlet{coldistribution}{Cerulean!60!Black}
\lstdefinelanguage{Kepler}{
  keywords={},
  keywords=[2]{toplevel,and,or,pand,vot,fdep,spare,be,sbe,rbox,seq,csp,wsp,hsp},
  keywords=[3]{fail,repair,dorm,prio,fcfs,rand,lambda,prob,res,rate,mean},
  keywords=[4]{exponential,erlang,uniform,normal,lognormal,%
               weibull,gamma,rayleigh,fsteq,dirac,%
               exp,uni,dir},
  otherkeywords={},
  morecomment=**[l]{//},
  morecomment=**[s]{/*}{*/},
  moredelim=**[is][\color{colcodenumeric}]{^}{^},
  moredelim=**[is][\color{colcodekey}]{-|}{|-},
  moredelim=**[is][\slshape]{__}{__},
}
\lstdefinestyle{Kepler}{
  language=Kepler,
  xleftmargin=2em,
  basewidth=0.5em,
  basicstyle={\ttfamily},
  identifierstyle=\color{colcodeidentifier},
  keywordstyle=\bfseries,
  keywordstyle=[2]\color{colcodesyncaction},
  keywordstyle=[3]\bfseries\color{colcodekey},
  keywordstyle=[4]\bfseries\color{coldistribution},
  commentstyle=\color{colcodecomment},
  stringstyle=\mdseries,
  showstringspaces=false,
  numbers=left,
  numberstyle=\scriptsize\ttfamily\color{colcodecomment},
  numbersep=0.9em,
  tabsize=4,
  frame=none,
  aboveskip=\bigskipamount,
  belowskip=\medskipamount,
  abovecaptionskip=\smallskipamount,
  belowcaptionskip=\smallskipamount,
  captionpos=b,
  escapeinside={`}{`},
  mathescape=true,
}
\lstdefinelanguage{IOSA}{
  keywords={},
  keywords=[2]{dtmc,ctmc,mdp,module,endmodule,properties,endproperties,%
               init,label,formula},
  keywords=[3]{const,global,bool,int,double,clock,P,S,U},
  keywords=[4]{exponential,erlang,uniform,normal,lognormal,weibull,%
               gamma,rayleigh,fsteq},
  otherkeywords={->,:,[,],..,@,\?,\!},
  morecomment=**[l]{//},
  morecomment=**[s]{/*}{*/},
  morestring=[b]",
  moredelim=**[is][\color{colcodenumeric}]{^}{^},
  moredelim=**[is][\color{colcodesyncaction}]{~}{~},
  moredelim=**[is][\slshape]{__}{__},
}
\lstdefinestyle{IOSA}{
  language=IOSA,
  xleftmargin=\leftmargin,
  basewidth=0.5em,
  basicstyle={\ttfamily},
  identifierstyle=\color{colcodeidentifier},
  keywordstyle=\bfseries,
  keywordstyle=[2]\bfseries,
  keywordstyle=[3]\bfseries\color{colcodekey},
  keywordstyle=[4]\bfseries\color{coldistribution},
  commentstyle=\color{colcodecomment},
  stringstyle=\color{colcodesyncaction},
  showstringspaces=false,
  numbers=left,
  numberstyle=\scriptsize\ttfamily\color{colcodecomment},
  numbersep=0.9em,
  tabsize=4,
  frame=none,
  aboveskip=\bigskipamount,
  belowskip=\medskipamount,
  abovecaptionskip=.5ex,
  belowcaptionskip=\smallskipamount,
  captionpos=b,
  escapeinside={`}{`},
  mathescape=true,
}
\newlength{\MaxSizeOfLineNumbers}%
\makeatletter
\lst@AddToHook{OnEmptyLine}{\vspace{-0.75\baselineskip}}
\lst@Key{countblanklines}{true}[t]%
    {\lstKV@SetIf{#1}\lst@ifcountblanklines}
\lst@AddToHook{OnEmptyLine}{%
    \lst@ifnumberblanklines\else%
       \lst@ifcountblanklines\else%
         \advance\c@lstnumber-\@ne\relax%
       \fi%
    \fi}
\makeatother
\lstdefinestyle{Modest}{
  basicstyle=\lst@ifdisplaystyle\scriptsize\fi\ttfamily,
  columns=fullflexible,keepspaces,
  keywordstyle=\color{blue!67!black},
  emphstyle=[1]{\color{red!50!black}},
  emphstyle=[2]{\color{teal!67!black}},
  emphstyle=[3]{\color{yellow!55!black}},
  morecomment=[l][\color{green!33!black}\normalfont]{//},
  morekeywords={impatient,action,int,palt,alt,rate,do,when,observable,process,bool,property,timer,invariant,const,real,Pmax,Smax,option,false,stop},
  morestring=[b]",
  stringstyle=\color{red!50!black},
  emph=[1]{arrival,queue1_service,queue2_service,overflow},
  emph=[2]{q1,q2,t_arrivals,t_server_1,t_server_2,failed},
  emph=[3]{Exponential,Uniform},
  escapechar=§,
  numberstyle=\tiny\bfseries,
  numberblanklines=false,
  countblanklines=false,
  xleftmargin=\MaxSizeOfLineNumbers
}

\lstdefinelanguage{IFUNModest}{
  keywords={},
  keywords=[2]{max,if,then,else,},
  keywords=[3]{E,D,OD,TD,OGD,TGD},
  keywords=[4]{},
  otherkeywords={->,:,[,],..,@,\?,\!},
  morecomment=**[l]{//},
  morecomment=**[s]{/*}{*/},
  morestring=[b]",
  moredelim=**[is][\color{colcodenumeric}]{^}{^},
  moredelim=**[is][\color{colcodesyncaction}]{~}{~},
}
\lstdefinestyle{IFUNModest}{
  language=IFUNModest,
  xleftmargin=2em,
  basewidth=0.5em,
  basicstyle={\ttfamily},
  identifierstyle=\color{colcodeidentifier},
  keywordstyle=\bfseries,
  keywordstyle=[2]\color{colcodesyncaction},
  keywordstyle=[3]\bfseries\color{colcodekey},
  keywordstyle=[4]\bfseries\color{coldistribution},
  commentstyle=\color{colcodecomment},
  stringstyle=\mdseries,
  showstringspaces=false,
  numbers=left,
  numberstyle=\scriptsize\ttfamily\color{colcodecomment},
  numbersep=0.9em,
  tabsize=4,
  frame=none,
  aboveskip=\bigskipamount,
  belowskip=\medskipamount,
  abovecaptionskip=\smallskipamount,
  belowcaptionskip=\smallskipamount,
  captionpos=b,
  escapeinside={`}{`},
  mathescape=true,
}

\newcommand{\timers}{\mathcal{T}}
\newcommand{\actions}{\mathcal{A}}
\newcommand{\inactions}{\actions^{\mathsf i}}
\newcommand{\outactions}{\actions^{\mathsf o}}
\newcommand{\coactions}{\actions^{\mathsf u}}
\newcommand{\states}{\mathcal{S}}
\newcommand{\locations}{\mathcal{L}}
\newcommand{\trans}[1][]{\xrightarrow{{#1}}}

\colorlet{hlboxcoldraw}{black!65}
\colorlet{hlboxcolfill}{black!6}
\newsavebox\BODYBOX
\def\BODY{\unhbox\BODYBOX}
\NewDocumentEnvironment{hlbox}{ O{} }
  {\vspace{1ex plus .5ex minus .3ex}%
	\begingroup\centering\begin{spacing}{1.05}%
	\begin{lrbox}{\BODYBOX}\begin{minipage}{.96\linewidth}%
	\ifthenelse{\equal{#1}{}}{}{\textbf{#1}}}
  {\end{minipage}\end{lrbox}\begin{tikzpicture}%
	\node[rectangle,rounded corners=.3mm,inner sep=1.3ex,thick,
	      draw=hlboxcoldraw,fill=hlboxcolfill] {\BODY};
	\end{tikzpicture}\end{spacing}\endgroup%
	\vspace{1ex plus .5ex minus .3ex}}

\DeclareMathAlphabet{\pazocal}{OMS}{zplm}{m}{n}

\newcommand{\distparam}[2]{\ensuremath{{#1}\raisebox{1pt}{$\scriptstyle\,=\,$}{#2}}\xspace}

\newcommand{\gate}[1]{\ensuremath{\mathsf{\MakeUppercase{#1}}}\xspace}
\newcommand{\ANDgate}{\gate{and}}
\newcommand{\ORgate}{\gate{or}}
\newcommand{\VOTgate}{\gate{vot}}
\newcommand{\PANDgate}{\gate{pand}}
\newcommand{\SPAREgate}{\gate{spare}}

\newcommand{\BE}{\gate{be}}
\newcommand{\SBE}{\gate{sbe}}

\newcommand{\RBOX}{\gate{rbox}}

\def\modelm{\ensuremath{\mathfrak{m}}\xspace}
\newcommand{\modelDFTsmall}{\ensuremath{\modelm_{\mathlarger s}}\xspace}
\newcommand{\modelDFTlarge}{\ensuremath{\modelm_{\ell}}\xspace}
\newcommand{\modelQueue}{\ensuremath{\modelm_{Q}}\xspace}
\newcommand{\IFUNcase}{\ensuremath{\mathrm{\Phi}}\xspace}
\newcommand{\RESrun}[1][\mathfrak{a}]{\ensuremath{{#1}_\IFUNcase}\xspace}
\newcommand{\RESalgo}[2][]{\ensuremath{\textsf{\smaller[.5]#2}\ifthenelse{\equal{#1}{}}{}{_{#1}}}\xspace}
\newcommand{\tIFUN}[2]{\textsf{#1{\tt\_}#2}\xspace}
\newcommand{\tagnmono}{\tIFUN{tagn}{mono}}
\newcommand{\tagncomp}{\tIFUN{tagn}{comp}}
\newcommand{\tsenorder}{\tIFUN{tsen}{order}}
\newcommand{\tsentimed}{\tIFUN{tsen}{timed}}

\makeatletter
\define@key{modelparams}{m}{\def\PANDmodel@m{#1}}
\define@key{modelparams}{n}{\def\PANDmodel@n{#1}}
\setkeys{modelparams}{
  m=m,        
  n=n,    
}
\newcommand{\mPANDn}[1][]{%
  \begingroup
    \setkeys{modelparams}{#1}  
    \ensuremath{\PANDmodel@m\raisebox{1pt}{\texttt{\_}}\PANDgate{\PANDmodel@n}}%
    \endgroup
  \xspace
}

\makeatother

\newcommand{\IFUN}[1][]{\ensuremath{\pazocal{I}%
  \if\relax\detokenize{#1}\relax\else\left(#1\right)\fi}\xspace}

\newcommand{\cs}[1][]{\ensuremath{\mathsmaller{\pazocal{C\hspace{-1.5pt}S}}%
  \if\relax\detokenize{#1}\relax\else_\mathfrak{#1}\fi}\xspace}

\newcommand{\REL}[1][T]{\ensuremath{\mathrm{REL}%
  \if\relax\detokenize{#1}\relax\else_{#1}\fi}\xspace}
\newcommand{\UNREL}[1][T]{\ensuremath{\mathrm{UNREL}%
  \if\relax\detokenize{#1}\relax\else_{#1}\fi}\xspace}
\newcommand{\Tree}{\ensuremath{\triangle}\xspace}
\newcommand{\xbf}{\ensuremath{\boldsymbol{x}}\xspace}
\newcommand{\zbf}{\ensuremath{\boldsymbol{z}}\xspace}
\newcommand{\idx}[1][]{\ensuremath{\mathit{idx}%
  \if\relax\detokenize{#1}\relax\else(#1)\fi}\xspace}
\newcommand{\nodev}{\ensuremath{v}\xspace}
\newcommand{\nodew}{\ensuremath{w}\xspace}
\newcommand{\nodes}[1][\Tree]{\ensuremath{\mathit{nodes}%
  \if\relax\detokenize{#1}\relax\else(#1)\fi}\xspace}
\newcommand{\type}[1][]{\ensuremath{\mathit{type}%
  \if\relax\detokenize{#1}\relax\else(#1)\fi}\xspace}
\newcommand{\child}[1][]{\ensuremath{\mathit{chil}%
  \if\relax\detokenize{#1}\relax\else(#1)\fi}\xspace}

\newcommand{\safe}[1][]{\ifthenelse{\equal{#1}{}}{\mathrm{safe}}{\mathrm{safe\hspace{1pt}}(#1)}}

\SetAlgoCaptionLayout{raggedright}
\SetAlCapFnt{\small}
\SetAlCapNameFnt{\small}
\SetAlCapSkip{\abovecaptionskip\relax}
\SetEndCharOfAlgoLine{}
\SetArgSty{textrm}
\SetKwInOut{Input}{Input}\SetKwInOut{Output}{Output}
\SetCommentSty{textit}
\SetKw{Break}{break}
\SetKwProg{Type}{type}{}{end}
\SetKwProg{Function}{function}{}{end}
\SetKwFor{ForEach}{foreach}{do}{}
\SetKwFor{WhileTrue}{repeat}{}{end}
\SetKw{Continue}{continue}
\SetKwRepeat{DoWhile}{repeat}{while}
\SetKwProg{Fn}{function}{}{}
\crefname{algocf}{alg.}{algs.}
\Crefname{algocf}{Alg.}{Algs.}

\usepackage{pifont,fontawesome}
\newcommand{\xmarkred}{\mathchoice
  {\text{\smaller[2]\color{Red}\ding{55}}}%
  {\text{\smaller[2]\color{Red}\ding{55}}}%
  {\scalebox{0.7}{\text{\smaller[2]\color{Red}\ding{55}}}}%
  {\scalebox{0.5}{\text{\smaller[2]\color{Red}\ding{55}}}}%
}%
\newcommand{\wrench}{\mathchoice
  {\text{\smaller[3]\color{OliveGreen}\faWrench}}%
  {\text{\smaller[3]\color{OliveGreen}\faWrench}}%
  {\scalebox{0.7}{\text{\smaller[3]\color{OliveGreen}\faWrench}}}%
  {\scalebox{0.5}{\text{\smaller[3]\color{OliveGreen}\faWrench}}}%
}%
\newcommand{\stkout}[1]
{\ifmmode\text{\sout{\ensuremath{#1}}}\else\sout{#1}\fi}

\newcolumntype{C}[1]{>{\centering\let\newline\\\arraybackslash\hspace{0pt}}m{#1}}

\makeatletter%
\g@addto@macro\normalsize{%
  \setlength\abovedisplayskip{3pt}%
  \setlength\belowdisplayskip{3pt}%
  \setlength\abovedisplayshortskip{-3pt}%
  \setlength\belowdisplayshortskip{3pt}%
}%
\makeatother

\begin{document}
\title{A Taxonomy of Distance Metrics for Time-Sensitive Importance Splitting} %
\subtitle{Timer Bounds, Resampling, and the Global Age%
\thanks{%
This work was partially funded by
DFG grant 389792660 as part of \href{https://perspicuous-computing.science}{TRR~248 CPEC},
the European Union (EU) under the INTERREG North Sea project STORM\_SAFE of the European Regional Development Fund,
and the EU's Horizon 2020 research and innovation programme under Marie Skłodowska-Curie grant agreement 101008233 (MISSION).
}}
\titlerunning{A Taxonomy of Distance Metrics for Time-Sensitive ISPLIT}

\author{Gabriel~Dengler\inst{1}\,\orcidID{0000-0002-4217-4952} \and
Carlos E.~Budde\inst{2}\,\orcidID{0000-0001-8807-1548}\and
Laura~Carnevali\inst{3}\,\orcidID{0000-0002-5896-4860}}
\authorrunning{G. Dengler et al.}
\institute{%
Saarland University, Saarbrücken, Germany
\\\email{dengler@depend.uni-saarland.de}
\and
Technical University of Denmark, Lyngby, Denmark
\and
Department of Information Engineering, University of Florence, Florence, Italy
}
\maketitle

\begin{abstract}
\Ac{ISPLIT} evaluates the probabilities of rare events in non-Markovian models. It requires a heuristic \ac{IFUN} that estimates the distance to the target.
While including timer evaluations in the \ac{IFUN} can substantially improve the effectiveness of \ac{ISPLIT},
the existing \emph{time-sensitive} \acp{IFUN} evaluate simulation states with respect to single sampled timer values. Thus, reaching highly important states requires simultaneously sampling specific combinations of timer values, yielding many unproductive simulation runs.

In this paper, we revisit time-sensitive \ac{ISPLIT} with the goal of steering simulation runs towards important states. First, we study how timer values can be \emph{resampled} conditioned on the elapsed time. The importance can be evaluated by considering the set of feasible timer values, decoupling importance estimation from timer samples.
Second, we exploit the \emph{global age} of a simulation to identify and prune the executions that can no longer reach the target within the remaining time budget.
Together, these ideas lead to a \textit{taxonomy} of distance metrics clarifying the role of timer bounds, resampling, and the global age.
In particular, for models with unbounded timers, we show that time-sensitive \acp{IFUN} collapse to ordinary \acp{IFUN} under resampling.
Experiments demonstrate that the proposed formulations substantially improve the accuracy of \ac{ISPLIT} estimators.

\end{abstract}

\setcounter{footnote}{0}

\acresetall

\section{Introduction}
\label{sec:introduction}

Reliability engineering provides methods and tools to assess risks in a variety of application domains \cite{Rob00,ZZZ21},
performing quantitative evaluation of probabilistic models, e.g., the failure behavior of \emph{fault trees}~\cite{ruijters2015fault}.
If any stochastic temporal parameter is non-exponentially distributed, the underlying stochastic process is non-Markovian~\cite{ciardo1994characterization}, and the model is termed a \emph{non-Markovian model}.

\myparagraph{Analysis and simulation methods.} 
Providing efficient solution techniques for the quantitative evaluation of non-Markovian models is considered significantly harder than for Markovian models.
Efficient quantitative evaluation methods do exist for restricted model classes, such as
non-Markovian models with an underlying semi-Markov process~\cite{carnevali2021compositional},
at most one non-exponential timer in each state~\cite{german1995transient,amparore2013component},
and an underlying Markov regenerative process with a bounded number of steps between regenerations~\cite{horvath:2012:transient-sscs}.
These restrictions, however, limit the model expressiveness and do not fully address scalability challenges \cite{BDMS22, dengler:2024:transient-eval-sscs-sim, dengler:2025:time-sensitive-isplit}.
Thus, \emph{simulation}, also known as \ac{SMC}, remains the default approach for analyzing large non-Markovian models.
In contrast to state-space enumeration methods, \ac{SMC} \emph{estimates} intervals $\est\pm\varepsilon$ that in $(1-\delta) \cdot 100\,\%$ of the cases contain the true probability sought \prop~\cite{BHMWW25}.
This is applicable to models of arbitrary size, but it struggles to reliably estimate \emph{rare events} \cite{lecuyer:2010:asymptotic-robustness-res,glasserman:1999:multilevel-splitting,ZBC12}.

\Ac{ISPLIT} is a prominent technique for reducing the variance of rare event estimations by layering (splitting) the state space so that reaching one layer from another one is not rare. 
The selection of layers is given by a heuristic \ac{IFUN}, whose choice is the main deciding factor for the efficiency of \ac{ISPLIT}. For a wide range of application scenarios, it has been demonstrated how an effective \ac{IFUN} can be derived automatically \cite{BDH19, BDMS22}.

\myparagraph{Time-sensitive \ac{ISPLIT}.}
A notable challenge in non-Markovian models is that the progress towards the target depends both on the system's \emph{discrete} logical state and on the \emph{continuous} remaining times of the active timers.
In~\cite{dengler:2025:time-sensitive-isplit}, we addressed this challenge by incorporating concrete timer values into automatically generated \acp{IFUN}, increasing the amount of information available to
\ac{ISPLIT} while tying importance to a particular sampled realization of timer values.
\begin{wrapfigure}[11]{r}{0.321\textwidth}
	\vspace{-1.5ex}
	\centering
	\includegraphics[width=.9\linewidth]{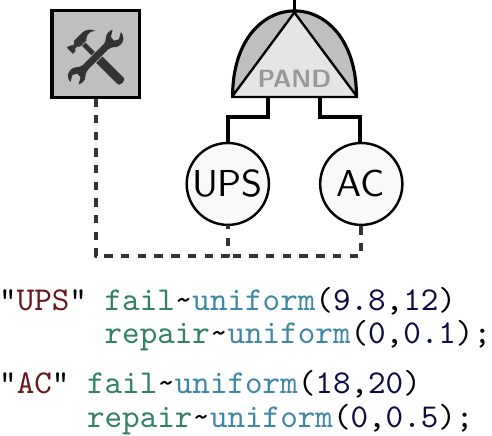}
	
	\vspace{-.5ex}
	\caption{Toy example}
	\label{fig:running-example}
\end{wrapfigure}

\vspace{-1em}
\begin{example}
\label{ex:running-example}
    \Cref{fig:running-example} shows a repairable \ac{DFT} from \cite{dengler:2025:time-sensitive-isplit}. It is made of two \acp{BE}: \textsf{UPS} (\underline{u}ninterruptible \underline{p}ower \underline{s}upply) and \textsf{AC} (\underline{a}lternating \underline{c}urrent grid). The system failure is governed by a \textsf{PAND} gate, i.e., it occurs only if the attached \acp{BE} fail from left to right.
    Given the failure and repair time (uniform distribution) bounds of the \acp{BE}, the shortest path to failure occurs when the \textsf{UPS} fails twice near the lower bound $9.8$, while the \textsf{AC} has to fail once at the upper bound $20$.
    So, an \ac{IFUN} should assign a state that satisfies this condition a higher importance than other states, which can only be done by considering concrete timer values.
\end{example}

Thus, two logical states may display quite different importance values based on the timer values.
Reaching high-importance states might occur rarely, as most of the probability mass is concentrated in realizations with low importance.
This raises the question: How can we observe more important simulation runs?

\paragraph{Contributions.} 
In this paper, we tackle the above problem
by
\begin{enumerate*}[label=(\roman*),itemjoin={{, }},itemjoin*={{ and by }}]
\item	\textit{resampling} of timers conditioned on their elapsed time
\item	exploiting the \textit{global age} 
\end{enumerate*}.
These variations lead to a plethora of distance metrics, whose role regarding timer bounds, resampling, and the global age can be categorized in a \emph{taxonomy}.

\myparagraph{Resampling.} 
The usual simulation semantics of non-Markovian models samples timer values when they are \emph{newly enabled}.
When a timer elapses, the values of the other persistent timers are reduced by the elapsed time.
Notably, there is no need to fix the value of a timer upon enabling. Conditioned on each timer's elapsed time, the future behavior is independent of the past samples and captures the reachable timing choices that remain available.
We refer to sampling from this conditioned distribution as \emph{resampling}.
Similar concepts exist for approximating distributions from a finite number of samples \cite{BS93,BBK09}.
Resampling enables us to enlarge the number of different timer samples that reach a high importance. 
Specifically, for evaluating the time-sensitive \ac{IFUN}, we consider the maximum importance that is reachable given the possible timer evaluations. 

We take a closer look at the special case of models where each timer has an infinite domain $[0, \infty)$, which equates to just considering in which order the active timers elapse. There, we show that \acp{IFUN} based on backwards reachability search behave equally in the time-agnostic as well as time-sensitive setting when resorting to resampling. For a simulation run where no resampling is used, we develop a more efficient data structure
to compute the time-sensitive \ac{IFUN}.

\myparagraph{Global age.} 
For time-bounded reachability properties, we exploit the \emph{global age}, i.e., the elapsed simulation time, to decide if the target condition can be reached \emph{within the remaining time budget}.
By integrating the global time as an artificial timer, the simulation can dynamically detect when the target condition is no longer reachable and immediately terminate unfeasible realizations.

\paragraph{Structure.}
We start by reviewing relevant background in \Cref{sec:background}. Then, we focus on the main contributions in the paper, i.e., resampling of timers (\Cref{sec:resampling-timers}) and inclusion of the global age (\Cref{sec:global-age}). The plethora of derived distance metrics is categorized in a taxonomy in \Cref{sec:taxonomy}.
We evaluate our variants for challenging \acp{DFT} and a queueing network in \Cref{sec:experiments}, and draw our conclusions in \Cref{sec:conclusion}.

\section{Background}
\label{sec:background}

\subsection{Input/Output Stochastic Automata}
\label{sec:background:IOSA}

We specify every model as a network of \acp{IOSA}.
Intuitively, these are automata with stochastic timers that can communicate with input/output actions, where urgent actions operate instantaneously:

\begin{definition}[\ac{IOSA}~\cite{DM18}]\label{def:iosac}
    An \emph{\acl{IOSA} with urgency} is a tuple $(\splitatcommas{\locations,\actions,\timers,\trans,T_0,l_0})$ consisting of a denumerable set of \emph{locations} $\locations$ with $l_0\in\locations$ being the \emph{initial location};
    a denumerable set of \emph{actions} $\actions$ partitioned into two sets of \emph{input} $\inactions$ and \emph{output} $\outactions$ actions, where $\coactions\subseteq\actions$ denotes the set of \emph{urgent} actions;
    a finite set of \emph{timers} $\timers$ s.t.\ each $t\in\timers$ has an associated continuous probability measure $\mu_t$ supported on $\R_{\ge0}$, where $T_0\subseteq\timers$ denotes the set of \emph{initial} timers;
    and
    a transition function ${\trans}\subseteq\locations\times2^\timers\times\actions\times2^\timers\times\locations$
    whose elements $\langle l, T, a, T', l' \rangle$ are written as $l\xtr{\scriptscriptstyle T,a,T'}l'$.
\end{definition}

Following \cite{dengler:2025:time-sensitive-isplit}, and in contrast to the original definition of \ac{IOSA} from \cite{DM18}, here we explicitly distinguish between \emph{locations} and (simulation) \emph{states}:
\begin{definition}[State]\label{def:state}
    A \emph{state} is a pair $\langle l, \tau\rangle \in \states$ with $\states=\mathcal{L} \times (\mathbb{R}_{\geq 0} \mathbin{\dot{\cup}}\{ \bot \})^{|\mathcal{T}|}$, made of a location~$l$ and a \emph{timer value function} $\tau$ yielding a valuation for each timer. We term a timer $t\in \mathcal{T}$ \emph{active} if $\tau(t) \neq \bot$ and \emph{inactive} otherwise.
\end{definition}

We simulate on the weakly-deterministic result of a full parallel composition of \ac{IOSA}, where eventual non-deterministic choices are spurious and thus amenable to \ac{SMC}---formal definitions and technical details are in \cite{DM18,MBD20}.
Informally speaking, the dynamics of such models are governed by transitions $l\xtr{\scriptscriptstyle \{t\},a,T'}l'$ where $a\in \outactions$\!.
This represents a system standing in location $l$ (e.g., where \textsf{AC} is operational in \Cref{ex:running-example}) for at most $0<\tau(t)$ time units.
When this time elapses, event $a$ (e.g., \textsf{AC} fails) takes the system to a new location $l'$ (where \textsf{AC} has failed), and activates new future behavior according to the timers in $T'$ (e.g., a repair time $0<\tau(t')$ is sampled for \textsf{AC} according to $\mu_{t'}$, given that $t'\in T'$).\!%
\footnote{%
	Note that, although not part of the original \ac{IOSA} definition, active timers can also be \emph{disabled}, meaning that they are removed from the set of active timers without elapsing.
	This happens, e.g., in the semantic mapping of \textsf{SPARE} gates to \ac{IOSA} (see \cite{MBD20}) when a \ac{SBE} is (de-)activated by the \textsf{SPARE} gate.
}

We require that exactly one subset of timers $T_l \subseteq \mathcal{T}$ is active in each location $l$~\cite{MBD20,DM18}.
Thus, given a location $l$, for all associated states $s = \langle l, \tau\rangle$ and timers $t\in \mathcal{T}$, it holds that $\tau(t) = \bot$ iff $t\not\in T_l$, viz., when a timer is not active in the current state, it takes the special value $\bot$ that denotes an inactive timer.
We assume that for each timer $t\in\mathcal{T}$ we can provide its bounds $[a,b] \subseteq \mathbb{Q}_{\geq 0}\cup \{\infty\}$:
\begin{definition}[Timer bounds]
	When enumerating the $n = |T_l|$ active timers of a location $l$ as $t_1, \ldots, t_n \in T_l$, we denote with $a_i\in \mathbb{Q}_{\geq 0}$ the lower bound and with $b_i\in \mathbb{Q}_{\geq 0}\cup \{\infty\}$ the upper bound of timer $t_i$ for $i \in \{1, \ldots, n\}$.
\end{definition}

\subsection{State Classes}
\label{sec:background:state-classes}
To determine the minimum number of transitions needed to reach the target given the current timer values, the work in \cite{dengler:2025:time-sensitive-isplit} used \acp{SC}~\cite{vicario:2001:tpn-analysis}:
\begin{definition}[\Acl{SC}]\label{def:tsclass}
	An \emph{\ac{SC}} $\Sigma = \langle l, D \rangle$ consists of a location $l \in \mathcal{L}$ and a joint domain~$D$ for the active timers in~$l$, i.e.,~the timers in $T_l$.
\end{definition}
The domain $D$ of an \ac{SC} $\Sigma = \langle l, D \rangle$ collects the (infinite) possible timer values. It can be efficiently encoded as a \ac{DBM} zone, i.e., the solution of a set of linear inequalities constraining the difference between two timers (with $t_\star$ being a fictitious timer denoting the time when $\Sigma$ is entered):
\begin{align*}
    D = \{ \tau(t_i) -\tau(t_j) \leq b_{ij} \quad \forall  \, t_i, t_j \in T_{l} \cup\{t_{\star}\} \mbox{ with } t_i \neq t_j\}
    \label{eq:DBM}
\end{align*}

To compute a distance metric that counts the number of transitions to reach the target, we perform backwards reachability search from the target locations.
For each \ac{SC} $\Sigma' = \langle l', D'\rangle$ and incoming transition $l\xtr{\scriptscriptstyle T,a,T'}l'$, we compute the predecessor \ac{SC} $\Sigma = \langle l, D\rangle$ and derive its distance metric $\omega_T(\Sigma)$ as the minimum number of transitions needed to reach the target locations from $\Sigma$. The predecessor calculations effectively resort to computing the \emph{weakest precondition}, i.e., all states of $\langle l, D\rangle$ end up in $\langle l', D'\rangle$ after taking the transition.
Details for the calculation can be found in
\ifthenelse{\boolean{isArXiv}}{%
	\cite{dengler:2025:time-sensitive-isplit} and \Cref{sec:app:reachability-state-classes}%
}{
	\cite{dengler:2025:time-sensitive-isplit} and \cite[App.\;A]{arXivVersion}%
}.
During this backwards search we can eliminate redundant information,
e.g., if for one location $l$ we have two \acp{SC} $\Sigma_1=\langle l, D_1\rangle$ and $\Sigma_2=\langle l, D_2\rangle$ such that $D_1 \supseteq D_2$ and $\omega_T(\Sigma_1) \leq \omega_T(\Sigma_2)$, so $\Sigma_1$ strictly includes $\Sigma_2$, we can omit $\Sigma_2$.
Analogously, by performing forwards search from the initial state with \acp{SC}, we can determine via \acp{SC} the set of states (locations and timer values) that are reachable during a simulation run \cite{vicario:2001:tpn-analysis}.

We can evaluate the distance (i.e., the minimum number of steps to reach the target given the timer constraints) of a state~$s=\langle l,\tau \rangle$ as the minimum distance to the target among those of the SCs which $s$ belongs to, i.e., $d_T \colon \mathcal{S} \to \mathbb{N}_0 \cup \{\infty\}$ with $d_T(s) = \min_{\Sigma \,\mid\, s \in \Sigma} \omega_T(\Sigma)$ ($d_T$ is an acronym for \underline{d}istance \underline{t}imed) and $d_T(s) = \infty$ if $s$ is not contained in any \ac{SC} $\Sigma$.
A state $s=\langle l,\tau \rangle$ belongs to an \ac{SC} $\Sigma = \langle l',D \rangle$ (which we write as $s \in \Sigma$) if $l=l'$ and the timer values defined by the timer value function $\tau$ satisfy $D$ (denoted with $\tau \in D$), which can be easily checked by evaluating all inequalities encoded in $D$ (runtime $\mathcal{O}(|T_l|^2)$).

By laying a probability density over domain $D$, we can quantify the probability of timer evaluations through \acp{SSC} \cite{vicario2009using, horvath:2012:transient-sscs}:
\begin{definition}[\Aclp{SSC}]
	An \emph{\ac{SSC}} is a tuple $\Sigma = \langle l,D,f\rangle$ where $l$, $D$ are defined as for \acp{SC} and $f \colon (\mathbb{R}_{\geq 0})^{|T_l|} \to \mathbb{R}_{\geq 0}$ is the \ac{PDF} (immediately after the previous transition) of the random vector of the times-to-fire $\tau$ of timers $t_1, \ldots, t_n \in T_l$ enabled in $l$.
\end{definition}
Similarly to \acp{SC}, we can enumerate the set of reachable \acp{SSC} from the initial state by a forwards search. This, for example, allows us to compute time-bounded reachability probabilities, either exactly, provided that some fair conditions on the model structure are satisfied (e.g., no cycles of timers elapsing in zero time), or assuming a truncation error in the enumeration \cite{horvath:2012:transient-sscs}. However, as the number of concurrent non-exponential timers increases, the evaluation comes with significant costs due to the state-space explosion and complexity of representing \acp{SSC}.

\subsection{Rare Event Simulation}
\label{sec:background:RES}

\myparagraph{Statistical model checking.}
In \Ac{SMC}, evaluating bounded reachability metrics boils down to the sampling of simulation traces, i.e., a sequence $x_i = \{\langle l_j, \tau_j\rangle\}_{j=0}^{n_i}$ of states with locations $l_j$ and timer values $\tau_j$, so that $x_i$ either visits a set of goal locations $\mathcal{L}_{g}$ or not \cite{YS02}.
An unbiased estimator for the probability of reaching $\mathcal{L}_{g}$ is given by $y_i = \mathop{\mathsf{\underline{if}}} \exists j.\, l_j \in \mathcal{L}_g \mathop{\mathsf{\underline{then}}} 1 \mathop{\mathsf{\underline{else}}} 0$.
For robustness, such point estimates are accompanied by an (asymptotic) statistical correctness guarantee, typically a \ac{CI} $\est\pm\varepsilon=[l, u] \subseteq \mathbb{R}$ with confidence level $\gamma = 1 - \delta$.

\myparagraph{Importance splitting.}
For \emph{rare events}, however, regular \ac{SMC} needs an unfeasible number of samples to provide a tight \ac{CI}.
This problem can be alleviated by using \ac{RES} methods, which can be categorized as \ac{IS} and \ac{ISPLIT}.
\ac{IS} modifies the underlying probability model such that the event of interest is more likely.
While sophisticated automated methods exist for Markovian models, applying them to non-Markovian models is seen as significantly more challenging~\cite{RBSH13,buijsrogge:2019:importance-sampling-tandem-queue}.
Instead, we focus in this work on \ac{ISPLIT} methods, which partition the state space into layers surrounding the goal states $\states_g\subsetneq\states$.
This is done using an \ac{IFUN} $f\colon\states\to\N_0 \cup \{-\infty\}$ that maps each state to an \emph{importance value} (or indicates that the target is unreachable by the value $-\infty$).
For \ac{ISPLIT}, there exist
multiple variations~\cite{Gar00, BDH19}, predominantly \emph{Fixed Effort} and \emph{RESTART} \cite{LLLT09,villen-altamirano:1991:RESTART}. In short, Fixed Effort iteratively starts from the selection of simulation runs that reached the next importance level, while RESTART merges split simulation runs when falling below a certain importance level.

\myparagraph{Importance function derivation.}
A general automated approach to derive an \ac{IFUN} uses backwards reachability search, i.e., a function $d \colon \mathcal{S} \to \N_0 \cup \{\infty\}$ with $d(\langle l, \tau\rangle) = \omega(l)$ that counts the number of needed transitions to the target, arriving at the \ac{IFUN} $f(s) \doteq \max_{s'\in \mathcal{S}, d(s')<\infty} \{d(s')\} - d(s)$ \cite{BDH19}.
This concept straightforwardly translates to a time-sensitive setting using $d_T \colon \mathcal{S} \rightarrow \mathbb{N}_0 \cup \{\infty\}$ (from \Cref{sec:background:state-classes}) to arrive at the \ac{IFUN} $f_T(s) \doteq \max_{s' \in \mathcal{S}, d_T(s')<\infty} \{d_T(s')\} - d_T(s)$.

\myparagraph{Threshold selection.} The practical efficiency of \ac{ISPLIT} strongly depends on the choice of importance levels (called thresholds) at which simulation runs are split. Too few thresholds lead to extinction (i.e., most runs fail to reach the target), whereas too many cause a significant performance overhead. Most methods for automatically computing thresholds~\cite{cerou:2012:sequential-mc, BDH19} start from the initial location and estimate level-up probabilities to determine the thresholds to include.

\section{Resampling of Timers}
\label{sec:resampling-timers}

This section is concerned with the first part of our main contribution, i.e., studying how the possibility of resampling timers influences the idea of time-sensitive \ac{ISPLIT}. After explaining the general concept, we consider models with unbounded timers in \Cref{sec:resampling-timers:models-infinite-domains} and the adaptations for threshold selection in \Cref{sec:resampling-timers:threshold-selection}.

As we can newly determine the value of each timer under the condition that a deterministic time has elapsed, we do not consider the currently sampled timer values, but the range of timer values that is possible given the elapsed time:
\begin{definition}[Elapsed time]
    \label{def:elapsed-time}
	For a state $s = \langle l, \tau\rangle \in \states$, $\tau_{\textit{elp}} \in (\mathbb{R}_{\geq 0}\mathbin{\dot{\cup}}\{ \bot \})^{|\mathcal{T}|}$ denotes the \emph{elapsed time} for each timer. Furthermore, we write $s_{\textit{elp}} = \langle l, \tau_{\textit{elp}}\rangle$ for a location with elapsed time information.
\end{definition}

We determine $\tau_\textit{elp}$ for each timer during a (forwards) simulation run as follows: In the initial state, no time has elapsed, so we set $\tau_\textit{elp}(t) = 0$ for all $t\in T_0$. Otherwise, consider a location $l$ with elapsed time $\tau_\textit{elp}$, concrete timer values $\tau$, and outgoing transition $l\xtr{\scriptscriptstyle T,a,T'}l'$ with $T = \{t\}$.
Then, the elapsed time of taking this transition is given by the value of the elapsed timer, i.e., $\Delta = \tau(t)$. For the successor state with elapsed timer values $s'_{\textit{elp}} = \langle l', \tau'_\textit{elp}\rangle$ it holds for all (newly enabled) timers $t_i\in T' \colon \tau'_\textit{elp}(t_i) = 0$ and for all (persistent) timers $t_i \in T_{l'} \backslash T' \colon \tau'_\textit{elp}(t_i) = \tau_\textit{elp}(t_i) + \Delta$.

Given the elapsed time, we can compute the lower bound of each timer as $\tau_{\textit{min}} \in (\mathbb{R}_{\geq 0} \mathbin{\dot{\cup}}\{ \bot \})^{|\mathcal{T}|}$ with $\tau_{\textit{min}}(t_i) = \max(0, a_i - \tau_{\textit{elp}}(t_i))$ for $t_i \in T_l$, and the upper bound $\tau_{\textit{max}} \in (\mathbb{R}_{\geq 0} \mathbin{\dot{\cup}}\{ \infty, \bot \})^{|\mathcal{T}|}$ with $\tau_{\textit{max}}(t_i) = b_i - \tau_{\textit{elp}}(t_i)$.\footnote{Note that it holds $\tau(t_i) \leq b_i$, so $b_i - \tau_\textit{elp}(t_i)$ is always non-negative.}

We define the timed distance metric under resampling as the (``regular'') minimum timed distance achievable by resampling the timers of the current state.
This corresponds to the highest regular importance achievable by resampling.
To formally define this metric, we observe that, depending on the value of $\tau_\textit{elp}$, $s_\textit{elp}$ induces an \ac{SC} containing the potential timer values reachable by resampling:
\begin{definition}[Induced state class]
	\label{def:induced-sc}
	The \emph{induced \ac{SC}} of a state $s = \langle l, \tau\rangle$ is the \ac{SC} $\Sigma_{s_\textit{elp}} = \langle l, D\rangle$ with domain $D = \bigtimes_{t_i \in T_{l}} [\tau_{\textit{min}}(t_i), \tau_{\textit{max}}(t_i)]$.
\end{definition}
As for the initial \ac{SC} in the \ac{SC} graph for forwards reachability analysis \cite{vicario2009using, horvath:2012:transient-sscs,vicario:2001:tpn-analysis}, the domain $D$ has a hyper-rectangular shape as the individual timers are solely shifted by a time constant and thus remain independent.
For the formal definition of the timed distance metric under resampling, we consider the minimum distance attainable by any state that is possible by resampling the current timers:
\begin{definition}[Time-sensitive distance with resampling]
    \label{def:time-sensitive-dist-resampling}
	Consider a state $s = \langle l, \tau\rangle$ with elapsed time information $s_{\textit{elp}} = \langle l, \tau_{\textit{elp}}\rangle$. Then, the \emph{timed distance metric under resampling} is defined as $d^{r}_{T}(s) = \min_{s'\in \Sigma_{s_\textit{elp}}} d_{T}(s')$.
\end{definition}
Since $d_{T}(s)$ is evaluated by iterating over the \acp{SC} associated with location~$l$, we can efficiently evaluate $d^{r}_{T}(s)$ in a similar vein: We determine the \ac{SC} with the smallest distance to the target that intersects with the induced \ac{SC} $\Sigma_{s_\textit{elp}} = \langle l, D\rangle$, formally $d^{r}_{T}(s) = \min_{\Sigma = \langle l', D'\rangle \,\mid\, l=l'\wedge D\cap D' \neq \emptyset} \omega_T(\Sigma)$ (see \Cref{sec:background:state-classes}).
Note that checking if a point lies inside the domain spanned by a \ac{DBM} zone has the same computational complexity as checking if the domains of two \ac{DBM} zones with equal size intersect \cite{vicario:2001:tpn-analysis, david:2005:DBM-merging}, namely for $n$ timers a runtime of $\mathcal{O}(n^2)$.

Moreover, an orthogonal perspective can be achieved by considering the probability densities associated with each timer, as $s_\textit{elp}$ also induces an \ac{SSC}:
\begin{definition}[Induced stochastic state class]
	\label{def:induced-ssc}
	Consider $s_\textit{elp} = \langle l, \tau_\textit{elp}\rangle$.
	For a timer $t_i \in T_l$, let $f_i \colon [a_i, b_i] \to \mathbb{R}_{\geq 0}$ be its \ac{PDF} for the corresponding probability measure $\mu_{t_i}$. We use $x_i$ for the variable describing the value of timer $t_i$, i.e., $\tau(t_i)$. Then, the \emph{induced \ac{SSC}} is the \ac{SSC} $\langle l, D, f\rangle$ with $D$ defined as in \Cref{def:induced-sc} and
	\[ f(x_1, \ldots, x_n) = \prod_{i=1}^{n} \underbrace{\frac{1}{\int_{\max(a_i, \tau_\textit{elp}(t_i))}^{b_i}f_i(x)dx}}_{\text{Normalization constant}} \cdot \underbrace{f_i(x_i + \tau_{\textit{elp}}(t_i))}_{\substack{\text{Marginal}\\ \text{distribution}}} \text{.}\]
\end{definition}
As for the induced \ac{SC}, the individual timers of the \ac{SSC} are stochastically independent.
Thus, each timer can be resampled individually---see 
\ifthenelse{\boolean{isArXiv}}{%
	\Cref{sec:app:resampling-implementation}%
}{
	\cite[App.\;B]{arXivVersion}%
}.

\begin{figure}[t]
	\centering
	\begin{subfigure}{\textwidth}
		\centering
		\begin{tikzpicture}
    \node[draw, circle] (l_ua) at (0, 0) {$l_{ua}$};
    \node[draw, circle] (l_nua) at (4.2, 0) {$l_{\stkout{u}a}$};
    \node[draw, circle] (l_nuna) at (7.2, 0) {$l_{\stkout{ua}}$};

    \path[draw] (-3.7, 0) edge[-stealth] node[above,pos=0.5,align=center,scale=0.8] {$\{t_{u\xmarkred} := \textsc{Unif}(9.8, 12),$\\ $\phantom{\{}t_{a\xmarkred} := \textsc{Unif}(18, 20)\}$} (l_ua);
    \path[draw] (l_ua) edge[-stealth, bend left=7] node[above,pos=0.5,align=center,scale=0.8] {$\{t_{u\xmarkred}\}, u\xmarkred,$\\ $\{t_{u\wrench} := \textsc{Unif}(0, 0.1)\}$} (l_nua);
    \path[draw] (l_nua) edge[-stealth, bend left=7] node[below,pos=0.5,align=center,scale=0.8] {$\{t_{u\wrench}\}, u\wrench,$\\ $\{t_{u\xmarkred} := \textsc{Unif}(9.8,12)\}$} (l_ua);
    \path[draw] (l_nua) edge[-stealth] node[above,pos=0.5,align=center,scale=0.8] {$\{t_{a\xmarkred}\}, a\xmarkred, \emptyset$} (l_nuna);
    
    \path[draw,gray] (l_ua) edge[-stealth, dashed, bend right=5] node[left,scale=0.8] {a\text{\scriptsize\color{gray}\ding{55}}} ++(0.3,-1.0);
\end{tikzpicture}
		\caption{Excerpt from composed \ac{IOSA} with omitted urgent actions}
		\label{fig:running-ex-scs-model}
		\vspace{2ex}
	\end{subfigure}
	\centering
	\hspace{\stretch{2}}
	\begin{minipage}[s]{0.3\textwidth}
		\vspace{-.18\paperheight}
		\begin{subfigure}{\textwidth}
			\centering
			\scalebox{0.9}{\begin{tikzpicture}
    \node (D1) at (0, 0) {$\langle l_{\stkout{ua}}, {\color{Orchid}D_1}\rangle$};
    \node (D2) at (0, 1) {$\langle l_{\stkout{u}a}, {\color{PineGreen}D_2}\rangle$};
    \node (D3) at (-2, 1) {$\langle l_{ua}, {\color{Blue}D_3}\rangle$};
    \node (D4) at (0, 2) {$\langle l_{\stkout{u}a}, {\color{Red}D_4}\rangle$};
    \node (D5) at (-2, 2) {$\langle l_{ua}, {\color{Orange}D_5}\rangle$};

    \draw (D2) edge[-stealth] node[left, scale=0.8] {$a\xmarkred$} (D1);
    \draw (D3) edge[-stealth] node[below, scale=0.8] {$u\xmarkred$} (D2);
    \draw (D4) edge[-stealth] node[above, sloped, scale=0.8] {$u\wrench$} (D3);
    \draw (D5) edge[-stealth] node[above, scale=0.8] {$u\xmarkred$} (D4);
\end{tikzpicture}}
			\vspace{-1ex}
			\caption{Backwards search}
			\label{fig:running-ex-scs-path}
		\end{subfigure}

		\vspace{3.7ex}
		\begin{subfigure}{\textwidth}
			\centering
			\begin{tikzpicture}
    \begin{axis}[
        xmin=0, xmax=0.25,
        ymin=0, ymax=2.7,
        axis lines=middle,
        xlabel=$\tau(t_{u\wrench})$,
        xlabel style={xshift=-1.5mm,yshift=3.5mm},
        axis y line=none,
        x=7.2cm,
        y=0.18cm,
        minor tick num=1,
        xtick={0,0.1,...,0.2},
        xlabel style={at={(ticklabel cs:1.05)},anchor=west},
        xticklabel style={yshift=0.1cm, rotate=-65},
        ]

        \addplot[
        fill=Orchid,
        fill opacity=0.3,
        draw=Orchid,
        very thick,
        ] coordinates {
            (0, 0)
            (0.1, 0)
            (0, 0)
        };
        \node[scale=1.1,Orchid] at (0.07, 1.5) {$D_1$};
        \node[scale=0.8,gray,align=center] at (11, 16) {Initial\\ domain};
    \end{axis}
\end{tikzpicture}
			\vspace{-.7ex}
			\caption{\acp{SC} for $l_{\stkout{ua}}$}
			\label{fig:running-ex-scs1}
		\end{subfigure}
	\end{minipage}
	\begin{subfigure}{0.3\textwidth}
		\centering
		\begin{tikzpicture}
    \begin{axis}[
        xmin=0, xmax=0.24,
        ymin=0, ymax=15.8,
        axis lines=middle,
        xlabel=$\tau(t_{u\wrench})$,
        ylabel=$\tau(t_{a\xmarkred})$,
        xlabel style={xshift=-0.1cm,yshift=0.3cm,at={(ticklabel cs:1.05)},anchor=west},
        ylabel style={
            at={(axis description cs:-0.5,0.95)},
            rotate=90,
            anchor=center,
        },
        x=7.2cm,
        y=0.24cm,
        grid=both,
        grid style={line width=.1pt, draw=gray!30},
        major grid style={line width=.2pt,draw=gray!50},
        minor tick num=1,
        xtick={0,0.1,...,0.2},
        ytick={0,3,...,15},
        xticklabel style={yshift=0.1cm, rotate=-65},
        clip mode=individual
        ]

        \addplot[
        fill=PineGreen,
        fill opacity=0.3,
        draw=PineGreen,
        thick,
        ] coordinates {
            (0,0)
            (0.1, 0)
            (0.1, 0.1)
            (0, 0)
        };

        \addplot[
        fill=Red,
        fill opacity=0.3,
        draw=Red,
        thick,
        ] coordinates {
            (0, 9.8)
            (0.1,9.9)
            (0.1,12.2)
            (0,12.1)
            (0,9.8)
        };

        \addplot[
        fill=gray,
        fill opacity=0.3,
        draw=gray,
        thick,
        ] coordinates {
        	(0,8.1)
        	(0.1,8.1)
        	(0.1,10.1)
        	(0,10.1)
        	(0,8.1)
        };

        \addplot[
        fill=gray,
        fill opacity=0.3,
        draw=gray,
        thick,
        ] coordinates {
        	(0,6.5)
        	(0.1,6.5)
        	(0.1,8.5)
        	(0,8.5)
        	(0,6.5)
        };

        \node[scale=1.1,PineGreen] at (0.08, 1.8) {$D_2$};
        \node[scale=1.1,Red] at (0.16, 11.5) {$D_4$};
        \node[scale=1.1,black!20!gray] at (0.16, 9.2) {$D_2'$};
        \node[scale=1.1,black!20!gray] at (0.16, 7.2) {$D_3'$};
    \end{axis}
\end{tikzpicture}
		\vspace{-3.55ex}
		\caption{\acp{SC} for $l_{\stkout{u}a}$}
		\label{fig:running-ex-scs2}
	\end{subfigure}
	\begin{subfigure}{0.33\textwidth}
		\centering
		\begin{tikzpicture}
    \begin{axis}[
        xmin=0, xmax=13.5,
        ymin=0, ymax=21.0,
        axis lines=middle,
        xlabel=$\tau(t_{u\xmarkred})$,
        ylabel=$\tau(t_{a\xmarkred})$,
        xlabel style={xshift=-0.2cm,yshift=0.3cm,at={(ticklabel cs:1.05)},anchor=west},
        ylabel style={
            at={(axis description cs:-0.36,0.95)},
            rotate=90,
            anchor=center,
        },
        x=0.18cm,
        y=0.18cm,
        grid=both,
        grid style={line width=.1pt, draw=gray!30},
        major grid style={line width=.2pt,draw=gray!50},
        minor tick num=1,
        xtick={0,4,...,12},
        ytick={0,4,...,20},
        xticklabel style={yshift=0.1cm, rotate=-65},
        clip mode=individual
        ]

        \addplot[
        fill=Blue,
        fill opacity=0.3,
        draw=Blue,
        thick,
        ] coordinates {
            (0, 0)
            (12, 12)
            (12, 12.1)
            (0, 0.1)
            (0, 0)
            (12, 12)
        };

        \addplot[
        fill=Orange,
        fill opacity=0.3,
        draw=Orange,
        thick,
        ] coordinates {
            (0, 12.2)
            (7.8, 20.0)
            (10.2, 20.0)
            (0, 9.8)
            (0, 12.2)
        };

        \addplot[
        fill=gray,
        fill opacity=0.3,
        draw=gray,
        thick,
        ] coordinates {
            (9.8, 18)
            (12, 18)
            (12, 20)
            (9.8, 20)
            (9.8, 18)
            (12,18)
        };
        \node[scale=1.1,Blue] at (5, 2) {$D_3$};
        \node[scale=1.1,Orange] at (6.4, 13.5) {$D_5$};

        \node[scale=1.1,black!20!gray,align=center] at (11, 16) {$D'_1$};

        \filldraw[Brown] (9.9,19.9) circle (0.5pt);
        \node[scale=1.1,Brown] at (9.9,21.3) {$\tau_1$};
        \filldraw[Brown] (9.9,18.4) circle (0.5pt);
        \node[scale=1.1,Brown] at (8.5,18.4) {\contour{white}{$\tau_2$}};
        \filldraw[Brown] (11.5,19.0) circle (0.5pt);
        \node[scale=1.1,Brown] at (13.3,19.0) {$\tau_3$};
    \end{axis}
\end{tikzpicture}
		\vspace{-3ex}
		\caption{\acp{SC} for $l_{ua}$}
		\label{fig:running-ex-scs3}
	\end{subfigure}
	\hspace{\stretch{2}}
	\caption{Selected \acp{SC} from backwards reachability search in the toy example from \Cref{fig:running-example}. Note the unequal axis scales in \Cref{fig:running-ex-scs2}.}
	\label{fig:running-ex-scs}
\end{figure}

\begin{example}
    \label{ex:running-example-expl-resampling}
	We revisit the toy example from \Cref{fig:running-example} in \Cref{fig:running-ex-scs}, where \Cref{fig:running-ex-scs-model} shows a part of the corresponding synchronized IOSA model.
	We denote with $u\xmarkred$/$u\wrench$ and $a\xmarkred$/$a\wrench$ the failure/repair of the \textsf{UPS} and \textsf{AC} \acp{BE}.
	\Cref{fig:running-ex-scs-path} shows the shortest possible sequence from an initial state to the target, which requires two failures with timer values $t_{u\xmarkred}$ near the lower bound of $9.8$ with intermediate repair $t_{u\wrench}$ and the failure $t_{a\xmarkred}$ near the upper bound of $20$ (see \Cref{ex:running-example}).
	The \ac{DBM} zones of
    ${\color{Orchid}D_1}$ to ${\color{Orange}D_5}$ are (computed via backwards search in \cite{dengler:2025:time-sensitive-isplit} and
\ifthenelse{\boolean{isArXiv}}{%
	\Cref{sec:app:reachability-state-classes}%
}{
	\cite[App.\;A]{arXivVersion}%
}):
	\begin{footnotesize}
		\begin{itemize}
			\item ${\color{Orchid}D_1} = \{ 0\leq \tau(t_{u\wrench})\leq 0.1 \}$;
			\item ${\color{PineGreen}D_2} = \{ 0\leq \tau(t_{u\wrench})\leq 0.1\wedge 0\leq \tau(t_{a\xmarkred}) \leq 0.1\wedge 0\leq \tau(t_{u\wrench}) - \tau(t_{a\xmarkred})\leq 0.1\}$;
			\item ${\color{Blue}D_3} = \{0\leq \tau(t_{u\xmarkred})\leq 12\wedge0 \leq \tau(t_{a\xmarkred})\leq 12.1\wedge -0.1\leq \tau(t_{u\xmarkred})-\tau(t_{a\xmarkred})\leq 0 \}$;
			\item ${\color{Red}D_4} = \{ 0\leq \tau(t_{u\wrench}) \leq 0.1 \wedge 9.8\leq \tau(t_{a\xmarkred})\leq 12.2 \wedge -12.1\leq \tau(t_{u\wrench}) - \tau(t_{a\xmarkred})\leq -9.8\}$;
			\item ${\color{Orange}D_5} = \{0\leq \tau(t_{u\xmarkred})\leq 10.2\wedge 9.8\leq \tau(t_{a\xmarkred})\leq 20 \wedge -12.2\leq \tau(t_{u\xmarkred})-\tau(t_{a\xmarkred})\leq -9.8 \}$.
		\end{itemize}
	\end{footnotesize}
	At simulation start, the values for the initial timers $t_{u\xmarkred}$ and $t_{a\xmarkred}$ are (uniformly) sampled from
    ${\color{black!20!gray}D'_1} = \{9.8 \leq \tau(t_{u\xmarkred}) \leq 12 \wedge 18\leq \tau(t_{a\xmarkred}) \leq 20 \}$, i.e., the Cartesian product of the domains of $\tau(t_{u\xmarkred})$ and $\tau(t_{a\xmarkred})$. Assume the following sampled values:
    \begin{footnotesize}
    	\begin{itemize}
    		\item ${\color{Brown}\tau_1} = \{t_{u\xmarkred} \mapsto 9.9, t_{a\xmarkred} \mapsto 19.9\}$ (shorthand notation for a function ${\color{Brown}\tau_1}$ that satisfies ${\color{Brown}\tau_1}(t_{u\xmarkred}) = 9.9$, ${\color{Brown}\tau_1}(t_{a\xmarkred}) = 19.9$, and ${\color{Brown}\tau_1}(t) = \bot$ for any other $t\in \mathcal{T}$); \vspace{-1.2em}
    		\begin{multicols}{2}
    			\item ${\color{Brown}\tau_2} = \{t_{u\xmarkred} \mapsto 9.9, t_{a\xmarkred} \mapsto 18.4\}$;
    			\item ${\color{Brown}\tau_3} = \{t_{u\xmarkred} \mapsto 11.5, t_{a\xmarkred} \mapsto 19.0\}$.
    		\end{multicols}
    	\end{itemize}
    \end{footnotesize}
	\vspace{-0.5em}
	Then, the distance evaluation for time-sensitive \ac{ISPLIT} operates as follows:
	\begin{itemize}
		\item \emph{Regular simulation:} It holds that ${\color{Brown}\tau_1} \in {\color{Orange}D_5}$, so $d_T(\langle l_{ua}, {\color{Brown}\tau_1}\rangle) = 4$. On the other hand, both ${\color{Brown}\tau_2} \not\in {\color{Orange}D_5}$ and ${\color{Brown}\tau_3} \not\in {\color{Orange}D_5}$, so it holds $d_T(\langle l_{ua}, {\color{Brown}\tau_2}\rangle) > 4$ and $d_T(\langle l_{ua}, {\color{Brown}\tau_3}\rangle) > 4$.
		Eventually, we obtain $d_T(\langle l_{ua}, {\color{Brown}\tau_2}\rangle) = d_T(\langle l_{ua}, {\color{Brown}\tau_3}\rangle) = 14$.
		\item \emph{With resampling:} In the initial state,
        the already elapsed time is $\tau_\textit{elp} = \{t_{u\xmarkred} \mapsto 0, t_{a\xmarkred} \mapsto 0\}$ and the \ac{SC} domain is ${\color{black!20!gray}D'_1}$, which intersects with ${\color{Orange}D_5}$. So, we obtain for all timer evaluations $d^r_T(\langle l_{ua}, {\color{Brown}\tau_1}\rangle) = d^r_T(\langle l_{ua}, {\color{Brown}\tau_2}\rangle) = d^r_T(\langle l_{ua}, {\color{Brown}\tau_3}\rangle) = 4$.
		
		Depending on whether we choose ${\color{Brown}\tau_1}$/${\color{Brown}\tau_2}$ ($t_{u\xmarkred} \mapsto 9.9$) or ${\color{Brown}\tau_3}$ ($t_{u\xmarkred} \mapsto 11.5$), however, the elapsed time $\Delta$ and the induced \ac{SC} in the location $l_{\stkout{u}a}$ differ:
        \begin{itemize}
            \item For ${\color{Brown}\tau_1}$/${\color{Brown}\tau_2}$, $\Delta = 9.9$, the elapsed time information for the subsequent state is $\tau'_{\textit{elp}} = \{t_{u\wrench} \mapsto 0.0, t_{a\xmarkred} \mapsto 9.9\}$ and the domain of the induced \ac{SC} ${\color{black!20!gray}D'_2} = \{0 \leq \tau(t_{u\wrench}) \leq 0.1 \wedge 8.1\leq \tau(t_{a\xmarkred}) \leq 10.1 \}$. The intersection ${\color{black!20!gray}D_2'} \cap {\color{Red}D_4} \neq \emptyset$ is non-empty, leading to $d^r_T(s') = 3$ for every $s'\in \langle l_{\stkout{u}a}, {\color{black!20!gray}D_2'}\rangle$.
            \item On the other hand, for ${\color{Brown}\tau_3}$, $\Delta = 11.5$, $\tau'_{\textit{elp}} = \{t_{u\wrench} \mapsto 0.0, t_{a\xmarkred} \mapsto 11.5\}$, and the domain of the induced \ac{SC} is ${\color{black!20!gray}D'_3} = \{0 \leq \tau(t_{u\wrench}) \leq 0.1 \wedge 6.5\leq \tau(t_{a\xmarkred}) \leq 8.5 \}$.
            As the intersection ${\color{black!20!gray}D_3'} \cap {\color{Red}D_4} = \emptyset$ is empty, $d^r_T(s') > 3$ for every $s'\in \langle l_{\stkout{u}a}, {\color{black!20!gray}D_3'}\rangle$. With further exploration, we obtain $d^r_T(s') = 13$.
        \end{itemize}
	\end{itemize}
    Summarizing, resampling considers the set of possible timer values in a state, eradicating the need to sample all timer values correctly for a certain importance (see ${\color{Brown}\tau_1}$). Practically, this leads to a lower variance per \ac{ISPLIT} run (see \Cref{sec:experiments}).
\end{example}

\subsection{Models with Unbounded Timers}
\label{sec:resampling-timers:models-infinite-domains}

This subsection is concerned with the timed distance metrics in the special case that each timer in the model has an infinite domain $[0, \infty)$. In fact, this property applies to a plethora of commonly used distributions, e.g., the exponential, Erlang, or Weibull distribution.
Although we can reuse the backwards reachability metric search for
the timed distance $d_T$ based on exploring the \ac{SC} graph, dealing with \ac{DBM} zones can induce a significant performance hit: For example (see \Cref{sec:background:state-classes}), we perform strict inclusion checking during backwards reachability search with \acp{SC} to limit the explored state space.
However, this approach could be further optimized: We could merge adjacent \ac{DBM} zones (e.g., when $D_1\cup D_2$ can be represented as a \ac{DBM} zone again), or omit an \ac{SC} $\Sigma_1 = \langle l, D_1\rangle$ such that there exist \acp{SC} $\Sigma_2 = \langle l, D_2\rangle$, \ldots, $\Sigma_k = \langle l, D_k\rangle$ (with $k > 2$) with lower distances to the target than $\Sigma_1$ and $D_1 \subseteq D_2 \cup \ldots \cup D_k$. Although such operations would be strongly beneficial to reduce the \ac{SC} graph size,
their implementation is known to be considerably hard for \ac{DBM} zones~\cite{david:2005:DBM-merging}, rendering their use impractical.

\myparagraph{Efficient data structure.}
For order-based distance metrics, denoted in the following as $d_{O} \colon \states \rightarrow \mathbb{N}_0 \cup \{\infty\}$ (acronym for \underline{d}istance \underline{o}rdered), we develop a specialized data structure by observing that, in this case, any \ac{DBM} zone can be captured by the union of \emph{elementary} \acp{SC} that store an ordering of timers:
\begin{definition}[Elementary state class for ordered distances]
	For a location $l$ and an ordering of timers $T_l$ encoded as permutation $\sigma \colon T_l \to T_l$, the \emph{elementary \ac{SC}} is $\Sigma_{\sigma} = \langle l, D_{\sigma}\rangle$ with $n=|T_l|$ timers and the \emph{elementary domain}
	\begin{align*}
		D_{\sigma} = \bigcap_{\mathmakebox[\widthof{\ensuremath{\scriptstyle 1\leq i<n}}][c]{1\leq i\leq n-1}} \{ \tau(\sigma(t_i)) \leq \tau(\sigma(t_{i+1})) \} = \{ \tau(\sigma(t_1)) \leq \tau(\sigma(t_2)) \leq \ldots \leq \tau(\sigma(t_n)) \}\text{.}
	\end{align*}
\end{definition}
There exist $|T_l|!$ elementary \acp{SC} for a location $l$ as there are as many permutations for a list of $|T_l|$ elements. By making use of the Lehmer code~\cite{lehmer:1960:teaching-tricks-code}, we can index each permutation, leading to the bijection $\lambda \colon \textit{Perm}(T_l) \to \{1, 2, \ldots, |T_l|!\}$.
To calculate the metric $d_{O}$, instead of working with \ac{DBM} zones directly, we store for each location $l$ an array of size $|T_l|!$ that contains the value of the distance metric $\omega_O(\Sigma)$ for every elementary \ac{SC} $\Sigma$. We start with the elementary \acp{SC} of the target locations and perform backwards reachability search as in \cite{dengler:2025:time-sensitive-isplit} and
\ifthenelse{\boolean{isArXiv}}{%
	\Cref{sec:app:reachability-state-classes}%
}{
	\cite[App.\;A]{arXivVersion}%
}.
Consequently, the size of the created \ac{SC} graph lies in $\mathcal{O}(\sum_{l\in \mathcal{L}} |T_l|!)$.

During simulation, the distance $d_{O}$ can be evaluated for a state $s = \langle l, \tau\rangle$ by determining the timer ordering permutation $\sigma$, evaluating the Lehmer code $\lambda(\sigma)$ (implemented in time $\mathcal{O}(|T_l|^2)$), and looking up the distance of $\omega_O(\Sigma_\sigma)$.

\myparagraph{Resampling for ordered distance metrics.} Next, we focus on the following question: What changes when we apply resampling to ordered distance metrics? As every timer is unbounded, resampling can always arbitrarily change the order of timers.
This raises the question: When working with unbounded timers, does resorting to time-sensitive \acp{IFUN} make a difference? The answer is no, formalized as follows (the proof is in
\ifthenelse{\boolean{isArXiv}}{%
	\Cref{sec:app:omitted-proofs}%
}{
	\cite[App.\;C]{arXivVersion}%
}),
where $d^r_O(s) = \min_{s'\in \Sigma_{s_\textit{elp}}} d_{O}(s')$ as in \Cref{def:time-sensitive-dist-resampling}:
\begin{theorem}
    \label{thm:distance-metric-equivalence-order}
	For every state $s\in\states$ it holds that $d(s) = d^r_O(s)$.
\end{theorem}
This result has practical implications: When working with unbounded timer domains, it suffices for resampling to construct the cheaper-to-compute regular distance metric. On the other hand, an \ac{ISPLIT} run without resampling is on average faster, albeit with higher variance. The practical performance tradeoff between both approaches is discussed in \Cref{sec:experiments}.

\subsection{Threshold Selection}
\label{sec:resampling-timers:threshold-selection}

Most threshold selection approaches for \ac{ISPLIT} (see \Cref{sec:background:RES}) start in the initial location and estimate the level-up probabilities. This requires that the initial states have (almost) the lowest possible importance value. However, in a time-sensitive setting, the initial states might have different importance values \cite{cerou:2012:sequential-mc, BDH19}.

\emph{Regular simulation.} To find the states in the initial location with lowest importance, we sample $n$ different initial states and record the importance values, from which we select the $k$ lowest ones and note their importance range.
In our experiments, we set $n = 1000$ and $k = 10$, which ensures that a wide range of possible importance values is covered, while preventing too expensive rejection sampling from the states with the lowest importance values.
Only states with an importance value in the range of the $k$ selected states are used as initial states.

\emph{With resampling.} On the other hand, when using resampling, the importance value is the same for any initial state, as the elapsed time information $\tau_{\textit{elp}}$ is equal, and thus the induced \ac{SC} $\Sigma_{s_{\textit{elp}}}$. Thus, no adjustments are necessary here.

\section{Considering the Global Age}
\label{sec:global-age}

As the second part of our main contribution, we analyze how to incorporate global time into the time-sensitive framework. For time-bounded reachability analysis, the global time can be used to determine whether the target remains reachable within the available time budget. Otherwise, the corresponding simulation offspring can be pruned, significantly reducing simulation effort.

The basis for including the global time information is by means of an artificial timer that records the remaining time budget until the simulation time bound.
\begin{definition}[Global age timer]
    \label{def:global-age-timer}
    We denote with $t_\textit{age}$ the \emph{global age timer} recording the remaining time budget \emph{upon entering} a state, and with $\mathcal{T}_\textit{age} = \mathcal{T}\cup \{t_\textit{age}\}$ the set of all timers including the global age timer.
\end{definition}
This idea is a standard concept in forwards analysis with \acp{SC} or \acp{SSC} \cite{horvath:2012:transient-sscs,dengler:2024:transient-eval-sscs-sim}.
In our backwards reachability search setting, timer elapsing is followed in the opposite direction. Thus, at the end of the simulation, $t_\textit{age}$ reaches the lowest value~$0$, resulting in potential values of $\tau(t_\textit{age})$ in $[0, \infty)$.
Consequently, we extend the state space $\states$ by a dimension for $t_\textit{age}$, so $\states_\textit{age} = \mathcal{L} \times (\mathbb{R}_{\geq 0} \mathbin{\dot{\cup}}\{ \bot \})^{|\mathcal{T}|} \times \mathbb{R}_{\geq 0}$.

\myparagraph{Distance metrics.} For a global time bound $t_{\textit{age}}^{\textit{max}} \in \mathbb{R}_{\geq 0}$, we extend the distance metrics $d_{T}$ (\Cref{sec:background:state-classes}) and $d^r_{T}$ (\Cref{sec:resampling-timers}, \Cref{def:time-sensitive-dist-resampling}), as well as $d_{O}$ and $d^r_{O}$ (\Cref{sec:resampling-timers:models-infinite-domains}) with the following global age variants returning the minimum number of transitions needed to reach the target 
within the time budget $t_\textit{age}^\textit{max}$:
\begin{itemize}
    \item $d_{TG}$ and $d^r_{TG}$ (for \underline{d}istance \underline{t}imed \underline{g}lobal), considering timer bounds;
    \item $d_{OG}$ and $d^r_{OG}$ (for \underline{d}istance \underline{o}rder \underline{g}lobal), neglecting timer bounds.
\end{itemize}
Analogously, the concept can be readily extended to transient probabilities, i.e., the probability of being at a target state exactly at a given time $t_\textit{age}^\textit{max}$:
\begin{itemize}
    \item $d_{TEG}$ and $d^r_{TEG}$ (for \underline{d}istance \underline{t}imed \underline{e}xactly \underline{g}lobal), considering timer bounds;
    \item $d_{OEG}$ and $d^r_{OEG}$ (for \underline{d}istance \underline{o}rder \underline{e}xactly \underline{g}lobal), neglecting timer bounds.
\end{itemize}

\myparagraph{State class adaptations.} For the computation of the distance metrics $\omega_{TG}(\Sigma)$, $\omega_{OG}(\Sigma)$, $\omega_{TEG}(\Sigma)$, and $\omega_{OEG}(\Sigma)$ (naming scheme as above), we distinguish for all \acp{SC} $\Sigma' = \langle l', D'\rangle$ between time-bounded and transient reachability queries:
\begin{itemize}
    \item For time-bounded queries ($\omega_{TG}(\Sigma)$, $\omega_{OG}(\Sigma)$), no adjustments to the target \ac{SC} $\Sigma'$ as calculated in
\ifthenelse{\boolean{isArXiv}}{%
	\Cref{sec:app:reachability-state-classes}%
}{
	\cite[App.\;A]{arXivVersion}%
}
and \cite{dengler:2025:time-sensitive-isplit} need to be made, as the location $l'$ can be left before the global time bound in a simulation run.
    \item For transient reachability queries ($\omega_{TEG}(\Sigma)$, $\omega_{OEG}(\Sigma)$), we must ensure that a simulation run does not leave $l'$ before the end of the time bound. Thus, we restrict the domain $D'$ of $\Sigma'$ with $D'_{E} = D' \cap (\bigcap_{t\in T_{l'}} \{ \tau(t) \geq \tau(t_{\textit{age}}) \})$.
\end{itemize}
For calculating the predecessor \ac{SC} when inversely following a transition $l\xtr{\scriptscriptstyle T,a,T'}l'$, no adjustments to the procedures in
\ifthenelse{\boolean{isArXiv}}{%
	\Cref{sec:app:reachability-state-classes}%
}{
	\cite[App.\;A]{arXivVersion}%
}
and \cite{dengler:2025:time-sensitive-isplit} need to be made. The restrictions imposed on $t_\textit{age}$ and the other timers are updated implicitly.

For evaluating the importance of a state under consideration of resampling, one needs to consider: As $t_\textit{age}$ is not a classical timer to which a probability distribution is assigned, it can also not be resampled and has a fixed value of $\tau(t_\textit{age})$. Thus, we obtain as induced \ac{SC} (adapted from \Cref{def:induced-sc}):
\begin{definition}[Induced state class with global age]
    \label{def:induced-sc-age}
    The \emph{induced \ac{SC}} of $s = \langle l, \tau\rangle$ \emph{(with global age)} is the \ac{SC} $\Sigma_{s_\textit{elp}} = \langle l, D\rangle$ with domain $D = (\bigtimes_{t_i \in T_{l}} [\tau_{\textit{min}}(t_i), \tau_{\textit{max}}(t_i)]) \times \{\tau(t_\textit{age})\}$.
\end{definition}
The extension of the induced \acp{SSC} (cf.\ \Cref{def:induced-ssc}) is analogous. For distance evaluation, we determine the \ac{SC} with the smallest distance to the target from the \ac{SC} graph (with global age) that intersects with the induced \ac{SC}, as in \Cref{sec:resampling-timers}.

\myparagraph{Pruning.} During backwards reachability search, we can prune \acp{SC} that are not reachable in the given time budget. Formally, when it holds for an \ac{SC} $\Sigma = \langle l, D\rangle$ that $\min_{\tau \in D} \tau(t_\textit{age}) > t_\textit{age}^\textit{max}$, no further exploration is needed for this branch: The time budget does not suffice for a simulation run to intersect with $\Sigma$.

Second, when during simulation the distance of a state $s$ evaluates to $\infty$, i.e., no \ac{SC} has been found, a simulation offspring can be pruned.
Note that an unbiased result requires the \ac{SC} graph to be fully constructed.
While this is also theoretically conceivable for distance metrics without incorporation of global timers, it is usually only effective with them, as (i) many models allow reaching the target from any location (pruning is solely possible when incorporating the time bound) and (ii) for tight time bounds, constructing the pruned \ac{SC} graph might be cheaper than constructing the complete \ac{SC} graph without age timer:
\begin{example}
    We revisit the toy example from \Cref{fig:running-example,fig:running-ex-scs}. Constructing the complete backwards \ac{SC} graph without global age results in $14$ locations and $97$ \acp{SC}.
    For time-bounded reachability queries, the \ac{SC} graph with global age and a time limit of $t_{\textit{age}}^{\textit{max}} = 21$ only has $43$ \acp{SC} when pruning \acp{SC} outside the time budget.
    For the initial location $l_{ua}$, the domains ${\color{Blue}D_3}$ and ${\color{Orange}D_5}$ are updated:
    \begin{footnotesize}
		\begin{itemize}
            \item ${\color{Blue}D_3} = \{0\leq \tau(t_{u\xmarkred})\leq 12\wedge0 \leq \tau(t_{a\xmarkred})\leq 12.1\wedge 0 \leq \tau(t_\textit{age}) \wedge -0.1\leq \tau(t_{u\xmarkred})-\tau(t_{a\xmarkred})\leq 0\wedge 0 \leq \tau(t_\textit{age}) - \tau(t_{u\xmarkred})\wedge 0 \leq \tau(t_\textit{age}) - \tau(t_{a\xmarkred}) \}$
			\item ${\color{Orange}D_5} = \{0\leq \tau(t_{u\xmarkred})\leq 10.2\wedge 9.8\leq \tau(t_{a\xmarkred})\leq 20 \wedge 9.8 \leq \tau(t_{\textit{age}}) \wedge -12.2\leq \tau(t_{u\xmarkred})-\tau(t_{a\xmarkred})\leq -9.8 \wedge 9.8 \leq \tau(t_{\textit{age}}) - \tau(t_{u\xmarkred}) \wedge 0 \leq \tau(t_{\textit{age}}) - \tau(t_{a\xmarkred}) \}$
        \end{itemize}
    \end{footnotesize}
    As in the initial state (cf.\ domain {\color{black!20!gray}$D'_1$}) $\tau(t_{u\xmarkred}) \geq 9.8$ and $\tau(t_{a\xmarkred}) \geq 18$, a timer evaluation $\tau\in {\color{Orange}D_5}$ must in particular fulfill $\tau(t_\textit{age}) \geq \max(9.8+9.8, 18) = 19.6$. So, the time budget must be at least $19.6$ such that we initially have $d_{TG}(s) = 4$. 
\end{example}

\section{Taxonomy of Distance Metrics}
\label{sec:taxonomy}

\begin{figure}
    \centering
    \includegraphics[width=.9\linewidth]{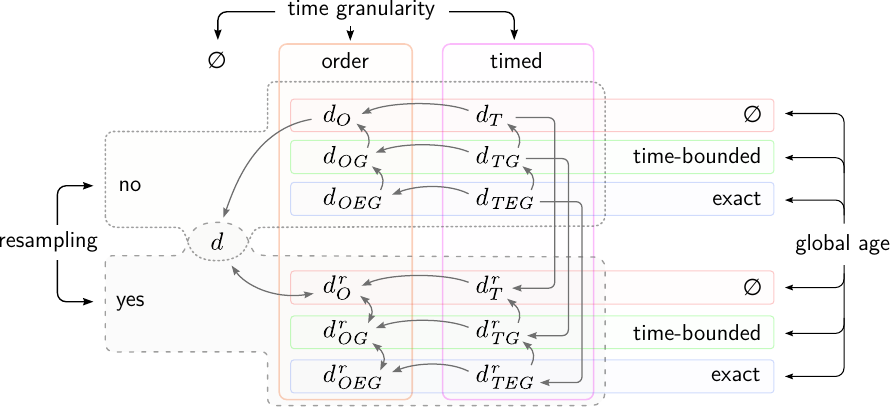}
    \caption{Taxonomy of the derived distance metrics}
    \label{fig:taxonomy}
	\vspace{-1.5ex}
\end{figure}

To categorize the distance metrics from \Crefrange{sec:background}{sec:global-age} into a \emph{taxonomy}, we address the following question: How does adding information (e.g., timer bounds, global age) affect the resulting \ac{IFUN}? Since these metrics estimate the distance to the target, adding information refines the estimate by excluding unfeasible paths, thus \textit{increasing} the calculated distance. According to this, a lower distance estimate (e.g., one considering only the logical location) indicates a coarser approximation that \textit{falsely} suggests a higher importance by ignoring additional information (e.g., the order of timers).
This is formalized as follows:
\begin{definition}[Distance metric granularity]
    We say that a distance metric $d_A$ is \emph{at least as fine-grained as} distance metric $d_B$, denoted $d_A \succeq d_B$, if it holds for all states $s\in \states$ that $d_A(s)\geq d_B(s)$.
\end{definition}
Relation $\preceq$ defines a preorder on the distance metrics.
In particular, two distance metrics $d_A$ and $d_B$ are said to be equivalent, denoted $d_A \equiv d_B$, when it holds for all states $s\in \states \colon d_A(s) = d_B(s)$.
In \Cref{thm:distance-metric-equivalence-order} we showed that $d \equiv d^r_{O}$.
In addition, the following statements also hold (proofs are in
\ifthenelse{\boolean{isArXiv}}{%
	\Cref{sec:app:omitted-proofs}%
}{
	\cite[App.\;C]{arXivVersion}%
}):
\begin{theorem}
	\label{thm:distance-metric-taxonomy}
	The following granularity statements about distance metrics hold:
	\begin{enumerate}[topsep=.5ex]
	\item \label{en:thm2-1} $d \equiv d^r_{O} \equiv d^r_{OG} \equiv d^r_{OEG}$
	    \hfill\textcolor{black!65}{(equivalence for order w/resampling)};
	\item \label{en:thm2-2} $d_O \preceq d_{OG} \preceq d_{OEG}$, $d_T \preceq d_{TG} \preceq d_{TEG}$, $d^r_T \preceq d^r_{TG} \preceq d^r_{TEG}$
	    \hfill\textcolor{black!65}{(age inclusion)};
	\item \label{en:thm2-3} $d \preceq d_O$, $d_O \preceq d_T$, $d_{OG} \preceq d_{TG}$, $d_{OEG} \preceq d_{TEG}$, $d \preceq d^r_T$
	    \hfill\textcolor{black!65}{(time granularity)};
	\item \label{en:thm2-4} $d^r_T \preceq d_T$, $d^r_{TG} \preceq d_{TG}$, $d^r_{TEG} \preceq d_{TEG}$
	    \hfill\textcolor{black!65}{(resampling yes/no)}.
	\end{enumerate}
\end{theorem}

\Cref{thm:distance-metric-taxonomy} can be visualized with the diagram shown in \Cref{fig:taxonomy}, where
\tikz[baseline=-.5ex,inner sep=2pt]{
	\node(A) at (0,0) {$d_A$};
	\node(B) at (.95,0) {$d_B$};
	\draw [-{Latex[length=5pt]},semithick] (B) -- (A);
}
denotes $d_A \preceq d_B$, viz.\ arrows point to the coarser (less fine-grained) distances.
Since $d$ is time-agnostic, there is no resampling variant shown for it.
\Cref{fig:taxonomy} further categorizes all metrics in three dimensions, based on whether:
\begin{itemize}[topsep=.5ex,parsep=.1ex]
\item	resampling is used (\textsf{yes/no})---items \ref{en:thm2-1} and \ref{en:thm2-4} in \Cref{thm:distance-metric-taxonomy};
\item	the timer bounds are considered ($\emptyset$, \textsf{order}, or \textsf{timed})---item \ref{en:thm2-3} in \Cref{thm:distance-metric-taxonomy};
\item	the global age is considered ($\emptyset$, \textsf{time-bounded}, or \textsf{exact})---item \ref{en:thm2-2} in \Cref{thm:distance-metric-taxonomy}.
\end{itemize}

\begin{example}
	Reconsider the toy example from \Cref{fig:running-example,fig:running-ex-scs}.
	\Cref{fig:granularity-levels} shows how the \ac{SC} $\langle l_{ua}, {\color{Blue}D_3} \rangle$ (\Cref{fig:running-ex-scs3}) from the backwards reachability search (\Cref{fig:running-ex-scs-path}) changes when, instead of $d_T$, one uses a coarser ($d_O$, see \Cref{fig:granularity-levels:order}) or finer ($d_{TG}$, see \Cref{fig:granularity-levels:global-age}) distance metric.
    As can be seen, $d_O$ allows a larger domain than $d_T$, which in turn allows a larger domain than $d_{TG}$.
\end{example}

\begin{figure}
	\centering
	\begin{subfigure}{0.29\linewidth}
		\centering
		\scalebox{.9}{\begin{tikzpicture}
	\begin{axis}[
		xmin=0, xmax=13.5,
		ymin=0, ymax=12.5,
		axis lines=middle,
		xlabel=$\tau(t_{u\xmarkred})$,
		ylabel=$\tau(t_{a\xmarkred})$,
		xlabel style={xshift=-0.7cm,yshift=-0.1cm,at={(ticklabel cs:1.05)},anchor=west},
		ylabel style={
            at={(axis description cs:-0.4,0.95)},
            rotate=90,
            anchor=center,
        },
		x=0.15cm,
		y=0.15cm,
		grid=both,
		grid style={line width=.1pt, draw=gray!30},
		major grid style={line width=.2pt,draw=gray!50},
		minor tick num=1,
		xtick={0,4,...,12},
		ytick={0,4,...,20},
		xticklabel style={yshift=0.1cm, rotate=-65},
		]

		\addplot[
			fill=Orange,
			fill opacity=0.3,
			draw=Orange,
            thick
		] coordinates {
			(0, 0)
			(13, 13)
            (0, 13)
            (0, 0)
		};
	\end{axis}
\end{tikzpicture}}
		\caption{Order alone ($d_O$)}
		\label{fig:granularity-levels:order}
	\end{subfigure}
	\begin{subfigure}{0.31\linewidth}
		\centering
		\scalebox{.9}{\begin{tikzpicture}
	\begin{axis}[
		xmin=0, xmax=13.5,
		ymin=0, ymax=12.5,
		axis lines=middle,
		xlabel=$\tau(t_{u\xmarkred})$,
		ylabel=$\tau(t_{a\xmarkred})$,
		xlabel style={xshift=-0.7cm,yshift=-0.1cm,at={(ticklabel cs:1.05)},anchor=west},
		ylabel style={
            at={(axis description cs:-0.4,0.95)},
            rotate=90,
            anchor=center,
        },
		x=0.15cm,
		y=0.15cm,
		grid=both,
		grid style={line width=.1pt, draw=gray!30},
		major grid style={line width=.2pt,draw=gray!50},
		minor tick num=1,
		xtick={0,4,...,12},
		ytick={0,4,...,20},
		xticklabel style={yshift=0.1cm, rotate=-65},
		]

		\addplot[
			fill=Blue,
			fill opacity=0.3,
			draw=Blue,
            thick
		] coordinates {
			(0, 0)
			(12, 12)
			(12, 12.1)
			(0, 0.1)
			(0, 0)
			(12, 12)
		};

        \node[scale=1.1,Blue] at (5, 2) {$D_3$};
	\end{axis}
\end{tikzpicture}}
		\caption{With timer bounds ($d_T$)}
		\label{fig:granularity-levels:bounds}
	\end{subfigure}
	\begin{subfigure}{0.36\linewidth}
		\centering
		\scalebox{.9}{\begingroup
\tdplotsetmaincoords{45}{27}
\begin{tikzpicture}[
		tdplot_main_coords,
        scale=0.15]

	\foreach \x in {0,4,...,12}
		\foreach \y in {0,4,...,12}
		{
			\draw[very thin, gray!50] (\x,0) -- (\x,12.5);
			\draw[very thin, gray!50] (0,\y) -- (13,\y);
		}
    \foreach \x in {2,6,...,10}
		\foreach \y in {2,6,...,10}
		{
			\draw[very thin, gray!30] (\x,0) -- (\x,12.5);
			\draw[very thin, gray!30] (0,\y) -- (13,\y);
		}

    \foreach \x in {2,6,...,10}
    {
        \draw[very thin, gray] (\x, -0.25) -- (\x, 0.25);
        \draw[very thin, gray] (-0.25, \x) -- (0.25, \x);
    }
    \foreach \x in {4,8,...,12}
    {
        \draw[very thin, gray] (\x, -0.5) -- (\x, 0.5);
        \draw[very thin, gray] (-0.5, \x) -- (0.5, \x);
        \node[rotate=-55] at (\x, -2) {$\x$};
    }
    \node[rotate=10] at (-1.3, 4) {$4$};
    \node[rotate=10] at (-1.3, 8) {$8$};
    \node[rotate=10] at (-1.7, 12) {$12$};

    \draw[very thin, gray] (0, -0.5, 6) -- (0, 0.5, 6);
    \node[rotate=20] at (0, -2, 6) {$6$};
    \draw[very thin, gray] (0, -0.5, 12) -- (0, 0.5, 12);
    \node[rotate=20] at (-1, -2, 12) {$12$};

	\draw[-stealth] (0,0,0) -- (13,0,0) node[anchor=west]{$\tau(t_{u\xmarkred})$};
	\draw[-stealth] (0,0,0) -- (0,12.5,0) node[anchor=south west]{$\tau(t_{a\xmarkred})$};
	\draw[-stealth] (0,0,0) -- (0,0,12.5) node[anchor=south east]{$\tau(t_{\textit{age}})$};

    \draw[fill=Orchid, fill opacity=0.3, draw=Orchid, thick] (0, 0, 0) -- (0, 0.1, 0.1) -- (12, 12.1, 12.1) -- (12, 12, 12) -- cycle;
    \draw[draw=Orchid, thick] (0, 0.1, 13) -- (0, 0.1, 0) -- (12, 12.1, 12) -- (12, 12.1, 13);
    \draw[draw=Orchid, thick] (0, 0, 13) -- (0, 0, 0) -- (12, 12, 12) -- (12, 12, 13);
    \fill[Orchid, fill opacity=0.3] (0, 0, 0) -- (12, 12, 12) -- (12, 12, 13) -- (0, 0, 13) -- cycle;
    \fill[Orchid, fill opacity=0.3] (12, 12, 13) -- (12, 12, 12) -- (12, 12.1, 12) -- (12, 12.1, 13) -- cycle;
\end{tikzpicture}
\endgroup}
		\caption{With global age ($d_{TG}$)}
		\label{fig:granularity-levels:global-age}
	\end{subfigure}
	\caption{Changes of domain ${\color{Blue}D_3}$ from \Cref{fig:running-ex-scs} for different distance metrics}
	\label{fig:granularity-levels}
    \vspace{-1ex}
\end{figure}
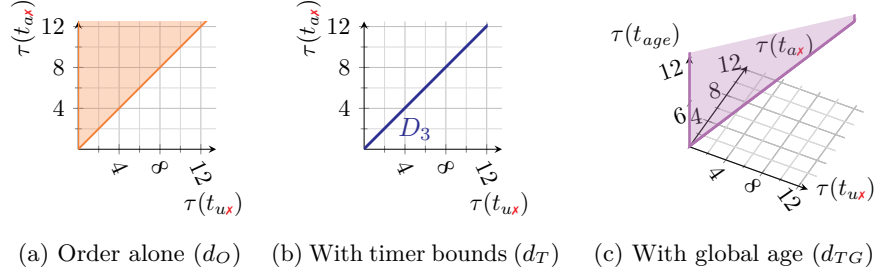

\myparagraph{From theory to practice.}
We highlight that this taxonomy compares the amount of information that each distance metric can contain, which makes it essentially a \emph{theoretical result}.
For example, $d_T\preceq d_{TG}$ states that by including the global age, $d_{TG}$ provides a finer-grained measure than $d_T$ of the number of transitions needed to reach the target, i.e.\ $d_{TG}$ can contain more information than $d_T$.
Analogously, $d^r_T\preceq d_T$ states that using resampling results in a coarser, less-informative measure---\emph{however, this does not immediately translate into practical performance gain/loss}.
Especially for resampling, while \Cref{thm:distance-metric-taxonomy} and \Cref{fig:taxonomy} state that $d_*^r \preceq d_*$ for every distance metric $d_*$ and its resampling variant $d_*^r$, in fact $d_*^r$ can be the more performant \ac{IFUN} when used for \ac{RES} estimations.
This is due to the better exploration of the induced \ac{SC}---and subsequent decrease in the estimator variance---achieved by resampling timers during simulation.
Therefore, for performance assessment one should compare solely inside the $d_*$ or $d_*^r$ class.

The above raises the question of \emph{practical results}:
To what extent does the preorder of distance metrics in the taxonomy translate into more performant \ac{RES} implementations?
To answer this question, which also concerns how efficiently the \acp{IFUN} can be stored and evaluated, \Cref{sec:experiments} performs practical experiments on three different case studies.

\section{Experimental Evaluation}
\label{sec:experiments}

We demonstrate our theory on repairable \acp{DFT} and a queueing network. Our implementation, based on the \textsc{modes} \ac{SMC} tool~\cite{HH14}, extends the prototypical time-sensitive \ac{ISPLIT} from \cite{dengler:2025:time-sensitive-isplit}.
All experiments ran on a 16-core AMD Ryzen 9 7950X3D CPU with 128\:GB of RAM running Ubuntu 22.04.

\subsection{Tested Model Classes}
\label{sec:experiments:models}

\subsubsection{Repairable \texorpdfstring{\acp{DFT}}{DFT}.}

\begin{wrapfigure}[9]{r}{0.52\textwidth}
	\vspace{-5ex}
	\centering
	\includegraphics[width=.83\linewidth]{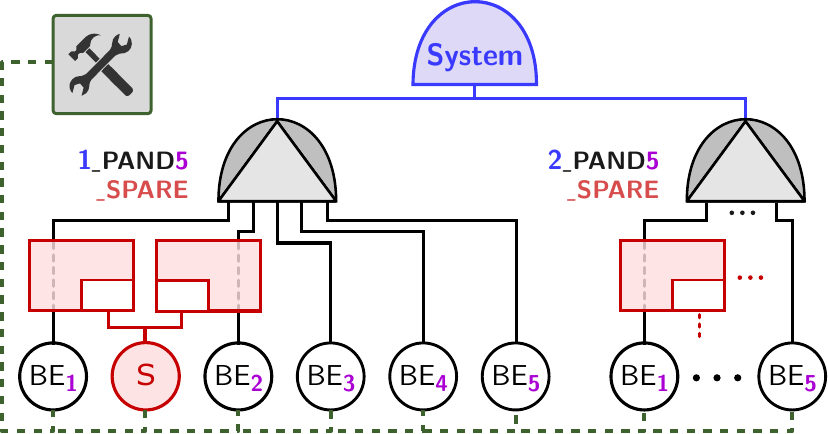}

	\vspace{-.5ex}
	\caption{Experimentation model \protect{\modelDFTlarge}}
	\label{fig:modelLarge}
\end{wrapfigure}
A \emph{\ac{DFT}} is a DAG-structured reliability model that decomposes a main system failure into a combination of partial failures, where the leaves model the failure of basic elements---see \Cref{ex:running-example} and \cite{ruijters2015fault}.
\acp{DFT} with repairs and non-Markovian failure/repair distributions were given formal semantics as \acp{IOSA} in~\cite{MBD20}.
We study time-bounded failure probabilities of two such models:
\begin{itemize*}[itemjoin={{; }}]
\item[(\modelDFTsmall)]
the case study from \cite{dengler:2025:time-sensitive-isplit}, with 4 \PANDgate gates and uniform failure/repair distributions
\item[(\modelDFTlarge)]
the larger system from \Cref{fig:modelLarge}, with \ANDgate, \PANDgate, and \SPAREgate gates, and exponential, normal, and Weibull failure/repair/dormant distributions.
\end{itemize*}
In both models---described in full detail in
\ifthenelse{\boolean{isArXiv}}{%
	\Cref{sec:app:details-experiments}%
}{
	\cite[App.\;E]{arXivVersion}%
}---there is a single \RBOX, making all \acp{BE} and \acp{SBE} interdependent.

\vspace{-2ex}
\subsubsection{Queueing networks.}

\begin{wrapfigure}[6]{r}{0.33\linewidth}
	\vspace{-5ex}
	\centering
	\includegraphics[width=.9\linewidth]{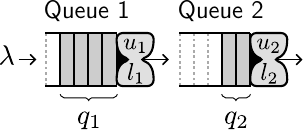}

	\vspace{-.5ex}
	\caption{Tandem queue}
	\label{fig:tandem-queue}
\end{wrapfigure}
A \emph{tandem queue} consists of $n$ queues connected in series, with arrivals (departures) from outside via the first (resp.\ last) queue only.
We experiment on the minimal $n=2$ case from \Cref{fig:tandem-queue} (\modelQueue), with exponentially distributed inter-arrival times at rate $\lambda$, and uniform service times at the $i$-th queue with $l_i<u_i$ bounds.
We study the time-bounded probability of a packet loss in the second queue due to arrivals over a limiting capacity, $q_2\geq C$.

\subsection{Experimental Setting}
\label{sec:experiments:setting}

We study time-bounded events on each model $\modelm\in\{\modelDFTsmall,\modelDFTlarge,\modelQueue\}$:
For the \acp{DFT}, we query the probability $\varphi_{\modelm} = P( \mathop{\modelm}(t)\in\mathit{fail} \mid t\leq T )$ of observing a full system failure in \modelDFTsmall (resp.\ \modelDFTlarge) within $T=1248$ (resp.\ $4500$) time units.
For the tandem queue, we query the overflow probability $\varphi_{\modelm} = P(q_2\geq3\mid t\leq 11)$.

Let \RESrun denote the use of \ac{RES} algorithm $\mathfrak{a}$ with \ac{IFUN} \IFUNcase.
To measure the quality of (our implementation of) \IFUNcase for \ac{ISPLIT} \ac{RES}, we run \RESrun with \textsc{modes} for 20\:min on each model \modelm, producing a 95\% confidence interval $\hat{\varphi}_{\modelm}^{\RESrun}$ whose width ${\left|\hat{\varphi}_{\modelm}^{\RESrun}\right|}$ is a quality metric of \RESrun.
We normalize these values using the relative error $\mathit{re}_{\modelm}^{\RESrun} \doteq {\left(\sfrac{1}{2}\cdot\left|\hat{\varphi}_{\modelm}^{\RESrun}\right|\right)} / {\varphi_{\modelm}}$, where the ground truth $\varphi_{\modelm}$ is approximated via 12-hour-long \ac{CMC} runs: 
\emph{A more efficient \RESrun yields a smaller~$\mathit{re}_{\modelm}^{\RESrun}$.}

\begin{wraptable}[7]{r}{.39\linewidth}
  \vspace{-4.7ex}
  \centering
  \smaller[1.5]
  \def\yep{\ensuremath{\boldsymbol{\checkmark}}\xspace}
  
  \caption{\smaller[.5]\acp{IFUN} tested}
  \label{tab:experiments:IFUNs}
  \vspace{-2ex}
  \begin{tabular}{l@{\:}l@{~}c@{~}c@{~}r}
    \toprule
    \ac{IFUN}  &         & $d_{*G}$ & $d_*^r$ & \textsf{prune}     \\
    \midrule
    \tagncomp  &         &          & \yep    &                    \\
    \tagnmono  & ($d$)   &          & \yep    &                    \\
    \tsenorder & ($d_O$) & \yep     &         & $\{d_{OG}\}$       \\
    \tsentimed & ($d_T$) & \yep     & \yep    & $\{d_{TG},d_{TG}^r\}$ \\
    \bottomrule
  \end{tabular}
\end{wraptable}
Concretely, we use the \ac{RES} algorithms Fixed Effort and \MakeUppercase{restart}~\cite{LLLT09,villen-altamirano:1991:RESTART}, selecting thresholds with lightweight Sequential Monte Carlo~\cite{cerou:2012:sequential-mc}, for efforts 8 and 15 for Fixed Effort and 2 and 3 for \MakeUppercase{restart}, i.e.\ $\mathfrak{a} \in \mathfrak{A}\doteq\{\RESalgo[8]{FE},$ $\RESalgo[15]{FE},\RESalgo[2]{RST}, \RESalgo[3]{RST}\}$~\cite{BDH19}.
We then aggregate all values per $\ac{IFUN}\times\text{model}$ into a single scalar
$\mathit{re}_{\modelm}^{\mathrm{\Phi}}=\mathop{\Gamma}_{\mathfrak{a}\in \mathfrak{A}}\{\mathit{re}_{\modelm}^{\RESrun}\}$
that demonstrates the quality of $\mathrm{\Phi}$ for \ac{ISPLIT}.
We use two aggregation functions:
$\Gamma=\mathrm{median}$, to measure expected quality; and
$\Gamma=\min$, to observe best-case possibilities.
In all cases \emph{a better \ac{IFUN} \IFUNcase yields a smaller~$\mathit{re}_{\modelm}^{\IFUNcase}$.}

We thus use relative errors to compare 13 \acp{IFUN} that \Cref{tab:experiments:IFUNs} maps to the taxonomy from \Cref{fig:taxonomy}:
\begin{itemize*}[label=\textbullet,itemjoin={{; }}]
\item
For the time-agnostic functions \tagnmono (the distance metric $d$) and \tagncomp (counting $q_2$ for \modelQueue, or a compositional \ac{IFUN} for the \acp{DFT}, see
\ifthenelse{\boolean{isArXiv}}{%
	\Cref{sec:app:comp-ifun}%
}{
	\cite[App.\;D]{arXivVersion}%
}
and \cite{BDMS22}), we also test their resampling variants $d_*^r$
\item
For the (time-sensitive) order metrics \tsenorder ($d_O$), we also test the global-timer variant $d_{OG}$, including a variant $\tsenorder_{\mathsf{global\mathop{\scriptscriptstyle\circ}prune}}$
that prunes the runs which provably cannot reach the target in the remaining time
\item
For the (time-sensitive) timed metrics \tsentimed ($d_T$), we test all combinations of resampling, global timer, and pruning
\end{itemize*}.
Since we are evaluating time-bounded queries, we do not test exact global age \acp{IFUN} ($d_{*EG}$ in \Cref{fig:taxonomy}).
We further omit experiments with some equivalent specimens from the taxonomy, e.g.\ $d_O^r$ and $d_{OG}^r$.

\subsection{Experimental Results}
\label{sec:experiments:results}

\begin{figure}[t]
	\centering
	\includegraphics[width=\linewidth]{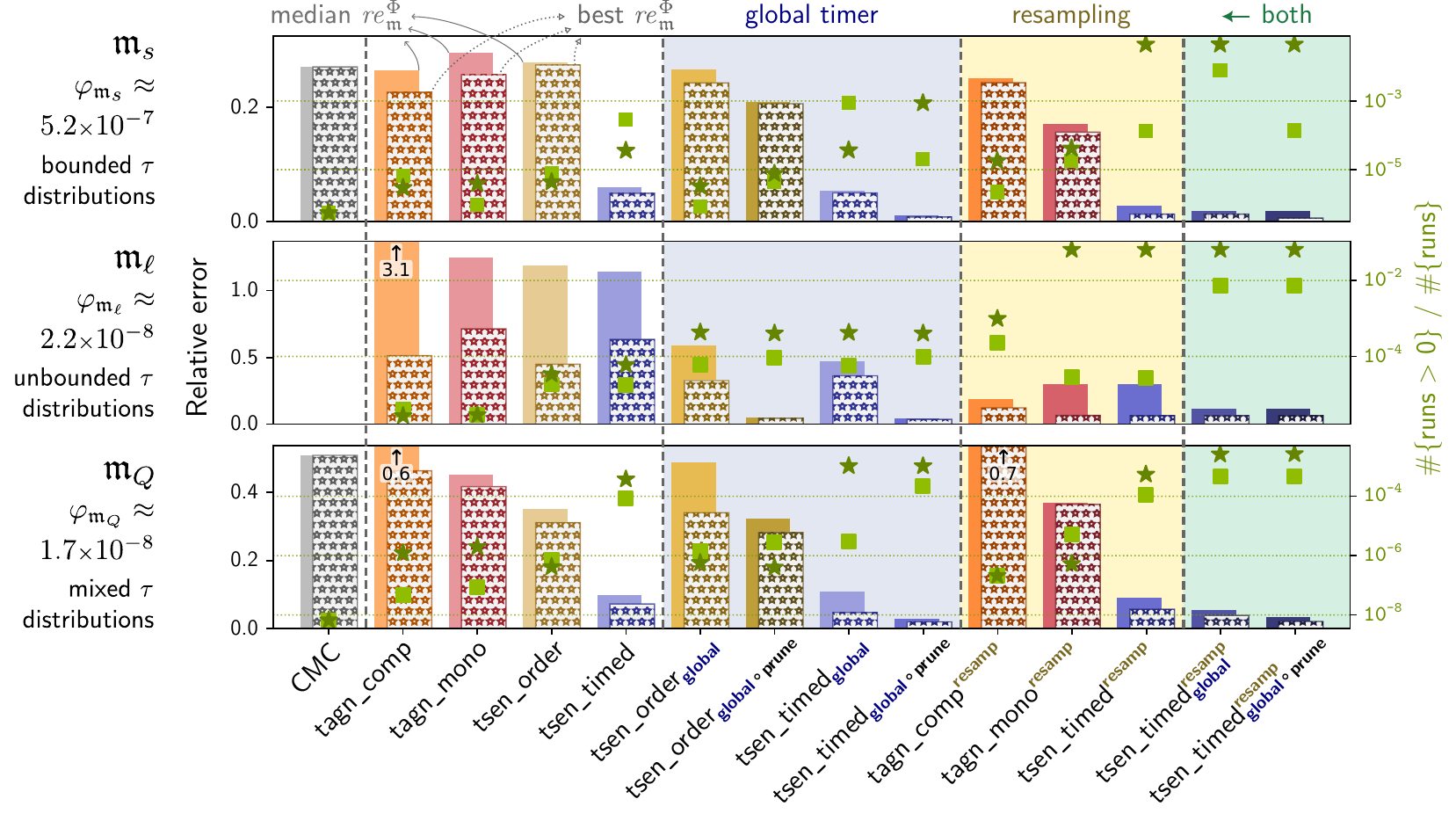}
	\caption{Relative error for \modelDFTsmall, \modelDFTlarge, and \modelQueue (lower is better)}
	\label{fig:experiments:results}
\end{figure}

\Cref{fig:experiments:results} shows the results of our experiments, for approximated ground truths $\varphi_{\modelDFTsmall}\approx5.2\eminus{7}$, $\varphi_{\modelDFTlarge}\approx2.2\eminus{8}$, and $\varphi_{\modelQueue}\approx1.7\eminus{8}$.
The height of a bar indicates the relative error achieved by each \ac{IFUN} \IFUNcase and \ac{CMC}:
Solid plain-color bars in the back stand for $\mathrm{median}$ $\mathit{re}_{\modelm}^{\IFUNcase}$ values, and star-hatched bars in the front stand for $\min$ (``best'') values.
The secondary $y$-axis plots the proportion of ``useful'' runs that observed a rare event, with \textcolor[HTML]{8EBC00}{$\blacksquare$} and \textcolor[HTML]{648500}{\raisebox{.2ex}{$\bigstar$}} respectively for the median and best runs.
For example, for model \modelDFTlarge and \ac{IFUN} $\tsentimed_{\mathsf{global}}$---the most central blue bars in the middle plot in \Cref{fig:experiments:results}---the median and best relative errors were\ 0.47 and 0.36, with resp.\ proportions of useful runs of $5.9\eminus{5}$ and $4.2\eminus{4}$.
\ifthenelse{\boolean{isArXiv}}{%
	\Cref{tab:experiments:results} in \Cref{sec:app:details-experiments}%
}{
	The table in \cite[App.\;E]{arXivVersion}%
}
contains all numerical results used to compile \Cref{fig:experiments:results}.

Since timers in \modelDFTsmall are bounded (uniform distributions),
the impact of time-granularity (\Cref{thm:distance-metric-taxonomy}) in the upper plot is apparent.
Visual comparisons should be done along a single dimension of \Cref{fig:taxonomy}, e.g.\ either $d_O$\:vs.\:$d_T$ (without global timer) or else $d_{OG}$\:vs.\:$d_{TG}$.
For instance, $d_O\preceq d_T$ means that \tsenorder should be less efficient than \tsentimed to estimate $\varphi_{\modelDFTsmall}$, which is evidenced by the higher bars of \tsenorder in the upper plot of \Cref{fig:experiments:results}.
Comparing $\tsenorder_{\mathsf{global}}$ vs.\ $\tsentimed_{\mathsf{global}}$ tells the same story, and can be explained by $d_{OG}\preceq d_{TG}$.
Analogous results can be seen in the lower plot corresponding to \modelQueue, where the service times also follow bounded distributions.
In contrast, the unbounded timers of \modelDFTlarge make the difference due to time granularity much less pronounced.

In cases like \modelDFTlarge, where specific timer constellations must be reached in an asymmetric state space, resampling achieves the biggest performance boost.
This applies to both time-agnostic and -sensitive \acp{IFUN} to a similar degree (in compliance with \Cref{thm:distance-metric-equivalence-order} for models with unbounded timer domains), as evidenced by the lower bars of the $d^r_\ast$ variants for \tagnmono, \tagncomp, and \tsentimed in the middle plot of \Cref{fig:experiments:results}.
Computational overhead also has a big impact in this (largest) model: For both $d_{OG}$ and $d_{TG}$, pruning runs that cannot reach the rare event before the time limit produces narrower \acp{CI}.
Note that this is purely a matter of achieving more runs in the given simulation budget, even though the proportion of useful runs generated is the same with or without pruning.

In contrast to these \acp{DFT}, queueing systems like \modelQueue have a fully connected and symmetrical state space, where resampling has a much lower impact.
However, early pruning significantly increases the number of executed runs and thus yields more accurate estimates, as we observe that $\tsenorder_{\mathsf{global\mathop{\scriptscriptstyle\circ}prune}}$ has lower relative errors than \tsenorder alone, and analogously for \tsentimed.

\myparagraph{On the representativeness of our experiments.}
These results demonstrate many practical implications of the taxonomy in \Cref{sec:taxonomy}, in a reduced but varied set of models.
We highlight that, despite getting results that we can explain with \Cref{sec:taxonomy}, we only have a partial empirical view.
For instance, to compare functions like $\tsenorder_{\mathsf{global\mathop{\scriptscriptstyle\circ}prune}}$ and $\tsentimed^{\mathsf{resamp}}_{\mathsf{global}}$, our numerical results do not allow us to discriminate how much performance gain is due to information depth (\ac{IFUN} granularity) and how much to computational overhead (resampling, pruning).
A more nuanced answer to these practical questions requires a much larger experimental setup, which falls outside the scope of this contribution.

\section{Conclusion}
\label{sec:conclusion}

In this paper, we substantially extended the original concept of time-sensitive \ac{ISPLIT}~\cite{dengler:2025:time-sensitive-isplit}
by considering resampling and information about the global age. We discovered different distance metrics and explored the tradeoff between information depth and computational costs.
We demonstrated the feasibility and effectiveness of our approach for a relevant set of dependability models.

\paragraph{Limitations and future work.}
While the experimental evaluation demonstrated significant performance improvements compared to existing approaches, these improvements have thus far only been established experimentally. Consequently, deriving theoretical performance guarantees under a controlled setting, analogous to previous \ac{ISPLIT} work \cite{lecuyer:2010:asymptotic-robustness-res, glasserman:1999:multilevel-splitting},
remains an important avenue for future research.
Methodologically, instead of fixing a concrete time for every elapsed timer in a step (see \Cref{def:elapsed-time}), one can do this only for a subset of the elapsed timers, effectively obtaining a mixture between \ac{SSC} analysis and simulation.

\renewcommand{\thelstlisting}{\arabic{lstlisting}}

\paragraph{Data availability statement.}
An artifact for experimental reproduction of the results in \Cref{sec:experiments} is publicly available at DOI~\protect{\href{https://doi.org/10.5281/zenodo.19616627}{\color{blue}\texttt{10.5281/zenodo.19616627}}}.

\paragraph{Acknowledgments.} We thank Pedro R. D'Argenio from Universidad Nacional de Córdoba for fruitful discussions about the relationship between time-sensitive \ac{ISPLIT} and resampling.

\begin{acronym}[ABCDEFGHIJK]
    \renewcommand{\acitalic}[1]{#1} 
    \acro{BE}[\textsf{BE}]{basic element}
    \acro{CDF}[CDF]{cumulative distribution function}
    \acro{CI}[CI]{confidence interval}
    \acro{CMC}[CMC]{crude Monte Carlo}
    \acro{DBM}[DBM]{difference bound matrix}
    \acro{DFT}[DFT]{dynamic fault tree}
    \acro{IFUN}[IFUN]{\protect\acitalic{importance function}}
    \acro{IOSA}[IOSA]{input/output stochastic automaton}
    \acro{IS}[IS]{importance sampling}
    \acro{ISPLIT}[ISPLIT]{importance splitting}
    \acro{PDF}[PDF]{probability density function}
    \acro{RES}[RES]{rare event simulation}
    \acro{SBE}[\textsf{SBE}]{spare basic element}
    \acro{SC}[SC]{state class}
    \acro{SMC}[SMC]{statistical model checking}
    \acro{SSC}[SSC]{stochastic state class}
    \acro{LCM}[LCM]{least common multiple}
    \acroplural{CI}[CIs]{confidence intervals}
    \acroplural{DBM}[DBMs]{difference bound matrices}
    \acroplural{IOSA}[IOSA]{input/output stochastic automata}
    \acroplural{SC}[SCs]{state classes}
    \acroplural{SSC}[SSCs]{stochastic state classes}
\end{acronym}

\bibliographystyle{splncs04}
\bibliography{references.bib}

\ifthenelse{\boolean{isArXiv}}{\newpage
\appendix
\renewcommand{\theHsection}{A\arabic{section}}

\section{Backwards Reachability Calculations with State Classes}
\label{sec:app:reachability-state-classes}

We recall the calculations of predecessor \acp{SC} in an \ac{IOSA} from \cite{dengler:2025:time-sensitive-isplit}.

\myparagraph{Target \acp{SC}.} 
For each target location~$l'$, we build the largest possible domain $D'$ for the active timers in~$l'$, considering only whether the target location~$l'$ is reachable from these evaluations of timers, not whether these evaluations of timers are reachable from an initial state.
To this end, we limit each timer $t_i \in T_{l'}$ by its upper bound 
$b_i \in \mathbb{Q}_{\geq 0} \cup \{\infty\}$, obtaining the \ac{SC} $\Sigma' = \langle l', D'\rangle$ with
$D' = \bigtimes_{t_i \in T_{l'}} [0,b_i] \subseteq (\mathbb{Q}_{\geq 0} \cup \{\infty\})^{|T_{l'}|}$.
Note that $D'$ does not introduce dependencies between the involved timers, so $D'$ has a hyper-rectangular shape.
Conversely, we do not limit $t_i$ by its lower bound $a_i \in \mathbb{Q}_{\geq 0}$ given that, depending on the sequence of executed transitions, it may not be \textit{newly} activated in~$l'$ (i.e., after the last incoming transition).

\myparagraph{Predecessor \ac{SC}.} Given an 
\ac{SC} $\Sigma' = \langle l', D'\rangle$ and an incoming transition $l\xtr{\scriptscriptstyle T,a,T'}l'$, we compute the largest
possible \ac{SC} $\Sigma = \langle l, D\rangle$ such that for any state $s$ 
in $\Sigma$, the action $a$ leads to a state $s'$ in $\Sigma'$ (we could say that $\Sigma$ is the \emph{weakest precondition} of $\Sigma'$ given~$a$). 
At each step of this derivation, we normalize the
\ac{DBM} domain of the \ac{SC} being computed using the Floyd-Warshall algorithm~\cite{vicario:2001:tpn-analysis} (as in 
forwards analysis), ensuring the uniqueness of the result and the correct resolution of all transitive restrictions. 
In the following derivation steps, let $\boldsymbol{\tau}$ and $\boldsymbol{\tau}'$ be the vectors of random variables containing the active timers in $l$ and $l'$, respectively:

\begin{enumerate}
    \item \label{en:sc-step1} \emph{Inverse of newly activating:} 
    Let $\boldsymbol{\tau}' := \langle t_2', \ldots, t_{n+m}' \rangle$ and $t_{n+1}', \ldots, t_{n+m}' \in T'$ be the $m$ newly activated timers in $l'$.
    We limit each such timer by its lower bound (as no time elapses from new activation to entrance in $\Sigma'$), obtaining 
    \[
    \boldsymbol{\tau}_a :=
    \langle t_2^a, \ldots, t_n^a, t_{n + 1}^a, \ldots, t_{n+m}^a \rangle
    =
    \boldsymbol{\tau}' \,\mid\, t_{n+1}', \ldots, t_{n+m}' \in T'\text{,}
    \]
    with support
    \[D_a = D' \cap \{ a_i \leq \tau(t_i')\ \forall\, i \in \{n + 1, \ldots, n+m\} \}\text{.}\]
    After normalizing $D_a$, we remove 
    $t_{n+1}^a, \ldots, t_{n+m}^a$
    from $\boldsymbol{\tau}_a$, yielding 
    \[
    \boldsymbol{\tau}_b := 
    \langle t_2^b, \ldots, t_n^b \rangle =
    \langle t_2^a, \ldots, t_n^a \rangle
    \]
    with support:
    \begin{align*}
        D_b &= \{ \langle \tau(t_2^a), \ldots, \tau(t_n^a) \rangle\\
        &\strut\qquad \text{s.t. } \exists\,\splitatcommas{\tau(t_{n+1}^a), \ldots, \tau(t_{n+m}^a)}\\
        &\strut\qquad \text{s.t. } \langle \tau(t_2^a), \ldots, \tau(t_n^a), \tau(t_{n+1}^a), \ldots, \tau(t_{n+m}^a)\rangle \in D_a \}
    \end{align*}
    \item \label{en:sc-step2} \emph{Inverse of time advancement:} 
    We increase the timers by the timer expiring first in $\Sigma$, say $t_1^b$, obtaining
    \[
    \splitatcommas{\boldsymbol{\tau}_c
    := 
    \langle 
    t_1^{c}, \ldots, t_n^{c} \rangle
    = 
    \langle t_1^b, t_2^b + t_1^b, \ldots, t_n^b+t_1^b
    \rangle} 
    \] 
    with support:
    \begin{align*}
        D_c &= \{ \langle \tau(t_1^b), \tau(t_2^b), \ldots, \tau(t_n^b) \rangle\\
        &\strut\qquad \text{s.t. } \tau(t_1^b) \in [0, \infty) \\
        &\strut\qquad \text{s.t. } \langle \tau(t_2^b) - \tau(t_1^b), \ldots, \tau(t_n^b) - \tau(t_1^b)\rangle \in D_b\}
    \end{align*}
    \item \label{en:sc-step3} \emph{Applying upper bounds of timers:} After normalizing $D_c$, we limit each active timer by its upper bound (to guarantee the timer does not exceed the bound due to the time retardation), obtaining 
    \[\boldsymbol{\tau} 
    := \langle t_1, \ldots, t_n\rangle
    = \boldsymbol{\tau}_c \,\mid\ 
    t_i^c \leq b_i 
    \ \forall\, i \in \{1, \ldots, n\}\]
    with the final domain (which is finally normalized, similarly to the previous passages):
    \[D = D_c \cap \{\tau(t_i^c) \leq b_i\ \forall\, i \in \{1, \ldots, n\}\}\]
\end{enumerate}
Note that
there is no need for a conditioning step to guarantee that the 
timer $t_1^b$ expiring first
actually elapses
before the other timers. 
In fact, the non-negativity of the timer evaluations in domain $D_b$, 
i.e.,~$0\leq \tau(t^b_i) - \tau(t^b_1)$
$\forall\,i\in\{2, \ldots, n\}$,
already implies that 
$\tau(t^b_1) \leq \tau(t^b_i)$
$\forall\,i\in\{2, \ldots, n\}$.

\myparagraph{Disabled timers.}
As mentioned in \Cref{sec:background:IOSA}, there are cases occurring in the repairable \ac{DFT} semantics in \cite{MBD20}, although not strictly being part of the original \ac{IOSA} definition, where a set of timers, say $t_2, \ldots, t_{p+1} \in T_d$, is disabled while taking a transition.
In this case, only the elapsing timer was added in step~\ref{en:sc-step2} (inverse of time advancement).
Thus, we need to insert the $p$ disabled timers $t_2, \ldots, t_{p+1}$ after step \ref{en:sc-step2} into the zone $D_c$ while ensuring that these elapse after $t_1$.
Formally, we expand the set of timers $\boldsymbol{\tau}_c = \langle t_1^c, t_{p+2}^c, \ldots, t_n^c\rangle$, obtaining
\[
    \splitatcommas{\boldsymbol{\tau}_d
    := 
    \langle 
    t_1^{d}, t_2^{d}, \ldots, t_n^{d} \rangle
    = 
    \langle t_1^c, t_2, \ldots, t_{p+1}, t_{p+2}^c, \ldots, t_n^c
    \rangle}
\]
with support:
\begin{align*}
    D_d &= \{ \langle \tau(t_1^d), \tau(t_2^d), \ldots, \tau(t_n^d) \rangle\\
        &\strut\qquad \text{s.t. } \forall\,i\in \{2, \ldots, p+1\} \colon \tau(t_i^d) \in [0, \infty) \wedge \tau(t_1^d) \leq \tau(t_i^d) \\
        &\strut\qquad \text{s.t. } \langle \tau(t_1^d), \tau(t_{p+2}^d), \ldots, \tau(t_n^d)\rangle \in D_c\}
\end{align*}
In step \ref{en:sc-step3} (applying upper bounds of timers), we then have to work with $\boldsymbol{\tau}_d$ instead of $\boldsymbol{\tau}_c$ and $D_d$ instead of $D_c$, respectively.

\section{Resampling Implementation}
\label{sec:app:resampling-implementation}

As each timer can be resampled individually, we solely have to answer how to sample from a truncated timer distribution.
Formally, for a timer $t_i$, let $X_i$ be the random variable distributed according to the \ac{CDF} $F_i$. Analogously, let $Y_i = (X_i \mid X_i \geq \tau_{\textit{elp}}(t_i)) - \tau_{\textit{elp}}(t_i)$ be the random variable that is based on $X_i$, but truncates all values below $\tau_{\textit{elp}}(t_i)$ and subtracts the elapsed time. Sampling of $Y_i$ corresponds to resampling the value of timer $t_i$.
Exponential distributions can be resampled without adaptations due to the memoryless property as it holds here that $X_i = Y_i$. Furthermore, for the case $\tau_{\textit{elp}}(t_i) \leq a_i$, it holds $(X_i \mid X_i \geq \tau_{\textit{elp}}(t_i)) = X_i$, so we can sample from the distribution $X_i$ and subtract the elapsed time $\tau_{\textit{elp}}(t_i)$ to obtain a sample from $Y_i$. Otherwise, as in \cite{BDMS22}, we distinguish two cases:
\begin{itemize}
	\item For distributions where the inverse \ac{CDF} $F_i^{-1}\colon [0,1]\rightarrow \mathbb{R}_{\geq0}$ exists (e.g., uniform or Rayleigh distribution), we can compute a sample for $(X_i \mid X_i \geq \tau_{\textit{elp}}(t_i))$ via inverse transform sampling as follows: First, create a random value $u\sim \text{Unif}[0,1]$, and then create a sample for $(X_i \mid X_i \geq \tau_{\textit{elp}}(t_i))$ by evaluating $F^{-1}(u + (1 - u)\cdot F(\tau_{\textit{elp}}(t_i)))$.
	\item Otherwise (e.g., gamma distribution), we evaluate $(X_i \mid X_i \geq \tau_{\textit{elp}}(t_i))$ by rejection sampling, i.e., sample values of $X_i$ until one is at least $\tau_{\textit{elp}}(t_i)$ and then calculate $Y_i$. To limit the computational complexity beforehand, we fix a maximum number of sampling attempts (in our case $k = 20$) and resort to the old value $X_i$ if no suitable sample has been found.

    Note that limiting the number of sampling attempts to $k$ retains the old sample with a probability of $F(\tau_{\textit{elp}}(t_i))^k$ instead of creating a new one. Although the old sample is used, the timer evaluation still represents an unbiased sample from the conditioned distribution $(X_i \mid X_i \geq \tau_{\textit{elp}}(t_i))$ and does ultimately yield unbiased simulation results:
    As a time of $\tau_{\textit{elp}}(t_i)$ has elapsed for timer $t_i$, the old sampled timer value is distributed according to $(X_i \mid X_i \geq \tau_{\textit{elp}}(t_i))$. Furthermore, the probability that the old sample is retained is independent of the old sampled timer value.
\end{itemize}

\section{Omitted Proofs}
\label{sec:app:omitted-proofs}

\begin{proof}[\Cref{thm:distance-metric-equivalence-order}]
	Remember that $\omega(l)$ (distance of each location to the target) is calculated by backwards reachability search from the target states. Consider a state $s = \langle l, \tau\rangle$.
    First, we show that, for each state $s$, the distance $d(s)$ is an upper bound of the ordered distance with resampling $d_O^r(s)$. Specifically, let $\omega(l)=n$. By definition of backwards reachability, there exists a shortest sequence of actions and locations $l_n \stackrel{a_n}{\rightarrow} l_{n-1} \stackrel{a_{n-1}}{\rightarrow} l_{n-2} \stackrel{a_{n-2}}{\rightarrow} \ldots \stackrel{a_1}{\rightarrow} l_0$ such that $l_n=l$, $l_0$ is a target location, and $\omega(l_i)=i$ for all $i\in \{0, \ldots, n\}$. Then, we consider the corresponding sequence of elementary \acp{SC} $\Sigma_n \stackrel{a_n}{\rightarrow} \Sigma_{n-1} \stackrel{a_{n-1}}{\rightarrow} \Sigma_{n-2} \stackrel{a_{n-2}}{\rightarrow} \ldots \stackrel{a_1}{\rightarrow} \Sigma_0$ with $\omega_O(\Sigma_i) = i$ for $i\in \{0, \ldots, n\}$. Given that all timers are unbounded, no timer is disabled due to clock bounds. Consequently, a valid execution sequence matching the sequence of actions and locations is guaranteed to exist, implying that $\omega_O(\Sigma_n)\leq n$. Since $\omega(l)$ represents the shortest distance unconstrained by timer values, $\omega_O(\Sigma_n)$ cannot be strictly lower than $n$, implying that there exists an elementary \ac{SC} $\Sigma_{\sigma}$ s.t. $\omega(l) = \omega_O(\Sigma_{\sigma})$.
    
    Second, via resampling, independent of the elapsed time $\tau_{\textit{elp}}$, the induced \ac{SC} $\Sigma_{s_{\textit{elp}}}$ always intersects with any of the $|T_l|!$ elementary \acp{SC} of location $l$. For reference, consider any timer permutation $\sigma$. Define $\tau = \{\sigma(t_1) \mapsto 1, \sigma(t_2) \mapsto 2, \ldots, \sigma(t_n) \mapsto n \}$. Then, $\tau$ is both contained in $\Sigma_{\sigma}$ (the order of the timer values is trivially given) and $\Sigma_{s_{\textit{elp}}}$ (the lower bound is always zero, and no upper bound exists as we are solely working with unbounded timers), so $\Sigma_{\sigma}$ and $\Sigma_{s_{\textit{elp}}}$ intersect.
    
    Combining both arguments, it follows directly that $d(s) = d^r_O(s)$ for $s\in\states$.
\end{proof}

\begin{proof}[\Cref{thm:distance-metric-taxonomy}]
    We prove the statements individually:
    \begin{enumerate}
        \item \label{en:thm2p-1} \emph{$d \equiv d^r_{O} \equiv d^r_{OG} \equiv d^r_{OEG}$ (Equivalence for order with resampling)}\\[0.5em]
        $d \equiv d^r_{O}$ was already proven by \Cref{thm:distance-metric-equivalence-order}. Consider $d^r_{O} \equiv d^r_{OG}$. For a state $s$ and distance metric $d^r_O$, we have a look at the shortest possible sequence of elementary \acp{SC} to the target $\Sigma_n \stackrel{a_n}{\rightarrow} \Sigma_{n-1} \stackrel{a_{n-1}}{\rightarrow} \Sigma_{n-2} \stackrel{a_{n-2}}{\rightarrow} \ldots \stackrel{a_1}{\rightarrow} \Sigma_0$ with $\omega_O(\Sigma_i) = i$ for $i\in \{0, \ldots, n\}$ and delays $\Delta_i \in \mathbb{R}_{\geq{0}}$ for $i\in \{1, \ldots, n\}$. The same sequence is possible for $d^r_{OG}$ independent of the remaining (positive) time budget $\tau(t_\textit{age})$: Due to resampling and the unbounded distributions, one can scale the values of the sampled timers s.t. the time to execute all transitions is shorter than the remaining time budget, i.e., $\sum_{i=1}^{n} \Delta_i \leq \tau(t_\textit{age})$, while preserving the order of timers in each state. This makes the same sequence possible as for $d^r_{O}$. Similarly to the argument made in \Cref{thm:distance-metric-equivalence-order}, no shorter sequence exists, proving $d^r_{O} \equiv d^r_{OG}$. To prove $d^r_{OG} \equiv d^r_{OEG}$, one recognizes that a simulation run can always stay in a target state infinitely long (as no upper bounds on the timers exist), so it only matters that a target state is reached within the given time limit (and so $d^r_{OG}$ and $d^r_{OEG}$ coincide here).
        \vspace{0.5em}
        \item \label{en:thm2p-2} \emph{$d_O \preceq d_{OG} \preceq d_{OEG}$, $d_T \preceq d_{TG} \preceq d_{TEG}$, $d^r_T \preceq d^r_{TG} \preceq d^r_{TEG}$ (Age inclusion)}
        \vspace{0.2em}
        \item \label{en:thm2p-3} \emph{$d \preceq d_O$, $d_O \preceq d_T$, $d_{OG} \preceq d_{TG}$, $d_{OEG} \preceq d_{TEG}$, $d \preceq d^r_T$ (Time granularity)}\\[0.5em]
        All of these statements from \ref{en:thm2p-2}. and \ref{en:thm2p-3}. of the type $d_A \preceq d_B$ follow directly from the fact that the left side $d_A$ applies strictly fewer restrictions than the right side $d_B$. Thus, every sequence of states to the target $s_n \stackrel{a_n}{\rightarrow} s_{n-1} \stackrel{a_{n-1}}{\rightarrow} s_{n-2} \stackrel{a_{n-2}}{\rightarrow} \ldots \stackrel{a_1}{\rightarrow} s_0$ reflected by $d_A$ cannot be done in fewer steps when resorting to the restrictions denoted in the distance metric $d_B$.
        For $d \preceq d^r_T$ note that, as $d\equiv d^r_O$ (see \ref{en:thm2p-1}.), this statement is equivalent to $d^r_O \preceq d^r_T$ and is handled by the same argument (strictly fewer restrictions) as the other statements.
        \vspace{0.5em}
        \item \label{en:thm2p-4} \emph{$d^r_T \preceq d_T$, $d^r_{TG} \preceq d_{TG}$, $d^r_{TEG} \preceq d_{TEG}$ (Resampling yes/no)}\\[0.5em]
        We recall that \Cref{def:time-sensitive-dist-resampling} defines $d^{r}_{T}(s) = \min_{s'\in \Sigma_{s_\textit{elp}}} d_{T}(s')$. Thus, it holds for every state $s \in \Sigma_{s_\textit{elp}}$ that $d^{r}_{T}(s) \leq d_{T}(s)$. As this is true for every elapsed time information $s_\textit{elp}$, $d^{r}_{T}(s) \leq d_{T}(s)$ for all $s\in \states$ and thus $d^r_T \preceq d_T$. The proof is analogous for $d^r_{TG} \preceq d_{TG}$ and $d^r_{TEG} \preceq d_{TEG}$.
    \end{enumerate}
\end{proof}

\section{Time-Agnostic Compositional Importance Function Definition for Repairable Dynamic Fault Trees}
\label{sec:app:comp-ifun}

\begin{table}
    \centering
    \caption{\ac{IFUN} definition for repairable \acp{DFT} \cite{BDMS22}}
    \label{tab:comp-ifun-def}
%
\begingroup
\renewcommand{\arraystretch}{1.8}
\def\treei{\ensuremath{\Tree_i}}
\def\maxI[#1]{\ensuremath{\max^{\IFUN}_{#1}}}
\def\lcm{\ensuremath{\mathrm{lcm}}}
\def\ord{\ensuremath{\mathit{ord}}}
\def\bitsmaller[#1]{\scalebox{1.15}{$#1$}}
%
\begin{tabular}{>{\centering}m{2cm}@{\qquad}>{\centering\arraybackslash}m{6.2cm}m{1.2cm}}
	\toprule
	\addlinespace[-0.7ex]  
	\type[\nodev] & \multicolumn{1}{c}{$\IFUN_\nodev(\xbf)$} & \\
	\midrule
	\addlinespace[-.7ex]  
	\BE, \SBE &
		\raisebox{.7ex}{$\displaystyle \zbf_\nodev$}\\
	\addlinespace[-1ex]  
	\ANDgate &
		$\lcm_\nodev\cdot \bitsmaller[\sum_{w\in\child[\nodev]}
		\frac{\IFUN_w(\xbf)}{\maxI[w]}]$\\
	\ORgate &
		$\displaystyle \lcm_\nodev\cdot \max_{w\in\child[\nodev]}
		\left\{ \bitsmaller[\frac{\IFUN_w(\xbf)}{\maxI[w]}] \right\}$\\
	$\VOTgate_k$ &
		$\displaystyle \lcm_\nodev\cdot \max_{W\subseteq\child[\nodev], |W|=k}
		\left\{ \bitsmaller[\sum_{\nodew\in W}
		\frac{\IFUN_\nodew(\xbf)}{\maxI[\nodew]}] \right\}$\\
	\addlinespace[.4ex]  
	\SPAREgate &
		$\displaystyle \lcm_\nodev\cdot \max\Big(
			\bitsmaller[\sum_{w\in\child[\nodev]}
			\frac{\IFUN_w(\xbf)}{\maxI[w]}] ~,~ \zbf_\nodev\cdot m \Big)$\\
	\PANDgate &
		$\displaystyle \lcm_\nodev\cdot \max\Big(
			 \bitsmaller[\frac{\IFUN_l(\xbf)}{\maxI[l]}]
			+ \ord\:\bitsmaller[\frac{\IFUN_r(\xbf)}{\maxI[r]}]
                        ~,~ \zbf_\nodev\cdot 2 \Big)$\\[-2ex]
		& \multicolumn{2}{l}{\smaller
			where~$\ord=1$ if $\xbf_\nodev\in\{1,4\}$ and $\ord=-1$ otherwise}\\
         \multicolumn{3}{c}{with \ $\max^{\IFUN}_{\nodev}=\max_{\xbf\in\states}\IFUN_\nodev(\xbf)$ \ and \ $\mathrm{lcm}_\nodev=\mathrm{lcm}\left\{
	\max^{\IFUN}_\nodew \,\middle|\, \nodew\in\child[\nodev]
\right\}$}\\
	\bottomrule
\end{tabular}
\endgroup

\end{table}
\Cref{tab:comp-ifun-def} depicts the compositional \ac{IFUN} definition for repairable \acp{DFT} according to \cite{BDMS22}.
It specifies for each gate a local \ac{IFUN}; for \acp{BE} and \acp{SBE} this is either $0$ or $1$, depending on whether the \ac{BE}/\ac{SBE} is working or not, while the other gates take the values of their children into account.
$\zbf_v$ denotes if a gate has failed, while $\xbf_v$ refers to the internal state of the gate (for more details, see \cite{BDMS22}).
The weights of the children are balanced by their maximum importance; the \ac{LCM} is used to attain solely integer levels despite the balancing.
The local \ac{IFUN} of the topmost gate resembles the global \ac{IFUN}.
During simulation, the global \ac{IFUN} is evaluated by stepwise applying these rules until arriving at \acp{BE} or \acp{SBE}.

\section{More Details about the Experiments}
\label{sec:app:details-experiments}

\subsection{Theoretical Models}
\label{sec:app:details-experiments:theory}

\begin{figure}[ht]
  \centering
  \begin{subfigure}{.22\linewidth}
    \centering
    \includegraphics[width=\linewidth]{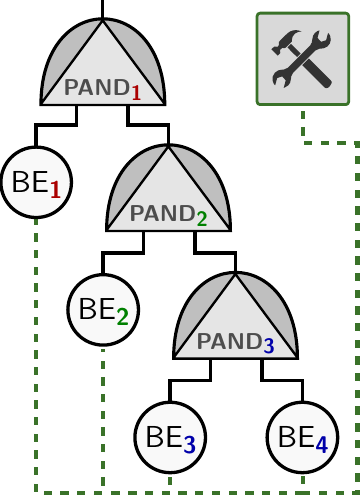}
    \caption{\protect{\modelDFTsmall} \cite{dengler:2025:time-sensitive-isplit}}
    \label{fig:experiments:models:modelSmall}
  \end{subfigure}
  \hfill
  \begin{subfigure}{.73\linewidth}
    \centering
    \includegraphics[width=\linewidth]{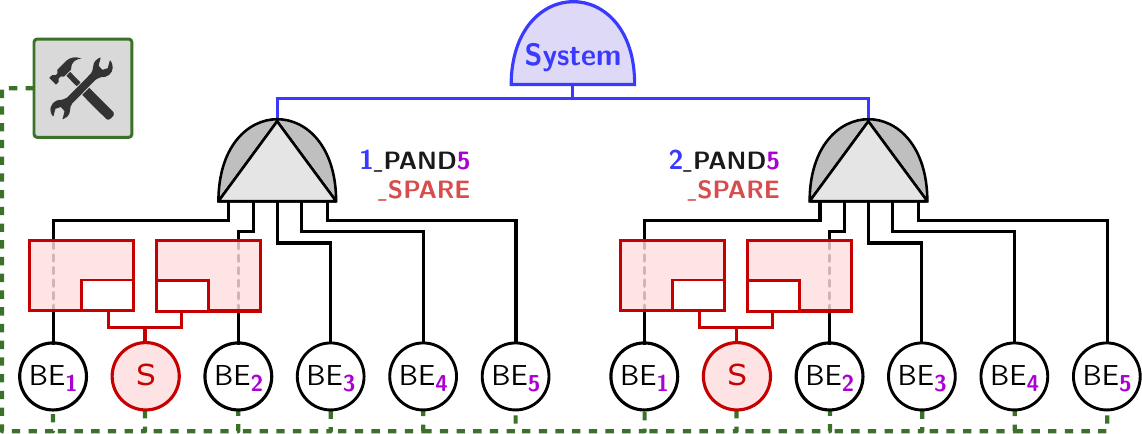}
    \caption{\protect{\modelDFTlarge}, where $n$-ary \PANDgate{s} are actually cascaded as in \Cref{fig:experiments:models:modelSmall}}
    \label{fig:experiments:models:modelLarge}
  \end{subfigure}
  \caption{\ac{DFT} experimentation models from \Cref{sec:experiments}}
  \label{fig:experiments:models}
\end{figure}

This
section provides more details about the models
used for experiments in \Cref{sec:experiments}.
For \acp{DFT} we study the time-bounded probability $\varphi_{\modelm} = P( \mathop{\modelm}(t)\in\mathit{fail} \mid t\leq T )$ of observing a full system failure, using $T=1248$ for \modelDFTsmall and $T=4500$ for \modelDFTlarge.
The corresponding ground truth failure probabilities---approximated with 12-hour-long 
\ac{CMC} runs---are $\varphi_{\modelDFTsmall}\approx5.2\eminus{7}$ and $\varphi_{\modelDFTlarge}\approx2.2\eminus{8}$.

All timers that govern failures and repairs in the \acp{BE} of \modelDFTsmall have (bounded) uniform distributions---the concrete bounds can be found in \Cref{code:modelSmall} and \cite{dengler:2025:time-sensitive-isplit}.%
\begingroup
  \colorlet{colfail}{violet!50!magenta!55!black}
    \colorlet{colrepair}{green!50!teal!30!black}
      \colorlet{coldormant}{red!35!black}
In contrast, \modelDFTlarge is sketched in \Cref{fig:experiments:models:modelLarge}, and has basic elements whose
  \textcolor{colfail}{failure},
    \textcolor{colrepair}{repair}, and
      \textcolor{coldormant}{dormant-failure}
events obey
  \textcolor{colfail}{$\mathop{\mathit{exponential}}(\distparam{\lambda}{2\cdot 10^{-4}})$},
    \textcolor{colrepair}{$\mathop{\mathit{normal}}(\distparam{\mu}{500},\distparam{\sigma}{50})$}, and
      \textcolor{coldormant}{$\mathop{\mathit{Weibull}}(\distparam{\alpha}{5.6}, \distparam{\beta}{10^4})$}
distributions in time, resp.
Moreover, all \acp{BE} and \acp{SBE} are repaired by a single \RBOX, which makes all components dependent on each other.
Finally, note that while $n$-ary \PANDgate{s} are represented compactly in \Cref{fig:modelLarge,fig:experiments:models:modelLarge}, they are implemented as right-cascading binary \PANDgate{s} equivalently to what is shown in \Cref{fig:experiments:models:modelSmall}, i.e.\ $\PANDgate(a,b,c)$ becomes $\PANDgate(a,\PANDgate(b,c))$.
Both models are described in \Cref{sec:app:details-experiments:code} in a domain-specific language for \acp{DFT}.
\endgroup

\smallskip
Besides fault trees, we experiment with a classic \ac{RES} benchmark from the domain of queueing networks: the two-station tandem queue.
We study the time-bounded probability $\varphi_{\modelQueue} = P(q_2\geq3\mid t\leq 11)$ of observing an overflow in the second station, whose approximated ground truth value is $\varphi_{\modelQueue}\approx1.7\eminus{8}$.

The tandem queue under consideration has mixed timers: External arrivals to the first station are governed by an (unbounded) $\mathop{\mathit{exponential}}(\distparam{\lambda}{0.17})$ distribution, while service times at the stations are governed by bounded distributions: $\mathit{uniform}(\distparam{l_1}{2},\distparam{u_1}{3})$ in the first station and $\mathit{uniform}(\distparam{l_2}{1},\distparam{u_2}{3})$ in the second.
These timers were chosen such that the loads at subsequent stations are decremental.
Together with the tight time bound ($t\leq 11$), this makes the probability $\varphi_{\modelQueue}$ a rare event that becomes exponentially more rare with the limiting capacity $C$ of the second queue whose overflow we are interested in measuring.
The \textsc{Modest} model used in our experiments is introduced as \Cref{code:modelTandemQueueSA} in \Cref{sec:app:details-experiments:code}.

\subsection{Models in Kepler/Modest Syntax}
\label{sec:app:details-experiments:code}
Below we provide the description of both \ac{DFT} models in Kepler~\cite{BDMS22}, an adaptation of the well-known Galileo input format for fault trees \cite{sullivan:1999:galileo}.
Kepler allows specifying general (i.e.\ non-Markovian) distributions for failure, repair, and dormant-failure events of \acp{BE} and \acp{SBE}.

\def\dist{\raisebox{-.2ex}{\textasciitilde}}

\begin{lstlisting}[%
    style=Kepler,
	basicstyle=\scriptsize\ttfamily,
    caption={Kepler syntax for \modelDFTsmall in \Cref{fig:experiments:models:modelSmall} \cite{dengler:2025:time-sensitive-isplit}},
    captionpos=b,
    label=code:modelSmall,
    abovecaptionskip=\medskipamount,
]
toplevel "PAND1"; 
"PAND1" pand "BE1" "PAND2";
"PAND2" pand "BE2" "PAND3";
"PAND3" pand "BE3" "BE4";
"BE1" fail`\dist`uniform(1198,1218) repair`\dist`uniform(10,15);
"BE2" fail`\dist`uniform(530,595)   repair`\dist`uniform(10,45);
"BE3" fail`\dist`uniform(385,465)   repair`\dist`uniform(10,45);
"BE4" fail`\dist`uniform(1105,1205) repair`\dist`uniform(10,15);
"RBOX" rbox prio "BE1" "BE2" "BE3" "BE4";
\end{lstlisting}

\begin{lstlisting}[%
    style=Kepler,
	basicstyle=\scriptsize\ttfamily,
    caption={Kepler syntax for \modelDFTlarge in \Cref{fig:experiments:models:modelLarge}},
    captionpos=b,
    label=code:modelLarge,
    abovecaptionskip=\medskipamount,
]
toplevel "System";
"System" and "G1_PAND5_a" "G2_PAND5_a";
"G1_PAND5_a" pand "G1_SPARE_a" "G1_PAND5_b";
"G1_PAND5_b" pand "G1_SPARE_b" "G1_PAND5_c";
"G1_PAND5_c" pand "G1_BE_c" "G1_PAND5_d";
"G1_PAND5_d" pand "G1_BE_d" "G1_BE_e";
"G1_SPARE_a" spare "G1_BE_a" "G1_SBE";
"G1_SPARE_b" spare "G1_BE_b" "G1_SBE";
"G1_BE_a" fail`\dist`exponential(^2e-4^) repair`\dist`normal(^500^,^50^);
"G1_BE_b" fail`\dist`exponential(^2e-4^) repair`\dist`normal(^500^,^50^);
"G1_BE_c" fail`\dist`exponential(^2e-4^) repair`\dist`normal(^500^,^50^);
"G1_BE_d" fail`\dist`exponential(^2e-4^) repair`\dist`normal(^500^,^50^);
"G1_BE_e" fail`\dist`exponential(^2e-4^) repair`\dist`normal(^500^,^50^);
"G1_SBE"  fail`\dist`exponential(^2e-4^) repair`\dist`normal(^500^,^50^) dorm`\dist`weibull(^5.6^,^1e4^);
"G2_PAND5_a" pand "G2_SPARE_a" "G2_PAND5_b";
"G2_PAND5_b" pand "G2_SPARE_b" "G2_PAND5_c";
"G2_PAND5_c" pand "G2_BE_c" "G2_PAND5_d";
"G2_PAND5_d" pand "G2_BE_d" "G2_BE_e";
"G2_SPARE_a" spare "G2_BE_a" "G2_SBE";
"G2_SPARE_b" spare "G2_BE_b" "G2_SBE";
"G2_BE_a" fail`\dist`exponential(^2e-4^) repair`\dist`normal(^500^,^50^);
"G2_BE_b" fail`\dist`exponential(^2e-4^) repair`\dist`normal(^500^,^50^);
"G2_BE_c" fail`\dist`exponential(^2e-4^) repair`\dist`normal(^500^,^50^);
"G2_BE_d" fail`\dist`exponential(^2e-4^) repair`\dist`normal(^500^,^50^);
"G2_BE_e" fail`\dist`exponential(^2e-4^) repair`\dist`normal(^500^,^50^);
"G2_SBE"  fail`\dist`exponential(^2e-4^) repair`\dist`normal(^500^,^50^) dorm`\dist`weibull(^5.6^,^1e4^);
"RBOX" rbox prio "G2_SBE" "G2_BE_e" "G2_BE_d" "G2_BE_c" "G2_BE_b" "G2_BE_a"
                 "G1_SBE" "G1_BE_e" "G1_BE_d" "G1_BE_c" "G1_BE_b" "G1_BE_a";

\end{lstlisting}

\noindent
Furthermore, we provide in the following the description of the tandem queue in the \textsc{Modest} model description language \cite{HH14}:
\begin{lstlisting}[%
    style=Modest,
    numbers=left,
    caption={\textsc{Modest} syntax for \modelQueue in \Cref{fig:tandem-queue}},
    captionpos=b,
    label=code:modelTandemQueueSA,
    abovecaptionskip=\medskipamount,
]
// Actions
action arrival, queue1_service, queue2_service, overflow;

// Constants
const int C = 3; // Capacity of each queue
const int B = 11; // Time bound
const real l = 0.17; // Arrival rate

// Bounds for service distribution
const real lowb1 = 2.0, uppb1 = 3.0; // of first server
const real lowb2 = 1.0, uppb2 = 3.0; // of second server

// Discrete variables and timers
int(0..C) §${\color{teal!67!black}q_1}$§, §${\color{teal!67!black}q_2}$§;
bool failed = false;
timer t_arrivals = Exponential(l), t_server_1, t_server_2;

// Ad§{\color{green!33!black}-}§hoc importance function and time§{\color{green!33!black}-}§bounded failure property
property Importance = §${\color{teal!67!black}q_2}$§;
property TimeBounded = Pmax(<>[T<=B] (failed));

// Main loop
do {
// Logic for arrivals
:: when(t_arrivals <= 0) invariant(t_arrivals >= 0) arrival
   {= t_arrivals = Exponential(l), §${\color{teal!67!black}q_1}$§ = min(§${\color{teal!67!black}q_1}$§ + 1, C),
      t_server_1 = (§${\color{teal!67!black}q_1}$§ == 0) ? Uniform(lowb1,uppb1) : t_server_1 =}
// Logic for server 1 (without triggering an overflow)
:: when(§${\color{teal!67!black}q_1}$§ > 0 && §${\color{teal!67!black}q_2}$§ < C - 1 && t_server_1 <= 0)
   invariant(§${\color{teal!67!black}q_1}$§ == 0 || §${\color{teal!67!black}q_2}$§ >= C - 1 || t_server_1 >= 0) queue1_service
   {= §${\color{teal!67!black}q_1}$§--, §${\color{teal!67!black}q_2}$§ = min(§${\color{teal!67!black}q_2}$§ + 1, C),
      t_server_1 = (§${\color{teal!67!black}q_1}$§ >  1) ? Uniform(lowb1,uppb1) : t_server_1,
      t_server_2 = (§${\color{teal!67!black}q_2}$§ == 0) ? Uniform(lowb2,uppb2) : t_server_2 =}
// Logic for server 2
:: when(§${\color{teal!67!black}q_2}$§ > 0 && t_server_2 <= 0)
   invariant(§${\color{teal!67!black}q_2}$§ == 0 || t_server_2 >= 0) queue2_service
   {= §${\color{teal!67!black}q_2}$§--, t_server_2 = §${\color{teal!67!black}q_2}$§ > 1 ? Uniform(lowb2,uppb2) : t_server_2 =}
// Check for overflow (triggered by server 1 when server 2 is full)
:: when(§${\color{teal!67!black}q_1}$§ > 0 && §${\color{teal!67!black}q_2}$§ >= C - 1 && t_server_1 <= 0)
   invariant(§${\color{teal!67!black}q_1}$§ == 0 || §${\color{teal!67!black}q_2}$§ < C - 1 || t_server_1 >= 0)
   overflow {= failed = true =}; Overflow()
}

// Failure location: We stop here
process Overflow()
{
    stop
}
\end{lstlisting}

\subsection{Importance Functions}
\Cref{code:IFUNsSmall}/\Cref{code:IFUNsLarge} show the \acp{IFUN} used for \modelDFTsmall/\modelDFTlarge, matching the Kepler syntax identifiers from \Cref{code:modelSmall}/\Cref{code:modelLarge} and their corresponding \ac{IOSA} translations implemented in the \textsc{modes} tool.

For \modelDFTsmall, we limit the maximum expansion depth for each distance metric to $12$. However, this still ensures that the relevant state space is covered (see detailed discussion in \Cref{app:construction-times-IFUNs}).
For $\modelDFTlarge$, we notice that the system consists of two subtrees that are solely connected by a repair box. To limit the explored state space, we construct for each subtree a local \ac{IFUN} and sum up their individual \acp{IFUN}.
Note that the maximum expansion depth of $17$ for the distance metrics in $\modelDFTlarge$ is chosen large enough such that the complete \ac{SC} graphs are constructed.

\begin{lstlisting}[%
    style=IFUNModest,
	breaklines,
	basicstyle=\scriptsize\ttfamily,
    caption={\acp{IFUN} used to run \ac{RES} in model \modelDFTsmall (\Cref{code:modelSmall})},
    captionpos=b,
    label=code:IFUNsSmall,
    abovecaptionskip=\medskipamount,
]
"tagn_comp"  // IFUN time-agnostic compositional
  max(((4*E([BE1:[BE1__failure,BE1__repairing]]))+(if E([PAND1:[PAND1__first_failed,PAND1__gate_failure,PAND1__gate_repaired_first_fail,PAND1__failed_first_up,PAND1__failed_first_fail]]) then max(((2*E([BE2:[BE2__failure,BE2__repairing]]))+(if E([PAND2:[PAND2__first_failed,PAND2__gate_failure,PAND2__gate_repaired_first_fail,PAND2__failed_first_up,PAND2__failed_first_fail]]) then max((E([BE3:[BE3__failure,BE3__repairing]])+(if E([PAND3:[PAND3__first_failed,PAND3__gate_failure,PAND3__gate_repaired_first_fail,PAND3__failed_first_up,PAND3__failed_first_fail]]) then E([BE4:[BE4__failure,BE4__repairing]]) else -E([BE4:[BE4__failure,BE4__repairing]]))), if E([PAND3:[PAND3__failed_first_up,PAND3__failed_first_fail]]) then 2 else 0) else -max((E([BE3:[BE3__failure,BE3__repairing]])+(if E([PAND3:[PAND3__first_failed,PAND3__gate_failure,PAND3__gate_repaired_first_fail,PAND3__failed_first_up,PAND3__failed_first_fail]]) then E([BE4:[BE4__failure,BE4__repairing]]) else -E([BE4:[BE4__failure,BE4__repairing]]))), if E([PAND3:[PAND3__failed_first_up,PAND3__failed_first_fail]]) then 2 else 0))), if E([PAND2:[PAND2__failed_first_up,PAND2__failed_first_fail]]) then 4 else 0) else -max(((2*E([BE2:[BE2__failure,BE2__repairing]]))+(if E([PAND2:[PAND2__first_failed,PAND2__gate_failure,PAND2__gate_repaired_first_fail,PAND2__failed_first_up,PAND2__failed_first_fail]]) then max((E([BE3:[BE3__failure,BE3__repairing]])+(if E([PAND3:[PAND3__first_failed,PAND3__gate_failure,PAND3__gate_repaired_first_fail,PAND3__failed_first_up,PAND3__failed_first_fail]]) then E([BE4:[BE4__failure,BE4__repairing]]) else -E([BE4:[BE4__failure,BE4__repairing]]))), if E([PAND3:[PAND3__failed_first_up,PAND3__failed_first_fail]]) then 2 else 0) else -max((E([BE3:[BE3__failure,BE3__repairing]])+(if E([PAND3:[PAND3__first_failed,PAND3__gate_failure,PAND3__gate_repaired_first_fail,PAND3__failed_first_up,PAND3__failed_first_fail]]) then E([BE4:[BE4__failure,BE4__repairing]]) else -E([BE4:[BE4__failure,BE4__repairing]]))), if E([PAND3:[PAND3__failed_first_up,PAND3__failed_first_fail]]) then 2 else 0))), if E([PAND2:[PAND2__failed_first_up,PAND2__failed_first_fail]]) then 4 else 0))), if E([PAND1:[PAND1__failed_first_up,PAND1__failed_first_fail]]) then 8 else 0)
"tagn_mono"      // IFUN time-agnostic monolithic
  12-D([BE3,PAND1,RBOX,BE2,PAND3,PAND2,BE1,BE4],[PAND1:[PAND1__failed_first_fail]],12)
"tsen_order"     // IFUN time-sensitive order
  12-OD([BE3,PAND1,RBOX,BE2,PAND3,PAND2,BE1,BE4],[PAND1:[PAND1__failed_first_fail]],12)
"tsen_order`${}_{\mathtt{global}\!}$`" // IFUN time-sensitive order global
  12-OGD([BE3,PAND1,RBOX,BE2,PAND3,PAND2,BE1,BE4],[PAND1:[PAND1__failed_first_fail]],12)
"tsen_timed"     // IFUN time-sensitive timed
  12-TD([BE3,PAND1,RBOX,BE2,PAND3,PAND2,BE1,BE4],[PAND1:[PAND1__failed_first_fail]],12)
"tsen_timed`${}_{\mathtt{global}\!}$`" // IFUN time-sensitive timed global
  12-TGD([BE3,PAND1,RBOX,BE2,PAND3,PAND2,BE1,BE4],[PAND1:[PAND1__failed_first_fail]],12)
\end{lstlisting}

\begin{lstlisting}[%
    style=IFUNModest,
	breaklines,
	basicstyle=\scriptsize\ttfamily,
    caption={\acp{IFUN} used to run \ac{RES} in model \modelDFTlarge (\Cref{code:modelLarge})},
    captionpos=b,
    label=code:IFUNsLarge,
    abovecaptionskip=\medskipamount,
]
"tagn_comp"  // IFUN time-agnostic compositional
  // Omitted due to excessive length
"tagn_mono"      // IFUN time-agnostic monolithic
  17-D([G1_BE_e,G1_BE_a,G1_BE_c,G1_SBE,G1_BE_b,G1_SBE__multiplexer,G1_PAND5_d,G1_SPARE_b,G1_PAND5_a,RBOX,G1_PAND5_c,G1_BE_d,G1_SPARE_a,G1_PAND5_b],[G1_PAND5_a:[G1_PAND5_a__failed_first_fail]],17)+17-D([G2_SBE,G2_SBE__multiplexer,G2_BE_a,G2_BE_b,G2_BE_e,G2_PAND5_c,G2_SPARE_a,RBOX,G2_BE_d,G2_BE_c,G2_PAND5_b,G2_PAND5_a,G2_PAND5_d,G2_SPARE_b],[G2_PAND5_a:[G2_PAND5_a__failed_first_fail]],17)
"tsen_order"     // IFUN time-sensitive order
  17-OD([G1_BE_e,G1_BE_a,G1_BE_c,G1_SBE,G1_BE_b,G1_SBE__multiplexer,G1_PAND5_d,G1_SPARE_b,G1_PAND5_a,RBOX,G1_PAND5_c,G1_BE_d,G1_SPARE_a,G1_PAND5_b],[G1_PAND5_a:[G1_PAND5_a__failed_first_fail]],17)+17-OD([G2_SBE,G2_SBE__multiplexer,G2_BE_a,G2_BE_b,G2_BE_e,G2_PAND5_c,G2_SPARE_a,RBOX,G2_BE_d,G2_BE_c,G2_PAND5_b,G2_PAND5_a,G2_PAND5_d,G2_SPARE_b],[G2_PAND5_a:[G2_PAND5_a__failed_first_fail]],17)
"tsen_order`${}_{\mathtt{global}\!}$`" // IFUN time-sensitive order global
  17-OGD([G1_BE_e,G1_BE_a,G1_BE_c,G1_SBE,G1_BE_b,G1_SBE__multiplexer,G1_PAND5_d,G1_SPARE_b,G1_PAND5_a,RBOX,G1_PAND5_c,G1_BE_d,G1_SPARE_a,G1_PAND5_b],[G1_PAND5_a:[G1_PAND5_a__failed_first_fail]],17)+17-OGD([G2_SBE,G2_SBE__multiplexer,G2_BE_a,G2_BE_b,G2_BE_e,G2_PAND5_c,G2_SPARE_a,RBOX,G2_BE_d,G2_BE_c,G2_PAND5_b,G2_PAND5_a,G2_PAND5_d,G2_SPARE_b],[G2_PAND5_a:[G2_PAND5_a__failed_first_fail]],17)
"tsen_timed"     // IFUN time-sensitive timed
  17-TD([G1_BE_e,G1_BE_a,G1_BE_c,G1_SBE,G1_BE_b,G1_SBE__multiplexer,G1_PAND5_d,G1_SPARE_b,G1_PAND5_a,RBOX,G1_PAND5_c,G1_BE_d,G1_SPARE_a,G1_PAND5_b],[G1_PAND5_a:[G1_PAND5_a__failed_first_fail]],17)+17-TD([G2_SBE,G2_SBE__multiplexer,G2_BE_a,G2_BE_b,G2_BE_e,G2_PAND5_c,G2_SPARE_a,RBOX,G2_BE_d,G2_BE_c,G2_PAND5_b,G2_PAND5_a,G2_PAND5_d,G2_SPARE_b],[G2_PAND5_a:[G2_PAND5_a__failed_first_fail]],17)
"tsen_timed`${}_{\mathtt{global}\!}$`" // IFUN time-sensitive timed global
  17-TGD([G1_BE_e,G1_BE_a,G1_BE_c,G1_SBE,G1_BE_b,G1_SBE__multiplexer,G1_PAND5_d,G1_SPARE_b,G1_PAND5_a,RBOX,G1_PAND5_c,G1_BE_d,G1_SPARE_a,G1_PAND5_b],[G1_PAND5_a:[G1_PAND5_a__failed_first_fail]],17)+17-TGD([G2_SBE,G2_SBE__multiplexer,G2_BE_a,G2_BE_b,G2_BE_e,G2_PAND5_c,G2_SPARE_a,RBOX,G2_BE_d,G2_BE_c,G2_PAND5_b,G2_PAND5_a,G2_PAND5_d,G2_SPARE_b],[G2_PAND5_a:[G2_PAND5_a__failed_first_fail]],17)
\end{lstlisting}

\Cref{code:IFUNsQueue} shows the \acp{IFUN} used for the tandem queue model \modelQueue from \Cref{code:modelTandemQueueSA}. Note that \texttt{\small\color{colcodeidentifier}loc\_20} refers to the failure location (i.e., an overflow happened) of the model. Again, the maximum expansion depth of $30$ ensures that the complete \ac{SC} graph is constructed.
\begin{lstlisting}[%
    style=IFUNModest,
	breaklines,
	basicstyle=\scriptsize\ttfamily,
    caption={\acp{IFUN} used to run \ac{RES} in model \modelQueue (\Cref{code:modelTandemQueueSA})},
    captionpos=b,
    label=code:IFUNsQueue,
    abovecaptionskip=\medskipamount,
]
"tagn_comp"  // IFUN time-agnostic compositional
  // Encoded in the Modest code in property `\texttt{Importance}`
"tagn_mono"      // IFUN time-agnostic monolithic
  30-D([Main], [Main:[loc_20]], 30)
"tsen_order"     // IFUN time-sensitive order
  30-OD([Main], [Main:[loc_20]], 30)
"tsen_order`${}_{\mathtt{global}\!}$`" // IFUN time-sensitive order global
  30-OGD([Main], [Main:[loc_20]], 30)
"tsen_timed"     // IFUN time-sensitive timed
  30-TD([Main], [Main:[loc_20]], 30)
"tsen_timed`${}_{\mathtt{global}\!}$`" // IFUN time-sensitive timed global
  30-TGD([Main], [Main:[loc_20]], 30)
\end{lstlisting}

\subsection{Construction Times of Importance Functions}
\label{app:construction-times-IFUNs}
We report on the construction times that were necessary to construct the \acp{IFUN} used during experimentation:
\begin{itemize}[label=\textbullet]
    \item Fault tree $\modelDFTsmall$: The construction of the complete \ac{SC} graph is possible for the ordered distance metric $\omega_O(\Sigma)$ and the timed distance metrics with global age ($\omega_{TG}(\Sigma)$) in under one second. Meanwhile, constructing the full timed distance metric without global age ($\omega_{T}(\Sigma)$) does not terminate within five minutes. However, we know from a high-level perspective that at most $10$ traversed transitions are needed to reach the target: The distributions for failures and repairs of \acp{BE} have tight lower and upper bounds, limiting the maximum number of possible executed transitions within the given time frame; for more information, see \cite[Section 4]{dengler:2025:time-sensitive-isplit}.
    Thus, we limit the maximum expansion depth to $12$ ($10$ plus $2$ further steps as a safety margin). With this modification, all distance metrics can be computed in under one second.
    \item Fault tree $\modelDFTlarge$: All distance metrics can be fully constructed in a short time frame. However, the construction of the ordered distance metric $\omega_O(\Sigma)$ is significantly faster (about $2$ seconds) than for the timed distance metrics $\omega_T(\Sigma)$/$\omega_{TG}(\Sigma)$ (between $11$ and $12$ seconds). This indicates that the specialized data structure for ordered distance metrics is more efficient than plainly using \ac{DBM} zones, as discussed in \Cref{sec:resampling-timers:models-infinite-domains}.
    \item Tandem queue $\modelQueue$: All distance metrics can be constructed within one second.
\end{itemize}

\subsection{Numerical Results}
\Cref{tab:experiments:results} shows the numerical results of all experiments discussed in \Cref{sec:experiments:results}, and it was used to compile \Cref{fig:experiments:results}:

\begingroup
	\centering
	\sffamily\smaller
	\rowcolors{2}{gray!15}{white}
	\begin{longtable}{%
	p{.448\linewidth}@{\;}    
	>{}C{.11\linewidth}@{}   
	>{}C{.11\linewidth}@{}   
	>{}C{.095\linewidth}@{}    
	>{}C{.12\linewidth}@{\,}  
	>{\:}C{.097\linewidth}}    
\caption{%
	Numerical results of the experiments from \Cref{sec:experiments:results}%
	\label{tab:experiments:results}%
}\\
\rowcolor{black!25}
	\ac{CMC} or \ac{RES} run \RESrun
	& Number of runs
	& Non-zero runs ratio
	& Center $\hat{\varphi}_{\modelm}^{\RESrun}$
	& Half-width ${\left|\hat{\varphi}_{\modelm}^{\RESrun}\right|}/2$
	& Variance
\\
\endfirsthead
\rowcolor{black!25}
	\ac{CMC} or \ac{RES} run \RESrun
	& Number of runs
	& Non-zero runs ratio
	& Center $\hat{\varphi}_{\modelm}^{\RESrun}$
	& Half-width ${\left|\hat{\varphi}_{\modelm}^{\RESrun}\right|}/2$
	& Variance
\\
\endhead
\rowcolor{black!60}
	\multicolumn{6}{c}{\color{white}Results of runs on \modelm{\:}\texttt{=}\:\modelDFTsmall (\Cref{code:modelSmall})\phantom{$\big)$}}
\\
CMC & 1.04\textscale{.8}{\texttt{E+}}8 & 5.5\textscale{.8}{\texttt{E-}}7 & 5.5\textscale{.8}{\texttt{E-}}7 & 1.4\textscale{.8}{\texttt{E-}}7 & 5.5\textscale{.8}{\texttt{E-}}7 \\
tagn{\tt\_}comp{\tt\_}FE{\tt\_}8{\tt\_}standard & 1.22\textscale{.8}{\texttt{E+}}7 & 3.0\textscale{.8}{\texttt{E-}}6 & 3.6\textscale{.8}{\texttt{E-}}7 & 1.2\textscale{.8}{\texttt{E-}}7 & 4.5\textscale{.8}{\texttt{E-}}8 \\
tagn{\tt\_}comp{\tt\_}FE{\tt\_}8{\tt\_}resample & 1.10\textscale{.8}{\texttt{E+}}7 & 1.0\textscale{.8}{\texttt{E-}}5 & 6.4\textscale{.8}{\texttt{E-}}7 & 1.4\textscale{.8}{\texttt{E-}}7 & 5.9\textscale{.8}{\texttt{E-}}8 \\
tagn{\tt\_}comp{\tt\_}FE{\tt\_}15{\tt\_}standard & 6.17\textscale{.8}{\texttt{E+}}6 & 6.6\textscale{.8}{\texttt{E-}}6 & 4.5\textscale{.8}{\texttt{E-}}7 & 1.4\textscale{.8}{\texttt{E-}}7 & 3.1\textscale{.8}{\texttt{E-}}8 \\
tagn{\tt\_}comp{\tt\_}FE{\tt\_}15{\tt\_}resample & 5.47\textscale{.8}{\texttt{E+}}6 & 1.8\textscale{.8}{\texttt{E-}}5 & 4.9\textscale{.8}{\texttt{E-}}7 & 1.3\textscale{.8}{\texttt{E-}}7 & 2.3\textscale{.8}{\texttt{E-}}8 \\
tagn{\tt\_}comp{\tt\_}RST{\tt\_}2{\tt\_}standard & 4.62\textscale{.8}{\texttt{E+}}7 & 1.1\textscale{.8}{\texttt{E-}}6 & 5.6\textscale{.8}{\texttt{E-}}7 & 1.5\textscale{.8}{\texttt{E-}}7 & 2.8\textscale{.8}{\texttt{E-}}7 \\
tagn{\tt\_}comp{\tt\_}RST{\tt\_}2{\tt\_}resample & 4.55\textscale{.8}{\texttt{E+}}7 & 2.2\textscale{.8}{\texttt{E-}}6 & 5.6\textscale{.8}{\texttt{E-}}7 & 1.3\textscale{.8}{\texttt{E-}}7 & 2.1\textscale{.8}{\texttt{E-}}7 \\
tagn{\tt\_}comp{\tt\_}RST{\tt\_}3{\tt\_}standard & 4.54\textscale{.8}{\texttt{E+}}7 & 1.1\textscale{.8}{\texttt{E-}}6 & 5.6\textscale{.8}{\texttt{E-}}7 & 1.5\textscale{.8}{\texttt{E-}}7 & 2.8\textscale{.8}{\texttt{E-}}7 \\
tagn{\tt\_}comp{\tt\_}RST{\tt\_}3{\tt\_}resample & 4.47\textscale{.8}{\texttt{E+}}7 & 2.2\textscale{.8}{\texttt{E-}}6 & 5.7\textscale{.8}{\texttt{E-}}7 & 1.3\textscale{.8}{\texttt{E-}}7 & 2.1\textscale{.8}{\texttt{E-}}7 \\
tagn{\tt\_}mono{\tt\_}FE{\tt\_}8{\tt\_}standard & 1.21\textscale{.8}{\texttt{E+}}7 & 3.9\textscale{.8}{\texttt{E-}}6 & 4.7\textscale{.8}{\texttt{E-}}7 & 1.4\textscale{.8}{\texttt{E-}}7 & 5.7\textscale{.8}{\texttt{E-}}8 \\
tagn{\tt\_}mono{\tt\_}FE{\tt\_}8{\tt\_}resample & 1.12\textscale{.8}{\texttt{E+}}7 & 1.8\textscale{.8}{\texttt{E-}}5 & 5.7\textscale{.8}{\texttt{E-}}7 & 9.0\textscale{.8}{\texttt{E-}}8 & 2.4\textscale{.8}{\texttt{E-}}8 \\
tagn{\tt\_}mono{\tt\_}FE{\tt\_}15{\tt\_}standard & 6.02\textscale{.8}{\texttt{E+}}6 & 9.5\textscale{.8}{\texttt{E-}}6 & 6.1\textscale{.8}{\texttt{E-}}7 & 1.6\textscale{.8}{\texttt{E-}}7 & 4.1\textscale{.8}{\texttt{E-}}8 \\
tagn{\tt\_}mono{\tt\_}FE{\tt\_}15{\tt\_}resample & 5.64\textscale{.8}{\texttt{E+}}6 & 4.2\textscale{.8}{\texttt{E-}}5 & 5.6\textscale{.8}{\texttt{E-}}7 & 8.2\textscale{.8}{\texttt{E-}}8 & 9.9\textscale{.8}{\texttt{E-}}9 \\
tagn{\tt\_}mono{\tt\_}RST{\tt\_}2{\tt\_}standard & 3.79\textscale{.8}{\texttt{E+}}7 & 9.5\textscale{.8}{\texttt{E-}}7 & 4.8\textscale{.8}{\texttt{E-}}7 & 1.6\textscale{.8}{\texttt{E-}}7 & 2.4\textscale{.8}{\texttt{E-}}7 \\
tagn{\tt\_}mono{\tt\_}RST{\tt\_}2{\tt\_}resample & 3.62\textscale{.8}{\texttt{E+}}7 & 4.5\textscale{.8}{\texttt{E-}}6 & 5.7\textscale{.8}{\texttt{E-}}7 & 9.2\textscale{.8}{\texttt{E-}}8 & 8.1\textscale{.8}{\texttt{E-}}8 \\
tagn{\tt\_}mono{\tt\_}RST{\tt\_}3{\tt\_}standard & 3.54\textscale{.8}{\texttt{E+}}7 & 9.3\textscale{.8}{\texttt{E-}}7 & 4.7\textscale{.8}{\texttt{E-}}7 & 1.6\textscale{.8}{\texttt{E-}}7 & 2.3\textscale{.8}{\texttt{E-}}7 \\
tagn{\tt\_}mono{\tt\_}RST{\tt\_}3{\tt\_}resample & 3.65\textscale{.8}{\texttt{E+}}7 & 4.6\textscale{.8}{\texttt{E-}}6 & 5.8\textscale{.8}{\texttt{E-}}7 & 9.3\textscale{.8}{\texttt{E-}}8 & 8.2\textscale{.8}{\texttt{E-}}8 \\
tsen{\tt\_}order{\tt\_}FE{\tt\_}8{\tt\_}standard & 1.26\textscale{.8}{\texttt{E+}}7 & 4.4\textscale{.8}{\texttt{E-}}6 & 5.5\textscale{.8}{\texttt{E-}}7 & 1.4\textscale{.8}{\texttt{E-}}7 & 6.8\textscale{.8}{\texttt{E-}}8 \\
tsen{\tt\_}order{\tt\_}FE{\tt\_}15{\tt\_}standard & 6.49\textscale{.8}{\texttt{E+}}6 & 7.7\textscale{.8}{\texttt{E-}}6 & 5.2\textscale{.8}{\texttt{E-}}7 & 1.5\textscale{.8}{\texttt{E-}}7 & 3.6\textscale{.8}{\texttt{E-}}8 \\
tsen{\tt\_}order{\tt\_}RST{\tt\_}2{\tt\_}standard & 2.99\textscale{.8}{\texttt{E+}}7 & 9.7\textscale{.8}{\texttt{E-}}7 & 4.9\textscale{.8}{\texttt{E-}}7 & 1.8\textscale{.8}{\texttt{E-}}7 & 2.4\textscale{.8}{\texttt{E-}}7 \\
tsen{\tt\_}order{\tt\_}RST{\tt\_}3{\tt\_}standard & 3.07\textscale{.8}{\texttt{E+}}7 & 9.5\textscale{.8}{\texttt{E-}}7 & 4.7\textscale{.8}{\texttt{E-}}7 & 1.7\textscale{.8}{\texttt{E-}}7 & 2.4\textscale{.8}{\texttt{E-}}7 \\
tsen{\tt\_}order{\tt\_}global{\tt\_}FE{\tt\_}8{\tt\_}standard & 1.22\textscale{.8}{\texttt{E+}}7 & 3.3\textscale{.8}{\texttt{E-}}6 & 4.1\textscale{.8}{\texttt{E-}}7 & 1.3\textscale{.8}{\texttt{E-}}7 & 5.1\textscale{.8}{\texttt{E-}}8 \\
tsen{\tt\_}order{\tt\_}global{\tt\_}FE{\tt\_}8{\tt\_}standard{\tt\_}prune & 2.24\textscale{.8}{\texttt{E+}}7 & 4.5\textscale{.8}{\texttt{E-}}6 & 5.6\textscale{.8}{\texttt{E-}}7 & 1.1\textscale{.8}{\texttt{E-}}7 & 7.1\textscale{.8}{\texttt{E-}}8 \\
tsen{\tt\_}order{\tt\_}global{\tt\_}FE{\tt\_}15{\tt\_}standard & 6.09\textscale{.8}{\texttt{E+}}6 & 8.5\textscale{.8}{\texttt{E-}}6 & 5.9\textscale{.8}{\texttt{E-}}7 & 1.6\textscale{.8}{\texttt{E-}}7 & 4.2\textscale{.8}{\texttt{E-}}8 \\
tsen{\tt\_}order{\tt\_}global{\tt\_}FE{\tt\_}15{\tt\_}standard{\tt\_}prune & 1.11\textscale{.8}{\texttt{E+}}7 & 7.6\textscale{.8}{\texttt{E-}}6 & 5.1\textscale{.8}{\texttt{E-}}7 & 1.1\textscale{.8}{\texttt{E-}}7 & 3.4\textscale{.8}{\texttt{E-}}8 \\
tsen{\tt\_}order{\tt\_}global{\tt\_}RST{\tt\_}2{\tt\_}standard & 4.15\textscale{.8}{\texttt{E+}}7 & 8.4\textscale{.8}{\texttt{E-}}7 & 4.2\textscale{.8}{\texttt{E-}}7 & 1.4\textscale{.8}{\texttt{E-}}7 & 2.1\textscale{.8}{\texttt{E-}}7 \\
tsen{\tt\_}order{\tt\_}global{\tt\_}RST{\tt\_}2{\tt\_}standard{\tt\_}prune & 4.54\textscale{.8}{\texttt{E+}}7 & 1.3\textscale{.8}{\texttt{E-}}6 & 6.4\textscale{.8}{\texttt{E-}}7 & 1.6\textscale{.8}{\texttt{E-}}7 & 3.2\textscale{.8}{\texttt{E-}}7 \\
tsen{\tt\_}order{\tt\_}global{\tt\_}RST{\tt\_}3{\tt\_}standard & 4.14\textscale{.8}{\texttt{E+}}7 & 8.4\textscale{.8}{\texttt{E-}}7 & 4.2\textscale{.8}{\texttt{E-}}7 & 1.4\textscale{.8}{\texttt{E-}}7 & 2.1\textscale{.8}{\texttt{E-}}7 \\
tsen{\tt\_}order{\tt\_}global{\tt\_}RST{\tt\_}3{\tt\_}standard{\tt\_}prune & 5.99\textscale{.8}{\texttt{E+}}7 & 9.2\textscale{.8}{\texttt{E-}}7 & 4.6\textscale{.8}{\texttt{E-}}7 & 1.2\textscale{.8}{\texttt{E-}}7 & 2.3\textscale{.8}{\texttt{E-}}7 \\
tsen{\tt\_}timed{\tt\_}FE{\tt\_}8{\tt\_}standard & 1.35\textscale{.8}{\texttt{E+}}7 & 2.9\textscale{.8}{\texttt{E-}}4 & 5.2\textscale{.8}{\texttt{E-}}7 & 3.1\textscale{.8}{\texttt{E-}}8 & 3.5\textscale{.8}{\texttt{E-}}9 \\
tsen{\tt\_}timed{\tt\_}FE{\tt\_}8{\tt\_}resample & 4.82\textscale{.8}{\texttt{E+}}6 & 7.7\textscale{.8}{\texttt{E-}}3 & 5.4\textscale{.8}{\texttt{E-}}7 & 1.0\textscale{.8}{\texttt{E-}}8 & 1.3\textscale{.8}{\texttt{E-}}10 \\
tsen{\tt\_}timed{\tt\_}FE{\tt\_}15{\tt\_}standard & 7.16\textscale{.8}{\texttt{E+}}6 & 8.7\textscale{.8}{\texttt{E-}}4 & 5.6\textscale{.8}{\texttt{E-}}7 & 3.2\textscale{.8}{\texttt{E-}}8 & 1.9\textscale{.8}{\texttt{E-}}9 \\
tsen{\tt\_}timed{\tt\_}FE{\tt\_}15{\tt\_}resample & 2.09\textscale{.8}{\texttt{E+}}6 & 4.5\textscale{.8}{\texttt{E-}}2 & 5.3\textscale{.8}{\texttt{E-}}7 & 6.6\textscale{.8}{\texttt{E-}}9 & 2.4\textscale{.8}{\texttt{E-}}11 \\
tsen{\tt\_}timed{\tt\_}RST{\tt\_}2{\tt\_}standard & 5.27\textscale{.8}{\texttt{E+}}7 & 3.6\textscale{.8}{\texttt{E-}}5 & 5.1\textscale{.8}{\texttt{E-}}7 & 2.6\textscale{.8}{\texttt{E-}}8 & 9.2\textscale{.8}{\texttt{E-}}9 \\
tsen{\tt\_}timed{\tt\_}RST{\tt\_}2{\tt\_}resample & 4.03\textscale{.8}{\texttt{E+}}7 & 1.4\textscale{.8}{\texttt{E-}}4 & 5.3\textscale{.8}{\texttt{E-}}7 & 1.5\textscale{.8}{\texttt{E-}}8 & 2.2\textscale{.8}{\texttt{E-}}9 \\
tsen{\tt\_}timed{\tt\_}RST{\tt\_}3{\tt\_}standard & 5.12\textscale{.8}{\texttt{E+}}7 & 3.6\textscale{.8}{\texttt{E-}}5 & 5.1\textscale{.8}{\texttt{E-}}7 & 2.6\textscale{.8}{\texttt{E-}}8 & 9.2\textscale{.8}{\texttt{E-}}9 \\
tsen{\tt\_}timed{\tt\_}RST{\tt\_}3{\tt\_}resample & 4.34\textscale{.8}{\texttt{E+}}7 & 6.8\textscale{.8}{\texttt{E-}}5 & 5.2\textscale{.8}{\texttt{E-}}7 & 2.0\textscale{.8}{\texttt{E-}}8 & 4.4\textscale{.8}{\texttt{E-}}9 \\
tsen{\tt\_}timed{\tt\_}global{\tt\_}FE{\tt\_}8{\tt\_}standard & 1.29\textscale{.8}{\texttt{E+}}7 & 3.0\textscale{.8}{\texttt{E-}}4 & 5.6\textscale{.8}{\texttt{E-}}7 & 3.5\textscale{.8}{\texttt{E-}}8 & 4.1\textscale{.8}{\texttt{E-}}9 \\
tsen{\tt\_}timed{\tt\_}global{\tt\_}FE{\tt\_}8{\tt\_}resample & 4.96\textscale{.8}{\texttt{E+}}6 & 7.7\textscale{.8}{\texttt{E-}}3 & 5.4\textscale{.8}{\texttt{E-}}7 & 9.9\textscale{.8}{\texttt{E-}}9 & 1.3\textscale{.8}{\texttt{E-}}10 \\
tsen{\tt\_}timed{\tt\_}global{\tt\_}FE{\tt\_}8{\tt\_}standard{\tt\_}prune & 1.69\textscale{.8}{\texttt{E+}}8 & 3.0\textscale{.8}{\texttt{E-}}4 & 5.2\textscale{.8}{\texttt{E-}}7 & 9.0\textscale{.8}{\texttt{E-}}9 & 3.6\textscale{.8}{\texttt{E-}}9 \\
tsen{\tt\_}timed{\tt\_}global{\tt\_}FE{\tt\_}8{\tt\_}resample{\tt\_}prune & 2.39\textscale{.8}{\texttt{E+}}7 & 7.6\textscale{.8}{\texttt{E-}}3 & 5.3\textscale{.8}{\texttt{E-}}7 & 4.5\textscale{.8}{\texttt{E-}}9 & 1.3\textscale{.8}{\texttt{E-}}10 \\
tsen{\tt\_}timed{\tt\_}global{\tt\_}FE{\tt\_}15{\tt\_}standard & 6.98\textscale{.8}{\texttt{E+}}6 & 8.7\textscale{.8}{\texttt{E-}}4 & 5.2\textscale{.8}{\texttt{E-}}7 & 2.8\textscale{.8}{\texttt{E-}}8 & 1.4\textscale{.8}{\texttt{E-}}9 \\
tsen{\tt\_}timed{\tt\_}global{\tt\_}FE{\tt\_}15{\tt\_}resample & 2.05\textscale{.8}{\texttt{E+}}6 & 4.5\textscale{.8}{\texttt{E-}}2 & 5.3\textscale{.8}{\texttt{E-}}7 & 6.7\textscale{.8}{\texttt{E-}}9 & 2.4\textscale{.8}{\texttt{E-}}11 \\
tsen{\tt\_}timed{\tt\_}global{\tt\_}FE{\tt\_}15{\tt\_}standard{\tt\_}prune & 3.62\textscale{.8}{\texttt{E+}}8 & 8.7\textscale{.8}{\texttt{E-}}4 & 5.3\textscale{.8}{\texttt{E-}}7 & 4.1\textscale{.8}{\texttt{E-}}9 & 1.6\textscale{.8}{\texttt{E-}}9 \\
tsen{\tt\_}timed{\tt\_}global{\tt\_}FE{\tt\_}15{\tt\_}resample{\tt\_}prune & 1.02\textscale{.8}{\texttt{E+}}7 & 4.5\textscale{.8}{\texttt{E-}}2 & 5.3\textscale{.8}{\texttt{E-}}7 & 3.1\textscale{.8}{\texttt{E-}}9 & 2.5\textscale{.8}{\texttt{E-}}11 \\
tsen{\tt\_}timed{\tt\_}global{\tt\_}RST{\tt\_}2{\tt\_}standard & 5.15\textscale{.8}{\texttt{E+}}7 & 3.7\textscale{.8}{\texttt{E-}}5 & 5.2\textscale{.8}{\texttt{E-}}7 & 2.6\textscale{.8}{\texttt{E-}}8 & 9.4\textscale{.8}{\texttt{E-}}9 \\
tsen{\tt\_}timed{\tt\_}global{\tt\_}RST{\tt\_}2{\tt\_}resample & 3.95\textscale{.8}{\texttt{E+}}7 & 1.4\textscale{.8}{\texttt{E-}}4 & 5.2\textscale{.8}{\texttt{E-}}7 & 1.5\textscale{.8}{\texttt{E-}}8 & 2.2\textscale{.8}{\texttt{E-}}9 \\
tsen{\tt\_}timed{\tt\_}global{\tt\_}RST{\tt\_}2{\tt\_}standard{\tt\_}prune & 2.43\textscale{.8}{\texttt{E+}}9 & 2.0\textscale{.8}{\texttt{E-}}5 & 5.4\textscale{.8}{\texttt{E-}}7 & 5.3\textscale{.8}{\texttt{E-}}9 & 1.8\textscale{.8}{\texttt{E-}}8 \\
tsen{\tt\_}timed{\tt\_}global{\tt\_}RST{\tt\_}2{\tt\_}resample{\tt\_}prune & 8.82\textscale{.8}{\texttt{E+}}7 & 1.4\textscale{.8}{\texttt{E-}}4 & 5.3\textscale{.8}{\texttt{E-}}7 & 9.9\textscale{.8}{\texttt{E-}}9 & 2.2\textscale{.8}{\texttt{E-}}9 \\
tsen{\tt\_}timed{\tt\_}global{\tt\_}RST{\tt\_}3{\tt\_}standard & 5.27\textscale{.8}{\texttt{E+}}7 & 3.7\textscale{.8}{\texttt{E-}}5 & 5.2\textscale{.8}{\texttt{E-}}7 & 2.6\textscale{.8}{\texttt{E-}}8 & 9.4\textscale{.8}{\texttt{E-}}9 \\
tsen{\tt\_}timed{\tt\_}global{\tt\_}RST{\tt\_}3{\tt\_}resample & 4.36\textscale{.8}{\texttt{E+}}7 & 6.8\textscale{.8}{\texttt{E-}}5 & 5.2\textscale{.8}{\texttt{E-}}7 & 2.0\textscale{.8}{\texttt{E-}}8 & 4.4\textscale{.8}{\texttt{E-}}9 \\
tsen{\tt\_}timed{\tt\_}global{\tt\_}RST{\tt\_}3{\tt\_}standard{\tt\_}prune & 2.42\textscale{.8}{\texttt{E+}}9 & 2.0\textscale{.8}{\texttt{E-}}5 & 5.3\textscale{.8}{\texttt{E-}}7 & 5.2\textscale{.8}{\texttt{E-}}9 & 1.7\textscale{.8}{\texttt{E-}}8 \\
tsen{\tt\_}timed{\tt\_}global{\tt\_}RST{\tt\_}3{\tt\_}resample{\tt\_}prune & 1.15\textscale{.8}{\texttt{E+}}8 & 7.0\textscale{.8}{\texttt{E-}}5 & 5.4\textscale{.8}{\texttt{E-}}7 & 1.2\textscale{.8}{\texttt{E-}}8 & 4.5\textscale{.8}{\texttt{E-}}9 \\
\rowcolor{black!60}
	\multicolumn{6}{c}{\color{white}Results of runs on \modelm{\:}\texttt{=}\:\modelDFTlarge (\Cref{code:modelLarge})\phantom{$\big)$}}
\\
CMC & 1.71\textscale{.8}{\texttt{E+}}7 & 0 & -- & -- & -- \\
tagn{\tt\_}comp{\tt\_}FE{\tt\_}8{\tt\_}standard & 1.47\textscale{.8}{\texttt{E+}}6 & 2.7\textscale{.8}{\texttt{E-}}6 & 9.8\textscale{.8}{\texttt{E-}}9 & 1.1\textscale{.8}{\texttt{E-}}8 & 4.7\textscale{.8}{\texttt{E-}}11 \\
tagn{\tt\_}comp{\tt\_}FE{\tt\_}8{\tt\_}resample & 8.65\textscale{.8}{\texttt{E+}}5 & 2.5\textscale{.8}{\texttt{E-}}3 & 3.9\textscale{.8}{\texttt{E-}}8 & 1.1\textscale{.8}{\texttt{E-}}8 & 3.0\textscale{.8}{\texttt{E-}}11 \\
tagn{\tt\_}comp{\tt\_}FE{\tt\_}15{\tt\_}standard & 7.35\textscale{.8}{\texttt{E+}}5 & 4.1\textscale{.8}{\texttt{E-}}6 & 4.0\textscale{.8}{\texttt{E-}}8 & 6.7\textscale{.8}{\texttt{E-}}8 & 8.6\textscale{.8}{\texttt{E-}}10 \\
tagn{\tt\_}comp{\tt\_}FE{\tt\_}15{\tt\_}resample & 3.75\textscale{.8}{\texttt{E+}}5 & 3.1\textscale{.8}{\texttt{E-}}2 & 2.8\textscale{.8}{\texttt{E-}}8 & 4.6\textscale{.8}{\texttt{E-}}9 & 2.1\textscale{.8}{\texttt{E-}}12 \\
tagn{\tt\_}comp{\tt\_}RST{\tt\_}2{\tt\_}standard & 2.21\textscale{.8}{\texttt{E+}}6 & 1.8\textscale{.8}{\texttt{E-}}5 & 4.9\textscale{.8}{\texttt{E-}}8 & 5.8\textscale{.8}{\texttt{E-}}8 & 1.9\textscale{.8}{\texttt{E-}}9 \\
tagn{\tt\_}comp{\tt\_}RST{\tt\_}2{\tt\_}resample & 1.92\textscale{.8}{\texttt{E+}}6 & 9.8\textscale{.8}{\texttt{E-}}4 & 3.2\textscale{.8}{\texttt{E-}}8 & 2.7\textscale{.8}{\texttt{E-}}9 & 3.5\textscale{.8}{\texttt{E-}}12 \\
tagn{\tt\_}comp{\tt\_}RST{\tt\_}3{\tt\_}standard & 4.76\textscale{.8}{\texttt{E+}}6 & 6.1\textscale{.8}{\texttt{E-}}6 & 6.8\textscale{.8}{\texttt{E-}}8 & 1.0\textscale{.8}{\texttt{E-}}7 & 1.3\textscale{.8}{\texttt{E-}}8 \\
tagn{\tt\_}comp{\tt\_}RST{\tt\_}3{\tt\_}resample & 4.83\textscale{.8}{\texttt{E+}}6 & 2.3\textscale{.8}{\texttt{E-}}4 & 3.2\textscale{.8}{\texttt{E-}}8 & 4.2\textscale{.8}{\texttt{E-}}9 & 2.2\textscale{.8}{\texttt{E-}}11 \\
tagn{\tt\_}mono{\tt\_}FE{\tt\_}8{\tt\_}standard & 1.70\textscale{.8}{\texttt{E+}}6 & 2.9\textscale{.8}{\texttt{E-}}6 & 2.7\textscale{.8}{\texttt{E-}}8 & 2.7\textscale{.8}{\texttt{E-}}8 & 3.2\textscale{.8}{\texttt{E-}}10 \\
tagn{\tt\_}mono{\tt\_}FE{\tt\_}8{\tt\_}resample & 1.07\textscale{.8}{\texttt{E+}}6 & 7.2\textscale{.8}{\texttt{E-}}3 & 3.2\textscale{.8}{\texttt{E-}}8 & 2.5\textscale{.8}{\texttt{E-}}9 & 1.7\textscale{.8}{\texttt{E-}}12 \\
tagn{\tt\_}mono{\tt\_}FE{\tt\_}15{\tt\_}standard & 8.13\textscale{.8}{\texttt{E+}}5 & 7.4\textscale{.8}{\texttt{E-}}6 & 3.3\textscale{.8}{\texttt{E-}}8 & 4.9\textscale{.8}{\texttt{E-}}8 & 5.1\textscale{.8}{\texttt{E-}}10 \\
tagn{\tt\_}mono{\tt\_}FE{\tt\_}15{\tt\_}resample & 4.51\textscale{.8}{\texttt{E+}}5 & 6.5\textscale{.8}{\texttt{E-}}2 & 3.3\textscale{.8}{\texttt{E-}}8 & 1.4\textscale{.8}{\texttt{E-}}9 & 2.4\textscale{.8}{\texttt{E-}}13 \\
tagn{\tt\_}mono{\tt\_}RST{\tt\_}2{\tt\_}standard & 4.48\textscale{.8}{\texttt{E+}}6 & 2.9\textscale{.8}{\texttt{E-}}6 & 2.4\textscale{.8}{\texttt{E-}}8 & 1.6\textscale{.8}{\texttt{E-}}8 & 2.8\textscale{.8}{\texttt{E-}}10 \\
tagn{\tt\_}mono{\tt\_}RST{\tt\_}2{\tt\_}resample & 4.49\textscale{.8}{\texttt{E+}}6 & 2.8\textscale{.8}{\texttt{E-}}5 & 3.5\textscale{.8}{\texttt{E-}}8 & 6.6\textscale{.8}{\texttt{E-}}9 & 5.0\textscale{.8}{\texttt{E-}}11 \\
tagn{\tt\_}mono{\tt\_}RST{\tt\_}3{\tt\_}standard & 4.49\textscale{.8}{\texttt{E+}}6 & 2.9\textscale{.8}{\texttt{E-}}6 & 2.4\textscale{.8}{\texttt{E-}}8 & 1.5\textscale{.8}{\texttt{E-}}8 & 2.8\textscale{.8}{\texttt{E-}}10 \\
tagn{\tt\_}mono{\tt\_}RST{\tt\_}3{\tt\_}resample & 4.42\textscale{.8}{\texttt{E+}}6 & 2.8\textscale{.8}{\texttt{E-}}5 & 3.5\textscale{.8}{\texttt{E-}}8 & 6.6\textscale{.8}{\texttt{E-}}9 & 5.0\textscale{.8}{\texttt{E-}}11 \\
tsen{\tt\_}order{\tt\_}FE{\tt\_}8{\tt\_}standard & 1.42\textscale{.8}{\texttt{E+}}6 & 6.3\textscale{.8}{\texttt{E-}}6 & 3.5\textscale{.8}{\texttt{E-}}8 & 3.3\textscale{.8}{\texttt{E-}}8 & 4.0\textscale{.8}{\texttt{E-}}10 \\
tsen{\tt\_}order{\tt\_}FE{\tt\_}15{\tt\_}standard & 7.13\textscale{.8}{\texttt{E+}}5 & 1.8\textscale{.8}{\texttt{E-}}5 & 3.1\textscale{.8}{\texttt{E-}}8 & 2.6\textscale{.8}{\texttt{E-}}8 & 1.2\textscale{.8}{\texttt{E-}}10 \\
tsen{\tt\_}order{\tt\_}RST{\tt\_}2{\tt\_}standard & 2.08\textscale{.8}{\texttt{E+}}6 & 3.4\textscale{.8}{\texttt{E-}}5 & 2.3\textscale{.8}{\texttt{E-}}8 & 9.7\textscale{.8}{\texttt{E-}}9 & 5.1\textscale{.8}{\texttt{E-}}11 \\
tsen{\tt\_}order{\tt\_}RST{\tt\_}3{\tt\_}standard & 1.96\textscale{.8}{\texttt{E+}}6 & 5.9\textscale{.8}{\texttt{E-}}5 & 3.9\textscale{.8}{\texttt{E-}}8 & 1.4\textscale{.8}{\texttt{E-}}8 & 1.0\textscale{.8}{\texttt{E-}}10 \\
tsen{\tt\_}order{\tt\_}global{\tt\_}FE{\tt\_}8{\tt\_}standard & 1.82\textscale{.8}{\texttt{E+}}6 & 9.0\textscale{.8}{\texttt{E-}}5 & 3.3\textscale{.8}{\texttt{E-}}8 & 1.4\textscale{.8}{\texttt{E-}}8 & 9.1\textscale{.8}{\texttt{E-}}11 \\
tsen{\tt\_}order{\tt\_}global{\tt\_}FE{\tt\_}8{\tt\_}standard{\tt\_}prune & 2.24\textscale{.8}{\texttt{E+}}8 & 9.3\textscale{.8}{\texttt{E-}}5 & 2.6\textscale{.8}{\texttt{E-}}8 & 1.1\textscale{.8}{\texttt{E-}}9 & 6.9\textscale{.8}{\texttt{E-}}11 \\
tsen{\tt\_}order{\tt\_}global{\tt\_}FE{\tt\_}15{\tt\_}standard & 9.31\textscale{.8}{\texttt{E+}}5 & 4.3\textscale{.8}{\texttt{E-}}4 & 2.3\textscale{.8}{\texttt{E-}}8 & 7.1\textscale{.8}{\texttt{E-}}9 & 1.2\textscale{.8}{\texttt{E-}}11 \\
tsen{\tt\_}order{\tt\_}global{\tt\_}FE{\tt\_}15{\tt\_}standard{\tt\_}prune & 1.16\textscale{.8}{\texttt{E+}}8 & 4.1\textscale{.8}{\texttt{E-}}4 & 2.7\textscale{.8}{\texttt{E-}}8 & 9.8\textscale{.8}{\texttt{E-}}10 & 2.9\textscale{.8}{\texttt{E-}}11 \\
tsen{\tt\_}order{\tt\_}global{\tt\_}RST{\tt\_}2{\tt\_}standard & 7.21\textscale{.8}{\texttt{E+}}6 & 6.1\textscale{.8}{\texttt{E-}}5 & 3.8\textscale{.8}{\texttt{E-}}8 & 1.3\textscale{.8}{\texttt{E-}}8 & 3.0\textscale{.8}{\texttt{E-}}10 \\
tsen{\tt\_}order{\tt\_}global{\tt\_}RST{\tt\_}2{\tt\_}standard{\tt\_}prune & 3.19\textscale{.8}{\texttt{E+}}8 & 4.4\textscale{.8}{\texttt{E-}}5 & 3.3\textscale{.8}{\texttt{E-}}8 & 1.8\textscale{.8}{\texttt{E-}}9 & 2.8\textscale{.8}{\texttt{E-}}10 \\
tsen{\tt\_}order{\tt\_}global{\tt\_}RST{\tt\_}3{\tt\_}standard & 7.22\textscale{.8}{\texttt{E+}}6 & 6.1\textscale{.8}{\texttt{E-}}5 & 3.8\textscale{.8}{\texttt{E-}}8 & 1.3\textscale{.8}{\texttt{E-}}8 & 3.0\textscale{.8}{\texttt{E-}}10 \\
tsen{\tt\_}order{\tt\_}global{\tt\_}RST{\tt\_}3{\tt\_}standard{\tt\_}prune & 5.14\textscale{.8}{\texttt{E+}}8 & 8.2\textscale{.8}{\texttt{E-}}6 & 3.2\textscale{.8}{\texttt{E-}}8 & 1.7\textscale{.8}{\texttt{E-}}9 & 3.9\textscale{.8}{\texttt{E-}}10 \\
tsen{\tt\_}timed{\tt\_}FE{\tt\_}8{\tt\_}standard & 1.50\textscale{.8}{\texttt{E+}}6 & 6.0\textscale{.8}{\texttt{E-}}6 & 3.3\textscale{.8}{\texttt{E-}}8 & 3.1\textscale{.8}{\texttt{E-}}8 & 3.8\textscale{.8}{\texttt{E-}}10 \\
tsen{\tt\_}timed{\tt\_}FE{\tt\_}8{\tt\_}resample & 1.07\textscale{.8}{\texttt{E+}}6 & 7.2\textscale{.8}{\texttt{E-}}3 & 3.2\textscale{.8}{\texttt{E-}}8 & 2.5\textscale{.8}{\texttt{E-}}9 & 1.7\textscale{.8}{\texttt{E-}}12 \\
tsen{\tt\_}timed{\tt\_}FE{\tt\_}15{\tt\_}standard & 7.44\textscale{.8}{\texttt{E+}}5 & 1.7\textscale{.8}{\texttt{E-}}5 & 3.0\textscale{.8}{\texttt{E-}}8 & 2.5\textscale{.8}{\texttt{E-}}8 & 1.2\textscale{.8}{\texttt{E-}}10 \\
tsen{\tt\_}timed{\tt\_}FE{\tt\_}15{\tt\_}resample & 4.56\textscale{.8}{\texttt{E+}}5 & 6.5\textscale{.8}{\texttt{E-}}2 & 3.2\textscale{.8}{\texttt{E-}}8 & 1.4\textscale{.8}{\texttt{E-}}9 & 2.3\textscale{.8}{\texttt{E-}}13 \\
tsen{\tt\_}timed{\tt\_}RST{\tt\_}2{\tt\_}standard & 2.01\textscale{.8}{\texttt{E+}}6 & 5.9\textscale{.8}{\texttt{E-}}5 & 3.8\textscale{.8}{\texttt{E-}}8 & 1.4\textscale{.8}{\texttt{E-}}8 & 9.9\textscale{.8}{\texttt{E-}}11 \\
tsen{\tt\_}timed{\tt\_}RST{\tt\_}2{\tt\_}resample & 4.50\textscale{.8}{\texttt{E+}}6 & 2.8\textscale{.8}{\texttt{E-}}5 & 3.6\textscale{.8}{\texttt{E-}}8 & 6.6\textscale{.8}{\texttt{E-}}9 & 5.1\textscale{.8}{\texttt{E-}}11 \\
tsen{\tt\_}timed{\tt\_}RST{\tt\_}3{\tt\_}standard & 2.01\textscale{.8}{\texttt{E+}}6 & 6.0\textscale{.8}{\texttt{E-}}5 & 3.9\textscale{.8}{\texttt{E-}}8 & 1.4\textscale{.8}{\texttt{E-}}8 & 9.9\textscale{.8}{\texttt{E-}}11 \\
tsen{\tt\_}timed{\tt\_}RST{\tt\_}3{\tt\_}resample & 4.55\textscale{.8}{\texttt{E+}}6 & 2.8\textscale{.8}{\texttt{E-}}5 & 3.5\textscale{.8}{\texttt{E-}}8 & 6.5\textscale{.8}{\texttt{E-}}9 & 5.0\textscale{.8}{\texttt{E-}}11 \\
tsen{\tt\_}timed{\tt\_}global{\tt\_}FE{\tt\_}8{\tt\_}standard & 2.41\textscale{.8}{\texttt{E+}}6 & 9.7\textscale{.8}{\texttt{E-}}5 & 3.0\textscale{.8}{\texttt{E-}}8 & 1.1\textscale{.8}{\texttt{E-}}8 & 7.4\textscale{.8}{\texttt{E-}}11 \\
tsen{\tt\_}timed{\tt\_}global{\tt\_}FE{\tt\_}8{\tt\_}resample & 1.06\textscale{.8}{\texttt{E+}}6 & 7.2\textscale{.8}{\texttt{E-}}3 & 3.2\textscale{.8}{\texttt{E-}}8 & 2.5\textscale{.8}{\texttt{E-}}9 & 1.7\textscale{.8}{\texttt{E-}}12 \\
tsen{\tt\_}timed{\tt\_}global{\tt\_}FE{\tt\_}8{\tt\_}standard{\tt\_}prune & 3.21\textscale{.8}{\texttt{E+}}8 & 9.7\textscale{.8}{\texttt{E-}}5 & 2.6\textscale{.8}{\texttt{E-}}8 & 9.6\textscale{.8}{\texttt{E-}}10 & 7.8\textscale{.8}{\texttt{E-}}11 \\
tsen{\tt\_}timed{\tt\_}global{\tt\_}FE{\tt\_}8{\tt\_}resample{\tt\_}prune & 1.07\textscale{.8}{\texttt{E+}}6 & 7.2\textscale{.8}{\texttt{E-}}3 & 3.2\textscale{.8}{\texttt{E-}}8 & 2.5\textscale{.8}{\texttt{E-}}9 & 1.7\textscale{.8}{\texttt{E-}}12 \\
tsen{\tt\_}timed{\tt\_}global{\tt\_}FE{\tt\_}15{\tt\_}standard & 1.21\textscale{.8}{\texttt{E+}}6 & 4.2\textscale{.8}{\texttt{E-}}4 & 2.4\textscale{.8}{\texttt{E-}}8 & 7.9\textscale{.8}{\texttt{E-}}9 & 2.0\textscale{.8}{\texttt{E-}}11 \\
tsen{\tt\_}timed{\tt\_}global{\tt\_}FE{\tt\_}15{\tt\_}resample & 4.64\textscale{.8}{\texttt{E+}}5 & 6.5\textscale{.8}{\texttt{E-}}2 & 3.3\textscale{.8}{\texttt{E-}}8 & 1.4\textscale{.8}{\texttt{E-}}9 & 2.3\textscale{.8}{\texttt{E-}}13 \\
tsen{\tt\_}timed{\tt\_}global{\tt\_}FE{\tt\_}15{\tt\_}standard{\tt\_}prune & 1.82\textscale{.8}{\texttt{E+}}8 & 4.0\textscale{.8}{\texttt{E-}}4 & 2.7\textscale{.8}{\texttt{E-}}8 & 7.9\textscale{.8}{\texttt{E-}}10 & 3.0\textscale{.8}{\texttt{E-}}11 \\
tsen{\tt\_}timed{\tt\_}global{\tt\_}FE{\tt\_}15{\tt\_}resample{\tt\_}prune & 4.53\textscale{.8}{\texttt{E+}}5 & 6.5\textscale{.8}{\texttt{E-}}2 & 3.2\textscale{.8}{\texttt{E-}}8 & 1.4\textscale{.8}{\texttt{E-}}9 & 2.3\textscale{.8}{\texttt{E-}}13 \\
tsen{\tt\_}timed{\tt\_}global{\tt\_}RST{\tt\_}2{\tt\_}standard & 9.23\textscale{.8}{\texttt{E+}}6 & 5.9\textscale{.8}{\texttt{E-}}5 & 3.5\textscale{.8}{\texttt{E-}}8 & 1.0\textscale{.8}{\texttt{E-}}8 & 2.5\textscale{.8}{\texttt{E-}}10 \\
tsen{\tt\_}timed{\tt\_}global{\tt\_}RST{\tt\_}2{\tt\_}resample & 4.65\textscale{.8}{\texttt{E+}}6 & 2.8\textscale{.8}{\texttt{E-}}5 & 3.6\textscale{.8}{\texttt{E-}}8 & 6.4\textscale{.8}{\texttt{E-}}9 & 5.0\textscale{.8}{\texttt{E-}}11 \\
tsen{\tt\_}timed{\tt\_}global{\tt\_}RST{\tt\_}2{\tt\_}standard{\tt\_}prune & 3.37\textscale{.8}{\texttt{E+}}8 & 6.2\textscale{.8}{\texttt{E-}}5 & 3.2\textscale{.8}{\texttt{E-}}8 & 1.4\textscale{.8}{\texttt{E-}}9 & 1.8\textscale{.8}{\texttt{E-}}10 \\
tsen{\tt\_}timed{\tt\_}global{\tt\_}RST{\tt\_}2{\tt\_}resample{\tt\_}prune & 4.40\textscale{.8}{\texttt{E+}}6 & 2.8\textscale{.8}{\texttt{E-}}5 & 3.5\textscale{.8}{\texttt{E-}}8 & 6.6\textscale{.8}{\texttt{E-}}9 & 5.0\textscale{.8}{\texttt{E-}}11 \\
tsen{\tt\_}timed{\tt\_}global{\tt\_}RST{\tt\_}3{\tt\_}standard & 9.36\textscale{.8}{\texttt{E+}}6 & 5.9\textscale{.8}{\texttt{E-}}5 & 3.4\textscale{.8}{\texttt{E-}}8 & 1.0\textscale{.8}{\texttt{E-}}8 & 2.4\textscale{.8}{\texttt{E-}}10 \\
tsen{\tt\_}timed{\tt\_}global{\tt\_}RST{\tt\_}3{\tt\_}resample & 4.58\textscale{.8}{\texttt{E+}}6 & 2.8\textscale{.8}{\texttt{E-}}5 & 3.6\textscale{.8}{\texttt{E-}}8 & 6.6\textscale{.8}{\texttt{E-}}9 & 5.2\textscale{.8}{\texttt{E-}}11 \\
tsen{\tt\_}timed{\tt\_}global{\tt\_}RST{\tt\_}3{\tt\_}standard{\tt\_}prune & 5.08\textscale{.8}{\texttt{E+}}8 & 1.5\textscale{.8}{\texttt{E-}}5 & 3.3\textscale{.8}{\texttt{E-}}8 & 1.6\textscale{.8}{\texttt{E-}}9 & 3.6\textscale{.8}{\texttt{E-}}10 \\
tsen{\tt\_}timed{\tt\_}global{\tt\_}RST{\tt\_}3{\tt\_}resample{\tt\_}prune & 4.58\textscale{.8}{\texttt{E+}}6 & 2.8\textscale{.8}{\texttt{E-}}5 & 3.6\textscale{.8}{\texttt{E-}}8 & 6.5\textscale{.8}{\texttt{E-}}9 & 5.1\textscale{.8}{\texttt{E-}}11 \\
\rowcolor{black!60}
	\multicolumn{6}{c}{\color{white}Results of runs on \modelm{\:}\texttt{=}\:\modelQueue (\Cref{code:modelTandemQueueSA})\phantom{$\big)$}}
\\
CMC & 4.52\textscale{.8}{\texttt{E+}}8 & 6.6\textscale{.8}{\texttt{E-}}9 & 6.6\textscale{.8}{\texttt{E-}}9 & 7.5\textscale{.8}{\texttt{E-}}9 & 6.6\textscale{.8}{\texttt{E-}}9 \\
tagn{\tt\_}comp{\tt\_}FE{\tt\_}8{\tt\_}standard & 2.50\textscale{.8}{\texttt{E+}}7 & 5.6\textscale{.8}{\texttt{E-}}7 & 1.8\textscale{.8}{\texttt{E-}}8 & 1.0\textscale{.8}{\texttt{E-}}8 & 6.7\textscale{.8}{\texttt{E-}}10 \\
tagn{\tt\_}comp{\tt\_}FE{\tt\_}8{\tt\_}resample & 3.45\textscale{.8}{\texttt{E+}}7 & 2.3\textscale{.8}{\texttt{E-}}7 & 2.1\textscale{.8}{\texttt{E-}}8 & 1.5\textscale{.8}{\texttt{E-}}8 & 2.1\textscale{.8}{\texttt{E-}}9 \\
tagn{\tt\_}comp{\tt\_}FE{\tt\_}15{\tt\_}standard & 1.26\textscale{.8}{\texttt{E+}}7 & 1.2\textscale{.8}{\texttt{E-}}6 & 1.2\textscale{.8}{\texttt{E-}}8 & 6.8\textscale{.8}{\texttt{E-}}9 & 1.5\textscale{.8}{\texttt{E-}}10 \\
tagn{\tt\_}comp{\tt\_}FE{\tt\_}15{\tt\_}resample & 1.75\textscale{.8}{\texttt{E+}}7 & 3.4\textscale{.8}{\texttt{E-}}7 & 1.9\textscale{.8}{\texttt{E-}}8 & 1.6\textscale{.8}{\texttt{E-}}8 & 1.1\textscale{.8}{\texttt{E-}}9 \\
tagn{\tt\_}comp{\tt\_}RST{\tt\_}2{\tt\_}standard & 1.07\textscale{.8}{\texttt{E+}}8 & 4.7\textscale{.8}{\texttt{E-}}8 & 1.2\textscale{.8}{\texttt{E-}}8 & 1.0\textscale{.8}{\texttt{E-}}8 & 2.9\textscale{.8}{\texttt{E-}}9 \\
tagn{\tt\_}comp{\tt\_}RST{\tt\_}2{\tt\_}resample & 8.19\textscale{.8}{\texttt{E+}}7 & 2.1\textscale{.8}{\texttt{E-}}7 & 2.3\textscale{.8}{\texttt{E-}}8 & 1.1\textscale{.8}{\texttt{E-}}8 & 2.6\textscale{.8}{\texttt{E-}}9 \\
tagn{\tt\_}comp{\tt\_}RST{\tt\_}3{\tt\_}standard & 1.10\textscale{.8}{\texttt{E+}}8 & 4.5\textscale{.8}{\texttt{E-}}8 & 1.1\textscale{.8}{\texttt{E-}}8 & 9.9\textscale{.8}{\texttt{E-}}9 & 2.8\textscale{.8}{\texttt{E-}}9 \\
tagn{\tt\_}comp{\tt\_}RST{\tt\_}3{\tt\_}resample & 7.90\textscale{.8}{\texttt{E+}}7 & 2.0\textscale{.8}{\texttt{E-}}7 & 2.3\textscale{.8}{\texttt{E-}}8 & 1.1\textscale{.8}{\texttt{E-}}8 & 2.5\textscale{.8}{\texttt{E-}}9 \\
tagn{\tt\_}mono{\tt\_}FE{\tt\_}8{\tt\_}standard & 3.20\textscale{.8}{\texttt{E+}}7 & 1.1\textscale{.8}{\texttt{E-}}6 & 2.2\textscale{.8}{\texttt{E-}}8 & 8.2\textscale{.8}{\texttt{E-}}9 & 5.7\textscale{.8}{\texttt{E-}}10 \\
tagn{\tt\_}mono{\tt\_}FE{\tt\_}8{\tt\_}resample & 2.12\textscale{.8}{\texttt{E+}}7 & 2.4\textscale{.8}{\texttt{E-}}6 & 1.7\textscale{.8}{\texttt{E-}}8 & 5.5\textscale{.8}{\texttt{E-}}9 & 1.7\textscale{.8}{\texttt{E-}}10 \\
tagn{\tt\_}mono{\tt\_}FE{\tt\_}15{\tt\_}standard & 1.56\textscale{.8}{\texttt{E+}}7 & 1.9\textscale{.8}{\texttt{E-}}6 & 1.6\textscale{.8}{\texttt{E-}}8 & 6.2\textscale{.8}{\texttt{E-}}9 & 1.5\textscale{.8}{\texttt{E-}}10 \\
tagn{\tt\_}mono{\tt\_}FE{\tt\_}15{\tt\_}resample & 1.00\textscale{.8}{\texttt{E+}}7 & 5.0\textscale{.8}{\texttt{E-}}6 & 1.7\textscale{.8}{\texttt{E-}}8 & 5.5\textscale{.8}{\texttt{E-}}9 & 7.8\textscale{.8}{\texttt{E-}}11 \\
tagn{\tt\_}mono{\tt\_}RST{\tt\_}2{\tt\_}standard & 8.18\textscale{.8}{\texttt{E+}}7 & 3.2\textscale{.8}{\texttt{E-}}7 & 2.0\textscale{.8}{\texttt{E-}}8 & 7.6\textscale{.8}{\texttt{E-}}9 & 1.2\textscale{.8}{\texttt{E-}}9 \\
tagn{\tt\_}mono{\tt\_}RST{\tt\_}2{\tt\_}resample & 8.12\textscale{.8}{\texttt{E+}}7 & 5.2\textscale{.8}{\texttt{E-}}7 & 1.7\textscale{.8}{\texttt{E-}}8 & 5.4\textscale{.8}{\texttt{E-}}9 & 6.1\textscale{.8}{\texttt{E-}}10 \\
tagn{\tt\_}mono{\tt\_}RST{\tt\_}3{\tt\_}standard & 1.16\textscale{.8}{\texttt{E+}}8 & 8.6\textscale{.8}{\texttt{E-}}8 & 1.1\textscale{.8}{\texttt{E-}}8 & 6.7\textscale{.8}{\texttt{E-}}9 & 1.3\textscale{.8}{\texttt{E-}}9 \\
tagn{\tt\_}mono{\tt\_}RST{\tt\_}3{\tt\_}resample & 1.14\textscale{.8}{\texttt{E+}}8 & 2.7\textscale{.8}{\texttt{E-}}7 & 1.7\textscale{.8}{\texttt{E-}}8 & 6.0\textscale{.8}{\texttt{E-}}9 & 1.1\textscale{.8}{\texttt{E-}}9 \\
tsen{\tt\_}order{\tt\_}FE{\tt\_}8{\tt\_}standard & 2.86\textscale{.8}{\texttt{E+}}7 & 7.0\textscale{.8}{\texttt{E-}}7 & 1.1\textscale{.8}{\texttt{E-}}8 & 5.2\textscale{.8}{\texttt{E-}}9 & 2.0\textscale{.8}{\texttt{E-}}10 \\
tsen{\tt\_}order{\tt\_}FE{\tt\_}15{\tt\_}standard & 1.37\textscale{.8}{\texttt{E+}}7 & 2.3\textscale{.8}{\texttt{E-}}6 & 1.3\textscale{.8}{\texttt{E-}}8 & 5.3\textscale{.8}{\texttt{E-}}9 & 1.0\textscale{.8}{\texttt{E-}}10 \\
tsen{\tt\_}order{\tt\_}RST{\tt\_}2{\tt\_}standard & 9.05\textscale{.8}{\texttt{E+}}7 & 4.1\textscale{.8}{\texttt{E-}}7 & 1.4\textscale{.8}{\texttt{E-}}8 & 4.6\textscale{.8}{\texttt{E-}}9 & 5.0\textscale{.8}{\texttt{E-}}10 \\
tsen{\tt\_}order{\tt\_}RST{\tt\_}3{\tt\_}standard & 1.31\textscale{.8}{\texttt{E+}}8 & 2.0\textscale{.8}{\texttt{E-}}7 & 1.5\textscale{.8}{\texttt{E-}}8 & 6.0\textscale{.8}{\texttt{E-}}9 & 1.2\textscale{.8}{\texttt{E-}}9 \\
tsen{\tt\_}order{\tt\_}global{\tt\_}FE{\tt\_}8{\tt\_}standard & 2.39\textscale{.8}{\texttt{E+}}7 & 1.3\textscale{.8}{\texttt{E-}}6 & 1.8\textscale{.8}{\texttt{E-}}8 & 7.2\textscale{.8}{\texttt{E-}}9 & 3.2\textscale{.8}{\texttt{E-}}10 \\
tsen{\tt\_}order{\tt\_}global{\tt\_}FE{\tt\_}8{\tt\_}standard{\tt\_}prune & 4.27\textscale{.8}{\texttt{E+}}7 & 1.1\textscale{.8}{\texttt{E-}}6 & 1.7\textscale{.8}{\texttt{E-}}8 & 5.8\textscale{.8}{\texttt{E-}}9 & 3.8\textscale{.8}{\texttt{E-}}10 \\
tsen{\tt\_}order{\tt\_}global{\tt\_}FE{\tt\_}15{\tt\_}standard & 1.12\textscale{.8}{\texttt{E+}}7 & 3.1\textscale{.8}{\texttt{E-}}6 & 2.0\textscale{.8}{\texttt{E-}}8 & 7.6\textscale{.8}{\texttt{E-}}9 & 1.7\textscale{.8}{\texttt{E-}}10 \\
tsen{\tt\_}order{\tt\_}global{\tt\_}FE{\tt\_}15{\tt\_}standard{\tt\_}prune & 2.05\textscale{.8}{\texttt{E+}}7 & 2.7\textscale{.8}{\texttt{E-}}6 & 1.6\textscale{.8}{\texttt{E-}}8 & 4.8\textscale{.8}{\texttt{E-}}9 & 1.2\textscale{.8}{\texttt{E-}}10 \\
tsen{\tt\_}order{\tt\_}global{\tt\_}RST{\tt\_}2{\tt\_}standard & 1.10\textscale{.8}{\texttt{E+}}8 & 5.6\textscale{.8}{\texttt{E-}}7 & 1.9\textscale{.8}{\texttt{E-}}8 & 5.0\textscale{.8}{\texttt{E-}}9 & 7.3\textscale{.8}{\texttt{E-}}10 \\
tsen{\tt\_}order{\tt\_}global{\tt\_}RST{\tt\_}2{\tt\_}standard{\tt\_}prune & 1.19\textscale{.8}{\texttt{E+}}8 & 4.0\textscale{.8}{\texttt{E-}}7 & 1.4\textscale{.8}{\texttt{E-}}8 & 4.2\textscale{.8}{\texttt{E-}}9 & 5.4\textscale{.8}{\texttt{E-}}10 \\
tsen{\tt\_}order{\tt\_}global{\tt\_}RST{\tt\_}3{\tt\_}standard & 1.06\textscale{.8}{\texttt{E+}}8 & 5.5\textscale{.8}{\texttt{E-}}7 & 1.9\textscale{.8}{\texttt{E-}}8 & 5.1\textscale{.8}{\texttt{E-}}9 & 7.2\textscale{.8}{\texttt{E-}}10 \\
tsen{\tt\_}order{\tt\_}global{\tt\_}RST{\tt\_}3{\tt\_}standard{\tt\_}prune & 1.79\textscale{.8}{\texttt{E+}}8 & 2.3\textscale{.8}{\texttt{E-}}7 & 1.6\textscale{.8}{\texttt{E-}}8 & 5.3\textscale{.8}{\texttt{E-}}9 & 1.3\textscale{.8}{\texttt{E-}}9 \\
tsen{\tt\_}timed{\tt\_}FE{\tt\_}8{\tt\_}standard & 2.27\textscale{.8}{\texttt{E+}}7 & 8.0\textscale{.8}{\texttt{E-}}5 & 1.5\textscale{.8}{\texttt{E-}}8 & 1.5\textscale{.8}{\texttt{E-}}9 & 1.2\textscale{.8}{\texttt{E-}}11 \\
tsen{\tt\_}timed{\tt\_}FE{\tt\_}8{\tt\_}resample & 1.92\textscale{.8}{\texttt{E+}}7 & 1.1\textscale{.8}{\texttt{E-}}4 & 1.8\textscale{.8}{\texttt{E-}}8 & 1.3\textscale{.8}{\texttt{E-}}9 & 8.8\textscale{.8}{\texttt{E-}}12 \\
tsen{\tt\_}timed{\tt\_}FE{\tt\_}15{\tt\_}standard & 1.07\textscale{.8}{\texttt{E+}}7 & 3.6\textscale{.8}{\texttt{E-}}4 & 1.6\textscale{.8}{\texttt{E-}}8 & 1.1\textscale{.8}{\texttt{E-}}9 & 3.2\textscale{.8}{\texttt{E-}}12 \\
tsen{\tt\_}timed{\tt\_}FE{\tt\_}15{\tt\_}resample & 8.36\textscale{.8}{\texttt{E+}}6 & 5.2\textscale{.8}{\texttt{E-}}4 & 1.7\textscale{.8}{\texttt{E-}}8 & 8.4\textscale{.8}{\texttt{E-}}10 & 1.5\textscale{.8}{\texttt{E-}}12 \\
tsen{\tt\_}timed{\tt\_}RST{\tt\_}2{\tt\_}standard & 1.61\textscale{.8}{\texttt{E+}}8 & 3.0\textscale{.8}{\texttt{E-}}6 & 1.8\textscale{.8}{\texttt{E-}}8 & 1.8\textscale{.8}{\texttt{E-}}9 & 1.4\textscale{.8}{\texttt{E-}}10 \\
tsen{\tt\_}timed{\tt\_}RST{\tt\_}2{\tt\_}resample & 9.35\textscale{.8}{\texttt{E+}}7 & 4.0\textscale{.8}{\texttt{E-}}6 & 1.6\textscale{.8}{\texttt{E-}}8 & 1.7\textscale{.8}{\texttt{E-}}9 & 7.0\textscale{.8}{\texttt{E-}}11 \\
tsen{\tt\_}timed{\tt\_}RST{\tt\_}3{\tt\_}standard & 1.60\textscale{.8}{\texttt{E+}}8 & 2.9\textscale{.8}{\texttt{E-}}6 & 1.8\textscale{.8}{\texttt{E-}}8 & 1.8\textscale{.8}{\texttt{E-}}9 & 1.4\textscale{.8}{\texttt{E-}}10 \\
tsen{\tt\_}timed{\tt\_}RST{\tt\_}3{\tt\_}resample & 1.21\textscale{.8}{\texttt{E+}}8 & 2.4\textscale{.8}{\texttt{E-}}6 & 2.0\textscale{.8}{\texttt{E-}}8 & 2.3\textscale{.8}{\texttt{E-}}9 & 1.7\textscale{.8}{\texttt{E-}}10 \\
tsen{\tt\_}timed{\tt\_}global{\tt\_}FE{\tt\_}8{\tt\_}standard & 2.80\textscale{.8}{\texttt{E+}}7 & 2.1\textscale{.8}{\texttt{E-}}4 & 1.7\textscale{.8}{\texttt{E-}}8 & 1.0\textscale{.8}{\texttt{E-}}9 & 7.4\textscale{.8}{\texttt{E-}}12 \\
tsen{\tt\_}timed{\tt\_}global{\tt\_}FE{\tt\_}8{\tt\_}resample & 1.34\textscale{.8}{\texttt{E+}}7 & 4.3\textscale{.8}{\texttt{E-}}4 & 1.6\textscale{.8}{\texttt{E-}}8 & 7.9\textscale{.8}{\texttt{E-}}10 & 2.2\textscale{.8}{\texttt{E-}}12 \\
tsen{\tt\_}timed{\tt\_}global{\tt\_}FE{\tt\_}8{\tt\_}standard{\tt\_}prune & 1.78\textscale{.8}{\texttt{E+}}8 & 2.1\textscale{.8}{\texttt{E-}}4 & 1.7\textscale{.8}{\texttt{E-}}8 & 4.1\textscale{.8}{\texttt{E-}}10 & 7.9\textscale{.8}{\texttt{E-}}12 \\
tsen{\tt\_}timed{\tt\_}global{\tt\_}FE{\tt\_}8{\tt\_}resample{\tt\_}prune & 4.05\textscale{.8}{\texttt{E+}}7 & 4.4\textscale{.8}{\texttt{E-}}4 & 1.7\textscale{.8}{\texttt{E-}}8 & 4.8\textscale{.8}{\texttt{E-}}10 & 2.4\textscale{.8}{\texttt{E-}}12 \\
tsen{\tt\_}timed{\tt\_}global{\tt\_}FE{\tt\_}15{\tt\_}standard & 1.28\textscale{.8}{\texttt{E+}}7 & 1.0\textscale{.8}{\texttt{E-}}3 & 1.6\textscale{.8}{\texttt{E-}}8 & 6.9\textscale{.8}{\texttt{E-}}10 & 1.6\textscale{.8}{\texttt{E-}}12 \\
tsen{\tt\_}timed{\tt\_}global{\tt\_}FE{\tt\_}15{\tt\_}resample & 5.57\textscale{.8}{\texttt{E+}}6 & 2.5\textscale{.8}{\texttt{E-}}3 & 1.7\textscale{.8}{\texttt{E-}}8 & 5.7\textscale{.8}{\texttt{E-}}10 & 4.7\textscale{.8}{\texttt{E-}}13 \\
tsen{\tt\_}timed{\tt\_}global{\tt\_}FE{\tt\_}15{\tt\_}standard{\tt\_}prune & 8.57\textscale{.8}{\texttt{E+}}7 & 1.0\textscale{.8}{\texttt{E-}}3 & 1.7\textscale{.8}{\texttt{E-}}8 & 2.9\textscale{.8}{\texttt{E-}}10 & 1.8\textscale{.8}{\texttt{E-}}12 \\
tsen{\tt\_}timed{\tt\_}global{\tt\_}FE{\tt\_}15{\tt\_}resample{\tt\_}prune & 1.84\textscale{.8}{\texttt{E+}}7 & 2.5\textscale{.8}{\texttt{E-}}3 & 1.7\textscale{.8}{\texttt{E-}}8 & 3.1\textscale{.8}{\texttt{E-}}10 & 4.5\textscale{.8}{\texttt{E-}}13 \\
tsen{\tt\_}timed{\tt\_}global{\tt\_}RST{\tt\_}2{\tt\_}standard & 2.08\textscale{.8}{\texttt{E+}}8 & 3.0\textscale{.8}{\texttt{E-}}6 & 1.8\textscale{.8}{\texttt{E-}}8 & 1.6\textscale{.8}{\texttt{E-}}9 & 1.4\textscale{.8}{\texttt{E-}}10 \\
tsen{\tt\_}timed{\tt\_}global{\tt\_}RST{\tt\_}2{\tt\_}resample & 1.14\textscale{.8}{\texttt{E+}}8 & 8.5\textscale{.8}{\texttt{E-}}6 & 1.7\textscale{.8}{\texttt{E-}}8 & 1.1\textscale{.8}{\texttt{E-}}9 & 3.6\textscale{.8}{\texttt{E-}}11 \\
tsen{\tt\_}timed{\tt\_}global{\tt\_}RST{\tt\_}2{\tt\_}standard{\tt\_}prune & 6.75\textscale{.8}{\texttt{E+}}8 & 2.8\textscale{.8}{\texttt{E-}}6 & 1.6\textscale{.8}{\texttt{E-}}8 & 8.4\textscale{.8}{\texttt{E-}}10 & 1.2\textscale{.8}{\texttt{E-}}10 \\
tsen{\tt\_}timed{\tt\_}global{\tt\_}RST{\tt\_}2{\tt\_}resample{\tt\_}prune & 1.36\textscale{.8}{\texttt{E+}}8 & 8.1\textscale{.8}{\texttt{E-}}6 & 1.6\textscale{.8}{\texttt{E-}}8 & 9.8\textscale{.8}{\texttt{E-}}10 & 3.4\textscale{.8}{\texttt{E-}}11 \\
tsen{\tt\_}timed{\tt\_}global{\tt\_}RST{\tt\_}3{\tt\_}standard & 2.04\textscale{.8}{\texttt{E+}}8 & 3.0\textscale{.8}{\texttt{E-}}6 & 1.8\textscale{.8}{\texttt{E-}}8 & 1.6\textscale{.8}{\texttt{E-}}9 & 1.4\textscale{.8}{\texttt{E-}}10 \\
tsen{\tt\_}timed{\tt\_}global{\tt\_}RST{\tt\_}3{\tt\_}resample & 1.64\textscale{.8}{\texttt{E+}}8 & 2.0\textscale{.8}{\texttt{E-}}6 & 1.6\textscale{.8}{\texttt{E-}}8 & 1.8\textscale{.8}{\texttt{E-}}9 & 1.3\textscale{.8}{\texttt{E-}}10 \\
tsen{\tt\_}timed{\tt\_}global{\tt\_}RST{\tt\_}3{\tt\_}standard{\tt\_}prune & 6.72\textscale{.8}{\texttt{E+}}8 & 2.8\textscale{.8}{\texttt{E-}}6 & 1.6\textscale{.8}{\texttt{E-}}8 & 8.4\textscale{.8}{\texttt{E-}}10 & 1.2\textscale{.8}{\texttt{E-}}10 \\
tsen{\tt\_}timed{\tt\_}global{\tt\_}RST{\tt\_}3{\tt\_}resample{\tt\_}prune & 2.21\textscale{.8}{\texttt{E+}}8 & 4.2\textscale{.8}{\texttt{E-}}6 & 1.7\textscale{.8}{\texttt{E-}}8 & 1.1\textscale{.8}{\texttt{E-}}9 & 7.1\textscale{.8}{\texttt{E-}}11 \\
\end{longtable}

\endgroup
}{}

\end{document}